\documentclass[a4paper,11pt]{article}
\usepackage[a4paper,margin=1in]{geometry}
\usepackage[utf8]{inputenc}
\usepackage[T1]{fontenc}
\usepackage[dvipsnames]{xcolor}
\usepackage{graphicx}
\usepackage{amsmath, amssymb, amstext, amsfonts, mathtools}
\usepackage{amsthm}
\usepackage{bm}
\usepackage{thmtools}
\usepackage{thm-restate}

\usepackage{soul}
\usepackage{enumitem}

\usepackage{authblk}

\usepackage[pdftex,pdfpagelabels,hyperfigures,pagebackref,bookmarks=false,hyperindex=true]{hyperref}

\hypersetup{
 plainpages = false,
 bookmarksopen = true,
 bookmarksnumbered = true,
 breaklinks = true,
 linktocpage,
 colorlinks = true,
 linkcolor = blue!70!black,
 urlcolor  = blue!70!black,
 citecolor = red!70!black,
 anchorcolor = green!70!black,
 pdflang={en},
 pdftitle={Robust Contraction Decomposition for Minor-Free Graphs and its Applications},
 pdfauthor={Sayan Bandyapadhyay, William Lochet, Daniel Lokshtanov, D{\'{a}}niel Marx, Pranabendu Misra, Daniel Neuen, Saket Saurabh, Prafullkumar Tale, Jie Xue}
}

\newtheorem{lemma}{Lemma}
\newtheorem{lemma*}[lemma]{Lemma*}
\numberwithin{lemma}{section}
\newtheorem{theorem}[lemma]{Theorem}
\newtheorem{theorem*}[lemma]{Theorem*}
\newtheorem{corollary}[lemma]{Corollary}
\newtheorem{definition}[lemma]{Definition}
\newtheorem{observation}[lemma]{Observation}

\newtheorem{fact}[lemma]{Fact}

\newtheorem{claim}[theorem]{Claim}

\newenvironment{claimproof}{\begin{proof}}{\end{proof}}

\newcommand{\tw}{\mathbf{tw}}
\newcommand{\tor}{\mathsf{tor}}

\title{Robust Contraction Decomposition\\for Minor-Free Graphs and its Applications}

\author[1]{Sayan Bandyapadhyay}
\author[2]{William Lochet}
\author[3]{Daniel Lokshtanov}
\author[4]{D{\'{a}}niel Marx}
\author[5]{Pranabendu Misra}
\author[6]{Daniel Neuen}
\author[7]{Saket Saurabh}
\author[8]{Prafullkumar Tale}
\author[9]{Jie Xue}

\affil[1]{Portland State University, USA}
\affil[2]{LIRMM, Université de Montpellier, CNRS, France}
\affil[3]{University of California, USA}
\affil[4]{CISPA Helmholtz Center for Information Security, Germany}
\affil[5]{Chennai Mathematical Institute, India}
\affil[6]{Max Planck Institute for Informatics, Germany}
\affil[7]{Institute of Mathematical Sciences, India}
\affil[8]{Indian Institute of Science Education and Research Bhopal, India}
\affil[9]{New York University Shanghai, China}

\begin{document}

\maketitle

\begin{abstract}
 We prove a robust contraction decomposition theorem for $H$-minor-free graphs, which states that given an $H$-minor-free graph $G$ and an integer $p$, one can partition in polynomial time the vertices of $G$ into $p$ sets $Z_1,\dots,Z_p$ such that $\tw(G/(Z_i \setminus Z')) = O(p + |Z'|)$ for all $i \in [p]$ and $Z' \subseteq Z_i$.
 Here, $\tw(\cdot)$ denotes the treewidth of a graph and $G/(Z_i \setminus Z')$ denotes the graph obtained from $G$ by contracting all edges with both endpoints in $Z_i \setminus Z'$.
    
 Our result generalizes earlier results by Klein~[SICOMP 2008] and Demaine~et~al.~[STOC 2011] based on partitioning $E(G)$, and some recent theorems for planar graphs by Marx et al.\ [SODA 2022], for bounded-genus graphs (more generally, almost-embeddable graphs) by Bandyapadhyay et al.\ [SODA 2022], and for unit-disk graphs by Bandyapadhyay et al.\ [SoCG 2022].

 The robust contraction decomposition theorem directly results in parameterized algorithms with running time $2^{\widetilde{O}(\sqrt{k})} \cdot n^{O(1)}$ or $n^{O(\sqrt{k})}$ for every vertex/edge deletion problems on $H$-minor-free graphs that can be formulated as \textsc{Permutation CSP Deletion} or \textsc{2-Conn Permutation CSP Deletion}.
 Consequently, we obtain the first subexponential-time parameterized algorithms for \textsc{Subset Feedback Vertex Set}, \textsc{Subset Odd Cycle Transversal}, \textsc{Subset Group Feedback Vertex Set}, \textsc{2-Conn Component Order Connectivity} on $H$-minor-free graphs.
 For other problems which already have subexponential-time parameterized algorithms on $H$-minor-free graphs (e.g., \textsc{Odd Cycle Transversal}, \textsc{Vertex Multiway Cut}, \textsc{Vertex Multicut}, etc.), our theorem gives much simpler algorithms of the same running time.
\end{abstract}

\section{Introduction}

Baker's layering technique is a fundamental tool for a certain type of algorithms on planar graphs.
It was first used implicitly by Baker~\cite{Baker94} and can be formalized as the following statement: given a planar graph $G$ and an integer $p$, we can find in polynomial-time a partitioning $Z_1,\dots,Z_p$ of the vertex set of $G$ such that $G-Z_i$ has treewidth $O(p)$ for every $i \in [p]$.
This decomposition can be used for problems where we can argue that there exist an $i \in [p]$ such that $Z_i$ is irrelevant to the solution (or has negligible contribution to it) and hence $Z_i$ can be deleted from $G$.
Then we are left with an instance $G-Z_i$ having treewidth $O(p)$ and techniques for bounded-treewidth graphs can be used.
This approach has been used in the design of polynomial-time approximation schemes (PTAS) \cite{Baker94,DemaineHK05,KhannaM96} and parameterized algorithms \cite{BuiP92,DornFLRS13,Eppstein99,FominLMPPS16,Grohe03,Tazari12} for planar graphs.
Moreover, this decomposition algorithm can be further generalized to $H$-minor free graphs \cite{DemaineHK05}.

\paragraph{Contraction Decomposition Theorems.}
There are many natural problems where a set $Z_i$ cannot be removed even if it is irrelevant to the solution: most notably, problems involving connectivity typically have this property.
To handle such problems, Klein~\cite{Klein08} introduced a dual version of Baker's layering technique, which is based on contractions of edges.
This \emph{(Edge) Contraction Decomposition Theorem} was later generalized to $H$-minor free graphs by Demaine, Hajiaghayi, and Kawarabayashi~\cite{DemaineHK11}.

\begin{theorem}[Demaine, Hajiaghayi, and Kawarabayashi~\cite{DemaineHK11}]
 \label{thm-edge-contraction}
 Let $H$ be a fixed graph.
 Given an $H$-minor-free graph $G$ and an integer $p$, one can compute in polynomial time a partition of $E(G)$ into $p$ sets $Z_1,\dots,Z_p$ such that $\tw(G/Z_i) = O(p)$ for all $i \in [p]$, where the constant hidden in $O(\cdot)$ depends on $H$.
\end{theorem}

Contraction decomposition theorems of this form can be used as a building block in designing approximation schemes for, e.g., the {\sc Traveling Salesman Problem (TSP)} and {\sc Subset TSP}~\cite{BuiP92,DemaineHK11,DemaineHM10,Klein06,Klein08}, and parameterized algorithms for {\sc Bisection}, {\sc $k$-Way Cut}, {\sc Odd Cycle Transversal}, {\sc Subset Feedback Vertex Set} and many other problems \cite{BandyapadhyayLLSX24,KawarabayashiT11,MarxMNT22}.

\paragraph{Subexponential Parameterized Algorithms (Edge Problems).}

To illustrate this, let us briefly discuss how Theorem~\ref{thm-edge-contraction} can be used to obtain a subexponential-time parameterized algorithm for {\sc Edge Multiway Cut} in $H$-minor-free graphs.
In {\sc Edge Multiway Cut}, we are given a graph $G$, a set $T\subseteq V(G)$ of terminals, and an integer $k$, and the task is to find a set $S\subseteq E(G)$ of at most $k$ edges such that every component of $G \setminus S$ contains at most one terminal.

\begin{enumerate}[label = (\arabic*)]
 \item A result of Wahlström~\cite{Wahlstrom20} gives a randomized quasipolynomial kernel for the problem, which allows us to assume that $|V(G)|=k^{\textup{polylog}(k)}$.
 \item Given the decomposition $Z_1,\dots,Z_p$ of Theorem~\ref{thm-edge-contraction} where $p \coloneqq \lceil\sqrt{k}\rceil$, there is an $i \in [p]$ such that $|Z_i \cap S| \leq \sqrt{k}$ for the hypothetical solution $S$ of size $k$.
 \item Guess this $i \in [p]$ (there are $\sqrt{k}$ possibilities) and the set $Z_i \cap S$ (there are $|V(G)|^{O(\sqrt{k})}=k^{\sqrt{k}\cdot \textup{polylog}(k)}$ possibilities)
 \item Contracting $Z_i \setminus S$ does not change the problem at hand, since the solution $S$ if not affected.
 \item\label{i:rob} Since the graph $G/Z_i$ has treewidth $O(\sqrt{k})$, we get that the graph $G/(Z_i\setminus S)$ has treewidth $O(\sqrt{k}+|Z_i\cap S|)=O(\sqrt{k})$ (contracting an edge can decrease treewidth by at most one).
  Finally, we solve {\sc Edge Multiway Cut} on $G/(Z_i\setminus S)$ in time $2^{O(\sqrt{k}\log k)}\cdot |V(G)|^{O(1)}$ using standard bounded-treewidth techniques.
\end{enumerate}
Overall, we obtain a $2^{\sqrt{k}\cdot \textup{polylog}(k)}\cdot n^{O(1)}$ time randomized algorithm for {\sc Edge Multiway Cut}.
The same approach also works for many other problems \cite{BandyapadhyayLLSJ22,MarxMNT22}.

\paragraph{Subexponential Parameterized Algorithms (Vertex Problems).}
Let us try to apply a similar technique for the problem {\sc Vertex Multiway Cut}, where instead of deleting a set of edges, we need to delete now a set of vertices.
There are two main difficulties if we try to apply Theorem~\ref{thm-edge-contraction}.
First, it would be convenient to have a partition $Z_1,\dots,Z_p$ of \emph{vertices} such that for every $i \in [p]$ contracting $Z_i$ leaves a graph of treewidth $O(p)$.
Here contracting a vertex set $Z_i$ amounts to contracting all edges that have both endpoints in  the set $Z_i$.
Second, in Step~\ref{i:rob} above, we exploited that the decomposition of Theorem~\ref{thm-edge-contraction} is \emph{inherently robust}: if $Z_i'$ is obtained from $Z_i$ by omitting a set of $s$ edges, then the treewidth of $G/Z'_i$ is at most $s$ higher than the treewidth of $G/Z_i$.

The following example shows that an analogous statements for contracting vertex sets is not true.
Let $G$ be obtained from a $k\times k$ grid by adding two universal vertices $x$ and $y$.
Let $A$ and $B$ be the two bipartition classes of the $k\times k$ grid, and define $Z_1 \coloneqq A \cup \{x\}$ and $Z_2 \coloneqq B \cup \{y\}$.
Then $G/Z_i$ has treewidth 2 for both $i \in [2]$.
On the other hand, $G/(Z_1\setminus \{x\}) = G/A = G$ and $G/(Z_2\setminus \{y\}) = G/B = G$ which means the jump in treewidth can be unbounded even after omitting just a single vertex.
Note that $G$ is $K_7$-minor free, so we cannot hope that a vertex partition version of Theorem~\ref{thm-edge-contraction} is inherently robust even on $H$-minor-free graphs.

\paragraph{Robust Contraction Decomposition Theorem.}
The above naturally leads to the question of a \emph{robust} vertex contraction decomposition theorem, where omitting vertices from some $Z_i$ increases the treewidth only in a bounded way.
Such decomposition theorems were recently introduced independently for planar graphs \cite{MarxMNT22} and for bounded-genus graphs (and more generally, almost-embeddable graphs) \cite{BandyapadhyayLLSJ22} by subsets of the authors,
who used them to obtain the first subexponential-time (parameterized) algorithms for a number of problems, such as {\sc Edge Bipartization}, {\sc Odd Cycle Transversal}, {\sc Edge/Vertex Multiway Cut}, {\sc Edge/Vertex Multicut}, and {\sc Group Feedback Edge/Vertex Set} on the above graph classes.
More precisely, \cite{BandyapadhyayLLSJ22,MarxMNT22} design parameterized algorithms with running time $f(k)n^{O(\sqrt{k})}$ or even $2^{O(\sqrt{k})}n^{O(1)}$, where $k$ is the size of the solution (and $\widetilde{O}(\cdot)$ hides $\log k$ factors).

The main result of this paper is a \emph{vertex} contraction decomposition theorem on $H$-minor-free graphs that is \emph{robust} in the sense described above.

\begin{restatable}{theorem}{main}
 \label{thm-contraction}
 Let $H$ be a fixed graph.
 Given an $H$-minor-free graph $G$ and an integer $p$, one can compute in polynomial time a partition of $V(G)$ into $p$ sets $Z_1,\dots,Z_p$ such that $\tw(G/(Z_i \setminus Z')) = O(p + |Z'|)$ for all $i \in [p]$ and $Z' \subseteq Z_i$, where the constant hidden in $O(\cdot)$ depends on $H$.
\end{restatable}

Theorem~\ref{thm-contraction} generalizes the robust vertex contraction theorem for planar graphs \cite{MarxMNT22} and almost-embeddable graphs \cite{BandyapadhyayLLSJ22} to arbitrary $H$-minor-free graphs.
It also generalizes the edge contraction decomposition (Theorem~~\ref{thm-edge-contraction}) for $H$-minor free graphs by Demaine et al.~\cite{DemaineHK11}.

Combined with the results from~\cite{MarxMNT22}, Theorem~\ref{thm-contraction} directly yields parameterized algorithms of running time $2^{\widetilde{O}(\sqrt{k})} \cdot n^{O(1)}$ or $n^{O(\sqrt{k})}$ for all vertex/edge deletion problems on $H$-minor-free graphs that can be formulated as {\sc Permutation CSP Deletion} or {\sc 2-Conn Permutation CSP Deletion}, two broad classes of problems defined in \cite{MarxMNT22}.
Roughly speaking, a permutation CSP instance is a binary CSP instance satisfying certain nice properties, and thus corresponds to a graph in which the vertices are variables and the edges are constraints.
The {\sc Permutation CSP Edge/Vertex Deletion} problem basically asks whether one can delete at most $k$ vertices (resp., edges) from (the graph of) a permutation CSP instance such that every connected component in the resulting graph has a satisfying assignment.
The {\sc 2-Conn Permutation CSP Edge/Vertex Deletion} problem is defined similarly, but we only care about the satisfiability on every 2-connected component (instead of connected component).
We refer to Section~\ref{sec-app} for the precise definition of these problems.
Combining the results from~\cite{MarxMNT22} and Theorem~\ref{thm-contraction}, we get the following result.

\begin{theorem}\label{thm-app}
 Let $H$ be a fixed graph.
 Then {\sc Permutation CSP Edge/Vertex Deletion} and {\sc 2-Conn Permutation CSP Edge/Vertex Deletion} on $H$-minor-free instances $(\varGamma,w,\delta,k)$ can be solved in $(|\varGamma|+\delta)^{O(\sqrt{k})}$ time.
\end{theorem}

The above theorem directly gives us parameterized algorithms of running time exponential in $O(\sqrt{k})$ for many problems on $H$-minor-free graphs.
Specifically, we obtain the first subexponential-time parameterized algorithms for the following fundamental problems on $H$-minor-free graphs.
\begin{itemize}
 \item $n^{O(\sqrt{k})}$ time algorithms for \textsc{Subset Feedback Vertex Set}, \textsc{Subset Odd Cycle Transversal}, \textsc{Subset Group Feedback Vertex Set}, \textsc{2-Conn Component Order Connectivity}, and the edge-deletion version of these problems.
 \item $2^{\widetilde{O}(\sqrt{k})} \cdot n^{O(1)}$ time algorithms for \textsc{Subset Feedback Vertex Set} and \textsc{Subset Feedback Edge Set}.
\end{itemize}

Additionally, the main algorithmic results from~\cite{BandyapadhyayLLSJ22} (such as subexponential-time parameterized algorithms for \textsc{Odd Cycle Transversal}, \textsc{Vertex Multiway Cut}, \textsc{Vertex Multicut}, etc.),
that required a complicated two-layered dynamic programming algorithm over the structural decomposition theorem of Robertson and Seymour~\cite{RobertsonS03a} for $H$-minor free graphs,
now admit much cleaner and more direct algorithms by exploiting Theorem~\ref{thm-contraction} (these problems belong to the \textsc{Permutation CSP Deletion} category in Theorem~\ref{thm-app}).

\paragraph{Other Relevant Literature.}
The study of subexponential-time parameterized algorithms on planar and $H$-minor free graphs has been one of the most active sub-areas of parameterized algorithms, which led to exciting results and powerful methods.
Examples include Bidimensionality~\cite{DemaineFHT05}, applications of Baker's layering technique~\cite{BuiP92,DornFLRS13,Eppstein99,FominLMPPS16,Tazari12}, bounds on the treewidth of the solution~\cite{FominLKPS20,KleinM12,KleinM14,MarxPP18}, and pattern coverage~\cite{FominLMPPS16,Nederlof20a}.
The central theme of all these results is that planar graphs exhibit the ``\emph{ square root phenomenon}'': parameterized problems whose fastest parameterized algorithm run in time $f(k)n^{O(k)}$ or $2^{O(k)}$ on general graphs admit $f(k)n^{O(\sqrt{k})}$ or even $2^{O(\sqrt{k})}n^{O(1)}$ time algorithms when input is restricted to planar or $H$-minor free graphs.
However, these techniques do not apply for designing subexponential time parameterized algorithms for cut and cycle hitting problems such as {\sc Multiway Cut}, {\sc Multicut}, or {\sc Odd Cycle Transversal}.
Robust vertex contraction decomposition theorems, first obtained in \cite{BandyapadhyayLLSJ22,MarxMNT22}, provide the necessary tools to design subexponential-time parameterized algorithms for some of the cut and cycle hitting problems on $H$-minor free graphs.

On the decomposition front, apart from edge contraction decomposition theorems, there are several other kind of decomposition theorems that have been used to design polynomial-time approximation schemes (PTASs) and FPT algorithms on planar graphs or more generally on $H$-minor free graphs.
Most of these decomposition theorems prove strengthening of the classic Baker's layering technique~\cite{Baker94}.
This generally yields, in what is known as (Vertex) Edge Decomposition Theorems~\cite{Baker94,DemaineHK05,DeVosDOSRSV04,Dvorak18,Eppstein00} (see~\cite{Panolan0Z19} for a detailed introduction).

Finally, let us also point to some recent developments on the structural theory of $H$-minor-free graphs \cite{KawarabayashiTW20,ThilikosW23}.
Our current proof of Theorem \ref{thm-contraction} mostly relies on the original Robertson-Seymour decomposition for $H$-minor-free graphs \cite{RobertsonS03a}, and is independent of \cite{KawarabayashiTW20,ThilikosW23}.
Reformulating some of our arguments in the language of \cite{KawarabayashiTW20,ThilikosW23} may allow us to obtain a more streamlined proof in the future.

\paragraph{Structure of Paper.}
The remainder of the paper is structured as follows.
In Section \ref{sec-overview} we give an informal overview of the proof of Theorem~\ref{thm-contraction}.
After giving additional preliminaries in Section \ref{sec-pre}, the formal proof of Theorem~\ref{thm-contraction} is given in Section~\ref{sec-proof}.
Finally, we discuss the algorithmic applications in Section \ref{sec-app}.

\subsection{Overview of Proof of Theorem~\ref{thm-contraction}}
\label{sec-overview}

In this section, we give an informal overview of our proof of Theorem~\ref{thm-contraction}.
For convenience, we say $Z_1,\dots,Z_p$ is a \emph{robust contraction decomposition (RCD)} of a graph $G$ if $Z_1,\dots,Z_p$ is a partition of $V(G)$ satisfying $\tw(G/(Z_i \setminus Z')) = O(p+|Z'|)$ for all $i \in [p]$ and $Z' \subseteq Z_i$.
A $p$-RCD simply refers to an RCD of size $p$.
Our goal is to compute a $p$-RCD of an $H$-minor-free graph $G$ for a given integer $p \geq 1$.

Our starting point is the well-known Robertson-Seymour decomposition for $H$-minor-free graphs.
The Robertson-Seymour decomposition of an $H$-minor-free graph $G$ is a tree decomposition $(T,\beta)$ of $G$ in which the \emph{torso} of each node $t \in V(T)$ is $h$-\emph{almost-embeddable} and the \emph{adhesion} of each node $t \in V(T)$ is of size at most $h$, for some constant $h$ depending on $H$.
Here, the adhesion of $t$, denoted by $\sigma(t)$, is the intersection of the bag $\beta(t)$ and the bag of the parent of $t$.
The torso of $t$, denoted by $\tor(t)$, is a supergraph of $G[\beta(t)]$ (the subgraph of $G$ induced by $\beta(t)$), obtained by adding edges to $G[\beta(t)]$ to make the adhesion of each child of $t$ a clique.
Almost-embeddable graphs generalize bounded-genus graphs.
Roughly speaking, an $h$-almost-embeddable graph has its main part embedded in a surface of genus at most $h$, and the remaining part consists of at most $h$ well-structured subgraphs (called \emph{vortices}) and a set of at most $h$ ``bad'' vertices (called \emph{apex vertices} or \emph{apices}).
In this overview, the reader can feel free to ignore the vortices/apices of an almost-embeddable graph and think about it as a graph embedded in a surface of bounded genus.
The formal definitions of tree decompositions and almost-embeddable graphs can be found in Section~\ref{sec-pre}.

Intuitively, the Robertson-Seymour decomposition partitions an $H$-minor-free graph into almost-embeddable ``pieces''.
Since it is known that almost-embeddable graphs admit RCDs \cite{BandyapadhyayLLSJ22}, a natural idea comes:
Can we exploit Robertson-Seymour decomposition to somehow ``lift'' the RCDs of the almost-embeddable pieces to the $H$-minor-free graph?
To explore this idea, let us consider an $H$-minor-free graph $G$ and the Robertson-Seymour decomposition $(T,\beta)$ of $G$.
Since the torsos of $(T,\beta)$ are $h$-almost-embeddable, we can compute a $p$-RCD $Z_1^{(t)},\dots,Z_p^{(t)} \subseteq \beta(t)$ for each torso $\tor(t)$ using the results of \cite{BandyapadhyayLLSJ22}.
Ideally, if these RCDs are \emph{consistent} in the sense that for every vertex $v \in V(G)$ there exists an index $i \in [p]$ such that $v \in Z_i^{(t)}$ for all nodes $t \in V(T)$ whose bag contains $v$, then we can combine them to obtain a partition $Z_1,\dots,Z_p$ of $V(G)$ where $Z_i = \bigcup_{t \in V(T)} Z_i^{(t)}$ and hope it to be an RCD of $G$.
(Of course, as we will see later, such a na{\"i}ve construction of $Z_1,\dots,Z_p$ does not work. But it can provide us some useful intuition towards a proof of the theorem.)
Obtaining consistent RCDs for the torsos is actually not difficult, by properly using the small adhesion size of $(T,\beta)$.
The basic idea is the following.
For each node $t \in V(T)$, instead of computing a $p$-RCD for $\tor(t)$, we only compute a $p$-RCD for $\tor(t) - \sigma(t)$, and how the vertices in $\sigma(t)$ belong to the $p$ classes is determined by the RCDs on the ancestors of $t$.
The decompositions constructed in this way are consistent.
Furthermore, since $|\sigma(t)| \leq h$, given a $p$-RCD for $\tor(t) - \sigma(t)$, no matter how we assign the vertices in $\sigma(t)$ to the $p$ classes, the resulting decomposition is always a $p$-RCD for $\tor(t)$.
Thus, we obtain consistent RCDs for the torsos, which give us the aforementioned partition $Z_1,\dots,Z_p$ of $V(G)$.
The remaining question is then whether $Z_1,\dots,Z_p$ is an RCD of $G$, i.e., whether $\tw(G/(Z_i \setminus Z')) = O(p+|Z'|)$ for all $i \in [p]$ and $Z' \subseteq Z_i$.

The only reason for why $Z_1,\dots,Z_p$ could be an RCD of $G$ is clearly the fact that every $Z_i$ satisfies the RCD condition ``locally'', i.e., when restricted to a torso.
Specifically, let $i \in [p]$ and $Z' \subseteq Z_i$.
Set $Z = Z_i \setminus Z'$ for convenience.
What we want is $\tw(G/Z) = O(p+|Z'|)$.
For each $t \in V(T)$, we have $Z \cap \beta(t) = Z_i^{(t)} \setminus (Z' \cap Z_i^{(t)})$, and thus $\tw(\tor(t)/(Z \cap \beta(t))) = O(p+|Z'|)$.
How can we go from the locally bounded treewidth of each $\tor(t)/(Z \cap \beta(t))$ to the global treewidth of $G/Z$?
The following observation (given in Lemma~\ref{lem-torsotw}) seems useful.
\begin{itemize}
    \item[] If a graph admits a tree decomposition $\mathcal{T}$ in which each torso is of treewidth at most $w$, then the graph itself is of treewidth $O(w)$.
    Indeed, we can ``glue'' the width-$w$ tree decompositions of the torsos along the given tree decomposition $\mathcal{T}$ to obtain a width-$O(w)$ tree decomposition of the graph (mainly because in a torso the adhesions of the children are cliques).
\end{itemize}
This observation allows us to go from the treewidth of the torsos to the treewidth of the entire graph.
At the first glance, it does not directly apply to our situation here, because we are working on the \emph{contracted} graphs instead of the original ones: our goal is to lift the treewidth bound from the \emph{contracted} torsos $\tor(t)/(Z \cap \beta(t))$ to the \emph{contracted} graph $G/Z$.
However, it is a well-known fact (given in Fact~\ref{fact-inducedtd}) that a tree decomposition of $G$ naturally induces a tree decomposition of the contracted graph $G/Z$ by ``contracting'' each bag.
Specifically, consider the \emph{quotient} map $\pi\colon V(G) \to V(G/Z)$ which maps each vertex of $G$ to the corresponding vertex of $G/Z$.
Set $\beta^*(t) = \pi(\beta(t))$ for all $t \in V(T)$.
Then $(T,\beta^*)$ is a tree decomposition of $G/Z$.
Intuitively, the tree decomposition $(T,\beta^*)$ should allow us to apply the above observation to lift the treewidth bound from $\tor(t)/(Z \cap \beta(t))$ to $G/Z$, and eventually show that $Z_1,\dots,Z_p$ is an RCD of $G$.

Although the above argument seems promising, a closer inspection reveals that it actually has a fatal issue, which is also the main barrier to proving Theorem \ref{thm-contraction}.
The issue comes from a somewhat counter-intuitive fact: \emph{the contracted torsos are not identical to the torsos of the contracted graph}!
Formally, let $\tor^*(t)$ denote the torso of $t \in V(T)$ in the tree decomposition $(T,\beta^*)$ of $G/Z$.
Then we have $\tor(t)/(Z \cap \beta(t)) \neq \tor^*(t)$ in general, and even worse, the treewidth of $\tor^*(t)$ can be unbounded even if the treewidth of $\tor(t)/(Z \cap \beta(t))$ is bounded.
The main reason for why this happens is the following.
The edges of $\tor(t)$ do not necessarily appear in $G$: some of them are manually added to make the adhesions of the children of $t$ cliques (for convenience, we call them \emph{fake edges}).
When we contract $G$ to $G/Z$, only the edges appearing in $G$ can get contracted.
However, when we contract $\tor(t)$ to $\tor(t)/(Z \cap \beta(t))$, all fake edges in $\tor(t)[Z \cap \beta(t)]$ also get contracted.
This over-contracting can make $\tor(t)/(Z \cap \beta(t))$ much ``smaller'' than $\tor^*(t)$.
To see an example, suppose $G$ is the graph obtained from an $m \times m$ grid by subdividing each edge into two edges (with an intermediate vertex); see Figure~\ref{fig-grid}.
So we have $n = O(m^2)$.
Consider a tree decomposition $(T,\beta)$ of $G$ defined as follows.
The tree $T$ consists of a root with some children.
The bag of the root contains the $m^2$ grid vertices.
Each child of the root corresponds to an edge $e$ of the grid, whose bag contains the two endpoints of $e$ (two grid vertices) and the intermediate vertex for subdividing $e$.
Now the adhesion of each child of the root consists of the two endpoints of the corresponding grid edge.
Let $t \in V(T)$ be the root and $Z \subseteq V(G)$ consist of the grid vertices.
Then $\tor(t)$ is an $m \times m$ grid and $\tor(t)/(Z \cap \beta(t)) = \tor(t)/Z$ is a single vertex.
However, $Z$ is actually an independent set in $G$.
Thus, $G/Z = G$ and we have $\tor^*(t) = \tor(t)$, which is of treewidth $m$.
In fact, this simple example not only demonstrates that $\tor^*(t)$ can have much higher treewidth than $\tor(t)/(Z \cap \beta(t))$, but also directly shows that a partition $Z_1,\dots,Z_p$ satisfying the RCD condition locally (i.e., at each torso) may be not a global RCD in general (and thus our previous construction fails).
Suppose now we have an RCD $Z_1^{(t)},\dots,Z_p^{(t)}$ of $\tor(t)$ for the root $t \in V(T)$.
For each child $s$ of $t$, we arbitrarily assign the only vertex in $\beta(s) \setminus \sigma(s)$ to a class $Z_i$ different from the classes the two vertices in $\sigma(s)$ belong to.
Note that, since $|\beta(s)| = 3$, every partition of $\beta(s)$ is actually an RCD of $\tor(s)$.
In this way, we obtain a partition $Z_1,\dots,Z_p$ of $V(G)$ that is an RCD when restricted to every torso.
However, each $Z_i$ is an independent set of $G$, and thus $G/Z_i = G$.
So unless $p = \Omega(n)$, $Z_1,\dots,Z_p$ is not an RCD of $G$.

\begin{figure}
 \centering
 \includegraphics[height=3.6cm]{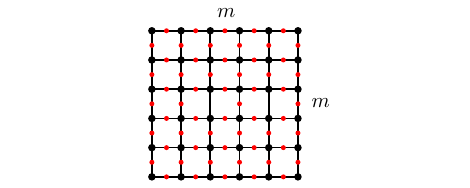}
 \caption{An $m \times m$ grid with subdivided edges. The black vertices are the grid vertices.}
 \label{fig-grid}
\end{figure}

The main technical contribution in our proof is to break the above barrier by constructing $Z_1,\dots,Z_p$ in a much more sophisticated way.
On a high-level, we still follow the idea of constructing RCDs locally for each torso and combine them to obtain $Z_1,\dots,Z_p$.
So in what follows, we always use $Z_1^{(t)},\dots,Z_p^{(t)}$ to denote the restriction of $Z_1,\dots,Z_p$ to $\tor(t)$, i.e., $Z_i^{(t)} \coloneqq Z_i \cap \beta(t)$.
To avoid ambiguity, for a subset $Z \subseteq V(G)$, we use $(T,\beta_Z^*)$ to denote the tree decomposition of $G/Z$ induced by $(T,\beta)$, and use $\tor_Z^*(t)$ to denote the torso of a node $t \in V(T)$ in $(T,\beta_Z^*)$.
We have seen that the main reason for why $\tor_Z^*(t)$ can have a much higher treewidth than $\tor(t)/(Z \cap \beta(t))$ is the fake edges in $\tor(t)$.
So ideally, if $\tor(t)/(Z \cap \beta(t))$ does not contract any fake edge in $\tor(t)$, we are happy.
But this is unlikely, because in the worst case (e.g., the grid example above) every edge in $\tor(t)$ is fake.
Fortunately, the fake edges in $\tor(t)$ are not always ``uncontractable''.
Indeed, if a fake edge in $\tor(t)$ has two endpoints lying in the same connected component of $G[Z]$, then we can safely contract it because its two endpoints are also contracted to one vertex in $G/Z$.
Therefore, it is enough to guarantee that $\tor(t)/(Z \cap \beta(t))$ only contracts these safe fake edges, which is equivalent to saying that the two endpoints of every edge of $\tor(t)[Z \cap \beta(t)]$ lie in the same connected component of $G[Z]$.
(If this holds, we should be able to somehow relate $\tw(\tor_Z^*(t))$ to $\tw(\tor(t)/(Z \cap \beta(t)))$ and make the previous proof work.)
However, this is still too difficult, because we are not dealing with a single set $Z$.
Instead, we have to take care of $Z = Z_i \setminus Z'$ for all $i \in [p]$ and $Z' \subseteq Z_i$, i.e., all subsets of $Z_1,\dots,Z_p$.
Even if the construction of $Z_i$ makes every fake edge $uv$ of $\tor(t)[Z_i \cap \beta(t)]$ have its endpoints $u$ and $v$ in the same connected component of $G[Z_i]$, when a fraction $Z' \subseteq Z_i$ is taken away from $Z_i$, it can happen that $u$ and $v$ lie in different connected components of $G[Z_i \setminus Z']$ (for example, when $Z'$ is a cut of $u$ and $v$ in $G[Z_i]$) and thus the fake edge $uv$ is no longer safe for contraction when considering $Z = Z_i \setminus Z'$.

The key observation to get rid of this issue is that we do not necessarily need to make \emph{every} fake edge of $\tor(t)[(Z_i \setminus Z') \cap \beta(t)] = \tor(t)[Z_i^{(t)} \setminus Z']$ safe.
Indeed, the desired bound for $\tw(\tor_{Z_i \setminus Z'}^*(t))$ and $\tw(G/(Z_i \setminus Z'))$ is $O(p+|Z'|)$, which also depends on $|Z'|$.
It turns out that as long as $\tor(t)[Z_i^{(t)} \setminus Z']$ only contains $O(|Z'|)$ ``unsafe'' fake edges (i.e., those with two endpoints in different connected components of $G[Z_i \setminus Z']$), we can obtain the bound $\tw(\tor_{Z_i \setminus Z'}^*(t)) = O(p+|Z'|)$.
To see this, let $Z'' \subseteq (Z_i \setminus Z') \cap \beta(t) = Z_i^{(t)} \setminus Z'$ consist of the endpoints of the $O(|Z'|)$ unsafe fake edges in $\tor(t)[Z_i^{(t)} \setminus Z']$, so we have $|Z''| = O(|Z'|)$.
Now every edge in $\tor(t)[Z_i^{(t)} \setminus (Z' \cup Z'')]$ has its two endpoints in the same connected component of $G[Z_i \setminus Z']$, and is thus safe.
Since $\tor(t)/(Z_i^{(t)} \setminus (Z' \cup Z''))$ only contracts safe edges, it can be shown that $\tor_{Z_i \setminus Z'}^*(t)$ is (almost\footnote{More precisely, $\tor_{Z_i \setminus Z'}^*(t)$ is a minor of a graph obtained from $\tor(t)/(Z_i^{(t)} \setminus (Z' \cup Z''))$ by adding few edges.}) a minor of $\tor(t)/(Z_i^{(t)} \setminus (Z' \cup Z''))$.
Thus, we can bound $\tw(\tor_{Z_i \setminus Z'}^*(t))$ using $\tw(\tor(t)/(Z_i^{(t)} \setminus (Z' \cup Z''))$, where the latter is $O(p+|Z' \cup Z''|) = O(p+|Z'|)$, under the assumption that $Z_1^{(t)},\dots,Z_p^{(t)}$ is an RCD of $\tor(t)$.
To summarize, $Z_1,\dots,Z_p$ is an RCD of $G$ if the following conditions hold.
\begin{enumerate}[label = (\Alph*)]
 \item\label{item:overview-a} $Z_1,\dots,Z_p$ is an RCD when restricted to every torso, i.e., $Z_1^{(t)},\dots,Z_p^{(t)}$ is an RCD of $\tor(t)$.
 \item\label{item:overview-b} All but at most $O(|Z'|)$ edges in $\tor(t)[Z_i^{(t)} \setminus Z']$ have both endpoints in the same connected component of $G[Z_i \setminus Z']$, for all $t \in V(T)$, $i \in [p]$, and $Z' \subseteq Z_i$.
\end{enumerate}
Condition~\ref{item:overview-a} naturally holds as long as we insist on constructing $Z_1,\dots,Z_p$ by combining local RCDs.
So the crucial question now is how to satisfy condition~\ref{item:overview-b}.

\paragraph{From Global to Local.}
Condition~\ref{item:overview-b} is not tractable, because it is a global constraint on $Z_1,\dots,Z_p$.
So we need to somehow reduce it to a local constraint on the RCD of each torso.
Let us fix a node $t \in V(T)$ and an index $i \in [p]$.
To have condition~\ref{item:overview-b}, the first thing we need is that (almost) every edge of $\tor(t)[Z_i^{(t)}]$ has its endpoints in the same connected component of $G[Z_i]$, which is the special case $Z' = \emptyset$.
But only having this is not enough.
We need in addition that loosing $k$ vertices in $Z_i$ only makes $O(k)$ fake edges become ``unsafe''.
To achieve this, our main idea is to require the two endpoints of every edge in $\tor(t)[Z_i^{(t)}]$ to be connected in $G[Z_i]$ by \emph{only} the vertices ``strictly below $t$'' in the tree decomposition $(T,\beta)$, that is, the vertices that only appear in the bags of descendants of $t$.
Formally, for every child $s$ of $t$, denote by $\gamma(s)$ the union of all bags in $T_s$, the subtree of $T$ rooted at $s$.
Then our requirement is that for every edge $uv$ of $\tor(t)[Z_i^{(t)}]$, $u$ and $v$ lie in the same connected component of $G[((\gamma(s) \setminus \sigma(s)) \cap Z_i) \cup \{u,v\}]$ for every child $s$ of $t$ such that $u,v \in \sigma(s)$.
Note that every non-fake edge trivially satisfies the requirement.

We show that this requirement is sufficient to guarantee condition~\ref{item:overview-b} above for all $Z' \subseteq Z_i$.
The definition of tree decompositions guarantees that the sets $\gamma(s) \setminus \sigma(s)$ are disjoint for different children $s$ of $t$.
We say a child $s$ of $t$ is \emph{bad} if $Z'$ contains at least one vertex in $\gamma(s) \setminus \sigma(s)$.
By the disjointness of the sets $\gamma(s) \setminus \sigma(s)$, the number of bad children of $t$ is at most $|Z'|$.
For convenience, we say a child $s$ of $t$ \emph{witnesses} an edge $uv$ of $\tor(t)[Z_i^{(t)}]$ if $u,v \in \sigma(s)$.
Note that each child $s$ of $t$ can witness at most $O(h^2)$ edges, because $|\sigma(s)| \leq h$.
Due to our requirement, if an edge $uv$ of $\tor(t)[Z_i^{(t)} \setminus Z']$ violates condition~\ref{item:overview-b} above, i.e., $u$ and $v$ lie in different connected components of $G[Z_i \setminus Z']$, then $Z'$ must contain at least one vertex in $\gamma(s) \setminus \sigma(s)$ for every child $s$ of $t$ satisfying $u,v \in \sigma(s)$, and in particular $(u,v)$ is witnessed by a bad node.
Since $t$ has at most $|Z'|$ bad children and each of them can witness $O(h^2)$ edges, the total number of edges of $\tor(t)[Z_i^{(t)} \setminus Z']$ violating condition~\ref{item:overview-b} is bounded by $O(|Z'|)$.

Based on the above argument, it now suffices to construct $p$-RCDs for each torso of $(T,\beta)$ such that the partition $Z_1,\dots,Z_p$ obtained by combining these local RCDs satisfies the aforementioned requirement for all $i \in [p]$ and $t \in V(T)$.
However, the current form of our requirement is global, because whether it is satisfied at a node $t$ depends on the construction at not only $t$ but also \emph{all descendants} of $t$ in $T$.
So next, let us transform it to a ``local'' form which can be checked by looking at the construction at each torso \emph{independently}.
We first observe that the following condition implies our previous requirement (which can be shown by a simple induction argument).
\begin{itemize}
 \item[] For every edge $uv$ of $\tor(t)[Z_i^{(t)}]$, the endpoints $u$ and $v$ lie in the same connected component of $\tor(s)[(Z_i^{(s)} \setminus \sigma(s)) \cup \{u,v\}]$ for every child $s$ of $t$ such that $u,v \in \sigma(s)$.
\end{itemize}
The above is actually equivalent to saying that for every child $s$ of $t$, all $u,v \in \sigma(s) \cap Z_i^{(s)}$ lie in the same connected component of $\tor(s)[(Z_i^{(s)} \setminus \sigma(s)) \cup \{u,v\}]$.
Importantly, this condition \emph{only} depends on the construction of $Z_1^{(s)},\dots,Z_p^{(s)}$, the RCD of $\tor(s)$.
If this condition holds on every node of $T$, then our previous requirement is also satisfied for every node of $T$.
Therefore, we have our new goal.
For every node $t \in V(T)$, we want the following.

\begin{enumerate}[label = (\Alph*)]
 \item $Z_1^{(t)},\dots,Z_p^{(t)}$ is an RCD of $\tor(t)$.
 \item For all $i \in [p]$, any two vertices $u,v \in \sigma(t) \cap Z_i^{(t)}$ lie in the same connected component of $\tor(t)[(Z_i^{(t)} \setminus \sigma(t)) \cup \{u,v\}]$.
\end{enumerate}

Now the new conditions above are both local, which allows us to consider each torso individually.
We shall construct the RCDs of the torsos in a top-down manner from the root of $T$ to the leaves.
Suppose we are at the node $t \in V(T)$ and going to construct the RCD $Z_1^{(t)},\dots,Z_p^{(t)}$ of $\tor(t)$,
At this point, the RCD at the parent of $t$ has been constructed, so each vertex in $\sigma(t)$ has already been assigned to one of the $p$ classes $Z_1^{(t)},\dots,Z_p^{(t)}$.
We then compute a $p$-RCD of $\tor(t) - \sigma(t)$ with an additional property that the $i$-th class of the RCD ``connects'' the vertices in $\sigma(t) \cap Z_i^{(t)}$ for all $i \in [p]$.
Combining this with the partition of $\sigma(t)$ gives us the desired $Z_1^{(t)},\dots,Z_p^{(t)}$ satisfying the above two conditions.
As long as such a single step can be achieved, we can eventually construct $Z_1,\dots,Z_p$.
So from now, we can restrict ourselves to one torso.

\paragraph{RCD with a Connectivity Constraint.}
Before discussing how to construct the desired RCD of $\tor(t) - \sigma(t)$, we need another twist to simplify the task in hand.
We notice that constructing an RCD of $\tor(t) - \sigma(t)$ satisfying the additional connectivity property is actually \emph{impossible} if we do not add any assumption on how the vertices in $\sigma(t)$ belong to the $p$ classes.
To see this, consider the following simple example where $\tor(t)$ is a star of 5 vertices, one central vertex connecting with 4 other vertices.
All vertices are contained in $\sigma(t)$ except the central vertex.
In $\sigma(t)$, two vertices belong to the class $Z_1^{(t)}$ and the other two vertices belong to $Z_2^{(t)}$.
Now in order to connect the two vertices in $\sigma(t) \cap Z_1^{(t)}$, the RCD of $\tor(t) - \sigma(t)$ must assign the central vertex to $Z_1^{(t)}$.
But if this is the case, we fail to connect the two vertices in $\sigma(t) \cap Z_2^{(t)}$.
Therefore, in this example, it is impossible to construct the desired RCD of $\tor(t) - \sigma(t)$.
In order to fix this issue, we have to somehow guarantee the ``connectivity task'' that $\sigma(t)$ leaves to us is not too complicated.

Fortunately, using a stronger version of Robertson-Seymour decomposition, we are able to reach a nice situation where only \emph{one} class of vertices in $\sigma(t)$ need to be connected by the RCD of $\tor(t) - \sigma(t)$.
This version of Robertson-Seymour decomposition, as stated in \cite{DemaineHK05,DeVosDOSRSV04}, guarantees that for every node $t \in V(T)$ and every child $s$ of $t$, the adhesion $\sigma(s)$ contains at most \emph{three} vertices in the embeddable part of $\tor(t)$ (recall that $\tor(t)$ is an $h$-almost-embeddable graph consisting of an embeddable part, vortices, and apices).
To see why this result helps us, let us consider an ideal case where every torso in the Robertson-Seymour decomposition has no vortices and apices.
In this case, every adhesion $\sigma(t)$ is of size at most three, because all vertices in the torso $\tor(t')$ of the parent $t'$ of $t$ are in the embeddable part of $\tor(t')$ and thus $\sigma(t)$ can contain at most three of them.
Since $|\sigma(t)| \leq 3$, there can be at most one class $Z_i^{(t)}$ that contains at least two vertices in $\sigma(t)$.
Therefore, when constructing the RCD of $\tor(t) - \sigma(t)$, we only need the $i$-th class to connect the vertices in $\sigma(t) \cap Z_i^{(t)}$.
In the general case where the torsos have vortices and apices, the situation becomes complicated.
But we can still manage to achieve this, which requires us to construct the RCD of each torso more carefully and then use the ``three-vertex'' property as above.

Now our goal becomes to construct an RCD of $\tor(t) - \sigma(t)$ in which one class is required to connect the corresponding vertices in $\sigma(t)$.
Essentially, we achieve this by proving the following.

\begin{lemma}[simplified version of Corollary~\ref{cor-decomp2}]
 \label{lem-simplified}
 Given a connected $h$-almost-embeddable graph $G$, a set $\varPhi \subseteq V(G)$ of size $c$, and a number $p$, one can compute in polynomial time $p$ disjoint sets $Z_1,\dots,Z_p \subseteq V(G)$ such that the following two conditions hold\footnote{The actual result, Corollary~\ref{cor-decomp2}, is more complicated, in which the sets $Z_1,\dots,Z_p$ have some additional properties and also there is an additional assumption on the almost-embeddable graph $G$. As we are not going to cover these techinical details in this overview, considering this simplified version should be enough.}.
 \begin{enumerate}[label = (\arabic*)]
  \item\label{item-simplified-1} $\tw(G/(Z_i \setminus Z')) = O_{h,c}(p+|Z'|)$ for all $i \in [p]$ and $Z' \subseteq Z_i$.
  \item\label{item-simplified-2} $\varPhi$ is contained in one connected component of $G - \bigcup_{i=1}^p Z_i$.
 \end{enumerate}
\end{lemma}

To see why the above lemma achieves our goal, note that $\tor(t) - \sigma(t)$ is $h$-almost-embeddable.
Furthermore, by constructing the Robertson-Seymour decomposition $(T,\beta)$ carefully, we can always make $\tor(t) - \sigma(t)$ connected and make every vertex in $\sigma(t)$ have a neighbor in $\beta(t) \setminus \sigma(t)$ in the torso $\tor(t)$; see Lemma~\ref{lem-rsdecomp} and Observation~\ref{obs-noadhtos}.
Without loss of generality, suppose we want to connect the vertices in $\sigma(t) \cap Z_k^{(t)}$ using the $k$-th class of the RCD of $\tor(t) - \sigma(t)$.
For each $v \in \sigma(t) \cap Z_i^{(t)}$, we mark a vertex $v' \in \beta(t) \setminus \sigma(t)$ that is neighboring to $v$ in $\tor(t)$.
Let $\varPhi \subseteq \beta(t) \setminus \sigma(t)$ be the set of marked vertices.
Now apply Lemma~\ref{lem-simplified} with $G = \tor(t) - \sigma(t)$ and the set $\varPhi$.
We then get the disjoint sets $Z_1,\dots,Z_p \subseteq \beta(t) \setminus \sigma(t)$.
The vertices in $Z_i$ are assigned to the class $Z_i^{(t)}$ for all $i \in [p]$.
Finally, the vertices in $(\beta(t) \setminus \sigma(t)) \setminus (\bigcup_{i=1}^p Z_i)$ are assigned to the class $Z_k^{(t)}$.
Then condition~\ref{item-simplified-1} of Lemma~\ref{lem-simplified} implies that $Z_1^{(t)},\dots,Z_p^{(t)}$ is an RCD of $\tor(t)$ and condition~\ref{item-simplified-2} guarantees that all $u,v \in \sigma(t) \cap Z_k^{(t)}$ lie in the same connected component of $\tor(s)[(Z_k^{(t)} \setminus \sigma(t)) \cup \{u,v\}]$.
Next, we summarize our ideas for proving Lemma~\ref{lem-simplified}.

\paragraph{Proof Sketch of Lemma~\ref{lem-simplified}.}
As mentioned at the beginning, the recent work~\cite{BandyapadhyayLLSJ22} already gives RCDs for almost-embeddable graphs, that is, it proves the special case of Lemma~\ref{lem-simplified} without the set $\varPhi$ and condition~\ref{item-simplified-2}.
Therefore, it is natural to ask whether one can directly generalize (possibly with slight modifications) the proof in \cite{BandyapadhyayLLSJ22} to obtain a proof of Lemma~\ref{lem-simplified}.
Unfortunately, this is not the case.
A close look at the proof in \cite{BandyapadhyayLLSJ22} reveals that the construction of $Z_1,\dots,Z_p$ in \cite{BandyapadhyayLLSJ22} ``inherently'' prevents us from having condition~\ref{item-simplified-2} of Lemma~\ref{lem-simplified}.
To see this, we need to briefly review how \cite{BandyapadhyayLLSJ22} constructs the sets $Z_1,\dots,Z_p$.

For convenience, we discuss the proof of \cite{BandyapadhyayLLSJ22} on a surface-embedded graph instead of an almost-embeddable graph.
In fact, the most difficult part of the work \cite{BandyapadhyayLLSJ22} also lies in decomposing a surface-embedded graph (specifically, the embeddable part of an almost-embeddable graph), and generalizing to almost-embeddable graphs is somehow easy.
In \cite{BandyapadhyayLLSJ22}, the construction of $Z_1,\dots,Z_p$ itself is simple, while the analysis of the RCD condition is involved.
In the first step, it generalizes the \emph{outerplanar layering} for planar graphs to surface-embedded graphs.
Recall that the outerplanar layering partitions the vertices of a planar graph $G$ into layers $L_1,\dots,L_m \subseteq V(G)$, where $L_i$ consists of the vertices on the outer boundary of $G - \bigcup_{j=1}^{i-1} L_j$.
Thus, $L_1$ is outer boundary of $G$, $L_2$ is the outer boundary of $G - L_1$, and so forth.
If $G$ is a graph embedded in a surface $\varSigma$, the same layering procedure still applies.
Indeed, we can fix a reference point $x_0 \in \varSigma$ and define the outer boundary of a $\varSigma$-embedded graph as the boundary of the face containing $x_0$ (assuming the image of the embedding is disjoint from $x_0$).
Given this, we can layer $G$ in the same way as above, by defining $L_i$ to be the vertices on the outer boundary of $G - \bigcup_{j=1}^{i-1} L_j$.
We call this generalization \emph{radial layering} for surface-embedded graphs.
Now let $L_1,\dots,L_m$ be the radial layering of $G$.
The analysis in \cite{BandyapadhyayLLSJ22} implies the following property of the layers (though not stated explicitly in \cite{BandyapadhyayLLSJ22}).

\begin{lemma}\label{lem-equi}
 If $Z = L_{i_1} \cup \cdots \cup L_{i_k}$ where $i_1 < \cdots < i_k$ and $|i_j - i_{j-1}| = O(p)$ for all $j \in [k+1]$ (set $i_0 = 0$ and $i_{k+1} = m$ for convenience), then $\tw(G/(Z \setminus Z')) = O(p+|Z'|)$ for all $Z' \subseteq Z$.
\end{lemma}

Then \cite{BandyapadhyayLLSJ22} simply defines each set $Z_i$ as the union of equidistant layers with distance $O(p)$ apart, that is, $Z_i = L_{q_i} \cup L_{\Delta+q_i} \cup L_{2 \Delta + q_i} \cup \cdots$ for some number $q_i \in [\Delta]$ where $\Delta = O(p)$.
If the choices of $q_1,\dots,q_p$ are different, $Z_1,\dots,Z_p$ are disjoint.
By Lemma~\ref{lem-equi}, $Z_1,\dots,Z_p$ satisfy the RCD condition.

Now let us see what is the difficulty to generalize this construction so that it further satisfies condition~\ref{item-simplified-2} of Lemma~\ref{lem-simplified}.
Consider the given set $\varPhi$ in Lemma~\ref{lem-simplified}, which is of size $c$.
We want $\varPhi$ to be contained in the complement $G-\bigcup_{i=1}^{p} Z_i$, and more importantly, $\varPhi$ has to be in one connected component of $G-\bigcup_{i=1}^{p} Z_i$.
To make $\varPhi$ contained in $G-\bigcup_{i=1}^{p} Z_i$ is easy.
We can mark the layers $L_i$ intersecting $\varPhi$ as ``bad'' layers.
There can be at most $c$ bad layers.
When constructing $Z_1,\dots,Z_p$, we do not include these bad layers.
According to Lemma~\ref{lem-equi}, we can still guarantee that $Z_1,\dots,Z_p$ satisfy condition~\ref{item-simplified-1} of Lemma~\ref{lem-simplified} even if we miss all bad layers, because the constant hidden in the bound of condition~\ref{item-simplified-1} can depend on $c$.
However, to further require the vertices in $\varPhi$ lying in the same connected component of $G-\bigcup_{i=1}^{p} Z_i$ is impossible.
Indeed, one can easily see that each layer $L_i$ is a separator in $G$ that separates the layers $L_1,\dots,L_{i-1}$ from the layers $L_{i+1},\dots,L_m$.
The layers in $Z_1,\dots,Z_p$ separate the remaining part of $G$ into many ``small'' pieces where each piece consists of at most $O(p)$ consecutive layers.
This highly-disconnected structure of $G-\bigcup_{i=1}^{p} Z_i$ naturally prevents the vertices in $\varPhi$ from lying in the same connected component, unless the layers containing $\varPhi$ are all close to each other.

Based on the above discussion, in order to satisfy both conditions in Lemma~\ref{lem-simplified}, we need to significantly modify the construction of \cite{BandyapadhyayLLSJ22} with various new ideas.
As we have seen, in the previous construction, the layers in $Z_1,\dots,Z_p$ serve as ``barriers'' that prevent the vertices in $\varPhi$ from being connected in the complement graph.
Therefore, a natural idea is to ``break'' these layers a little bit (i.e., remove some vertices from $Z_1,\dots,Z_p$) so that the vertices in $\varPhi$ can go across the barriers to connect with each other.
However, by doing this, the layers in each $Z_i$ become broken, and thus the new $Z_i$ might violate the RCD condition, i.e., condition~\ref{item-simplified-1} of Lemma~\ref{lem-simplified} (as we have fewer vertices to contract).
As such, we have to break the layers in some \emph{structured} way such that the vertices in $\varPhi$ can be connected in the complement graph and \emph{simultaneously} the RCD condition of $Z_1,\dots,Z_p$ preserves, which is the main challenge in our construction.
At this point, it is totally not clear how to achieve this goal.
So next, we begin with some simple intuitions.

\begin{figure}
 \centering
 \includegraphics[height=3.6cm]{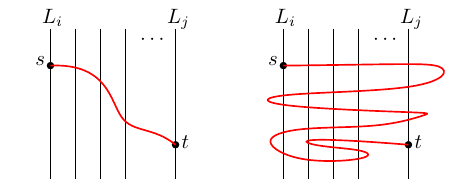}
 \caption{Monotone path vs. back-and-forth path}
 \label{fig-mono}
\end{figure}

Let us consider the simplest case where $\varPhi$ only contains two vertices $s$ and $t$.
In this case, it suffices to properly find an $(s,t)$-path $\pi$ in $G$ and then remove the vertices in $Z_1,\dots,Z_p$ on this path.
Since we do not want to lose the RCD condition of $Z_1,\dots,Z_p$, $\pi$ should break each layer in $Z_1,\dots,Z_p$ as ``little'' as possible.
So intuitively, a path that visits the layers $L_1,\dots,L_m$ \emph{monotonely} (left part of Figure~\ref{fig-mono}) is preferable to a path that goes back and forth across the layers many times (right part of Figure~\ref{fig-mono}).
If we do not have any requirement on the embedding of $G$ in $\varSigma$, it is not always possible to find a monotone (or mostly monotone) path connecting $s$ and $t$.
Thus, at the beginning (before the radial layering), we need to first modify the embedding of $G$ to a nice one, and do everything with respect to the nice embedding.
For surface graphs, a \emph{$2$-cell} embedding, in which each face is homeomorphic to a disk, is the nice embedding we need.
One can always compute a $2$-cell embedding for $G$ in polynomial time as long as the genus of $G$ is a constant.
(For almost-embeddable graphs, however, the situation is more involved, as we are not able to arbitrarily change the embedding of the embeddable part because of the vortices.
We have to define another type of embedding for which we can safely modify a given embedding of the embeddable part to.)
Now if $L_1,\dots,L_m$ are the radial layers for a 2-cell embedding of $G$, then one can go from every vertex in a layer $L_j$ to the previous layer $L_{j-1}$ by walking around the boundary of a face in between $L_{j-1}$ and $L_j$.
It turns out that for every vertex $v \in L_j$ and index $i \leq j$, there exists a path from $v$ to (some vertex in) $L_i$ that visits $L_j,L_{j-1},\dots,L_i$ monotonely.
Suppose $s \in L_i$ and $t \in L_j$ where $i \leq j$.
Note that although we can go from $t$ to $L_i$ monotonely, there does not necessarily exist a monotone path from $t$ to $s$.
So the key idea here is to, instead of using one path connecting $s$ and $t$, construct two monotone paths $\pi_s$ and $\pi_t$, where $\pi_s$ (resp., $\pi_t$) goes from $s$ (resp., $t$) to $L_1$.
Then we do not include the layer $L_1$ in any of $Z_1,\dots,Z_p$ (which is fine according to Lemma~\ref{lem-equi}).
Because $L_1$ is the outer boundary of $G$ and the embedding of $G$ is 2-cell, $G[L_1]$ is connected.
Therefore, $L_1$ together with the two paths $\pi_s$ and $\pi_t$ connects $s$ and $t$.
The same idea directly genearlizes to the case where $\varPhi$ has more than two vertices.
For each vertex $v \in \varPhi$, we try to construct a path $\pi_v$ which goes monotonely from $v$ to $L_1$.
Now $L_1$ together with all the paths connect everything in $\varPhi$.
It then remains to show how to construct these $c$ monotone paths carefully so that they do not break the layers in $Z_1,\dots,Z_p$ too much.

Since $c$ is treated as a constant in Lemma~\ref{lem-simplified}, if we are able to construct one path, it is not surprising that the same idea generalizes to constructing $c$ paths.
Thus, in what follows, we focus on the case $c = 1$.
We have $\varPhi = \{v\}$.
Since the path $\pi$ from $v$ to $L_1$ we are going to construct is monotone, it crosses each layer $L_i$ at most once, that is, after it leaves $L_i$ for $L_{i-1}$, it never comes back to $L_i$.
Therefore, we construct each part $L_i \cap \pi$ of $\pi$ from larger $i$ to smaller $i$ iteratively.
When constructing $L_i \cap \pi$, we basically need to consider how to walk in the layer $L_i$ from an arbitrary starting vertex $s \in L_i$ to an ``exit'' vertex that is adjacent to $L_{i-1}$ so that the walk does not break $L_i$ too much.
To this end, we have to figure out what ``too much'' means.
In other words, how much can we break each layer $L_i$ so that $Z_1,\dots,Z_p$ still satisfy the RCD condition?
By checking the proof of Lemma~\ref{lem-equi} in \cite{BandyapadhyayLLSJ22}, we see that the bound in Lemma~\ref{lem-equi} mainly relies on the following two properties of each layer $L_i$ (where $L_{>i}$ is the union of $L_{i+1},\dots,L_m$ and $L_{\leq i}$ is the union of $L_1,\dots,L_i$).
\begin{enumerate}[label = (\Roman*)]
 \item\label{item:layering-1} The neighbors of each connected component of $G[L_{>i}]$ lie in one face of $G[L_{\leq i}]$.
 \item\label{item:layering-2} For every $L' \subseteq L_i$, each connected component of $G[L_{>i}]$ is adjacent to $O(|L'|)$ vertices in the graph $G/(L_i \setminus L')$.
\end{enumerate}
Using the above two properties and the fact that the layers in each of $Z_1,\dots,Z_p$ are only $O(p)$ distance apart, one can finally show that $Z_1,\dots,Z_p$ satisfies the RCD condition.
Now we have the path $\pi$ that breaks $L_i$.
So we can only include in $Z_1,\dots,Z_p$ the part $L_i \setminus \pi$.
In this case, in order to make the proof work, we need the modified version of condition~\ref{item:layering-2} above: for every $L' \subseteq L_i \setminus \pi$, each connected component of $G[L_{>i}]$ is adjacent to $O(|L'|)$ vertices in the graph $G/((L_i \setminus \pi) \setminus L')$.
However, it is easy to see that this is impossible.
Consider the simplest case where $L' = \emptyset$.
We want each connected component of $G[L_{>i}]$ to have $O(1)$ neighbors in $G/(L_i \setminus \pi)$.
But in the worst case, after $\pi$ reaches $L_i$, it needs to go inside $L_i$ for a long distance in order to arrive at an exit vertex to leave $L_i$ for $L_{i-1}$.
Therefore, it can happen that $\pi$ contains a large fraction of $L_i$ in which many vertices can be adjacent to the same connected component $C$ of $G[L_{>i}]$.
As these vertices are not contracted in $G/(L_i \setminus \pi)$, the number of neighbors of $C$ in $G/(L_i \setminus \pi)$ can be unbounded.

Going deeper into the proof of \cite{BandyapadhyayLLSJ22} shows that the reason for why the above two properties help is basically the following.
These two properties allow us to partition $G$ into two parts $G[L_{\leq i}]$ and $G[L_{> i}]$ such that the connection between these two parts becomes ``weak'' after contracting a large subset of $L_i$, namely, each connected component of $G[L_{> i}]$ is only adjacent to few vertices in $G[L_{\leq i}]$ which lie in one face of $G[L_{\leq i}]$ (before contraction).
Now because the part in $L_i \cap \pi$ cannot be contracted, we are no longer able to have a weak connection between $G[L_{\leq i}]$ and $G[L_{> i}]$.
To get rid of this issue, our key idea here is to partition $G$ into two parts in a different way.
Specifically, when determining the part of $\pi$ contained in $L_i$, we also compute another subset $L_i^+ \subseteq L_i$.
Then we partition $G$ into two parts $G[L_{\leq i} \setminus L_i^+]$ and $G[L_{> i} \cup L_i^+]$, that is, we move $L_i^+$ from the part $G[L_{\leq i}]$ to the part $G[L_{> i}]$.
The choice of $L_i^+$ will guarantee the aforementioned weak connection between the two new parts.
Formally, $L_i^+$ (and $\pi$) satisfies the following properties.
\begin{enumerate}[label = (\Roman*)]
 \item The neighbors of each connected component of $G[L_{> i} \cup L_i^+]$ lie in one face of $G[L_{\leq i}]$ (and therefore one face of $G[L_{\leq i} \setminus L_i^+]$).
 \item For every $L' \subseteq L_i \setminus \pi$, each connected component of $G[L_{> i} \cup L_i^+]$ is adjacent to $O(|L'|)$ vertices in the graph $G/((L_i \setminus \{\pi \cup L_i^+\}) \setminus L')$.
\end{enumerate}
Note that such a set $L_i^+$ does not always exist for an arbitrary path $\pi$.
Therefore, we have to construct $L_i \cap \pi$ and $L_i^+$ simultaneously, and both of them need to be constructed carefully.
This task is achieved by Lemma~\ref{lem-key}.
As this step is technical and requires insights for surface-embedded graphs, we are not going to discuss it in detail.
Essentially, once we are able to construct each part $L_i \cap \pi$ of $\pi$ and the corresponding set $L_i^+$ satisfying the above two conditions, we can combine them to obtain the desired path $\pi$ from $v$ to $L_1$ and show that $\tw(G/((Z_i \setminus \pi) \setminus Z')) = O(p+|Z'|)$ for all $Z' \subseteq Z_i \setminus \pi$.
The sets $L_i^+$ are only used in the analysis of the treewidth bounds, and the analysis is built on the argument in \cite{BandyapadhyayLLSJ22}.
The same idea generalizes to the case where we need to construct $c$ paths (with a bit more work).
With this in hand, we are finally able to prove Lemma~\ref{lem-simplified} for surface-embedded graphs.
Of course, for almost-embeddable graphs, there are more technical details to be dealt with, but the basic idea remains the same.

\paragraph{Minimal embeddings of almost-embeddable graphs.}
In the last part of this overview, we discuss an interesting intermediate result achieved in our proof.
As aforementioned, we can modify the embedding of a (connected) surface graph to a 2-cell embedding.
However, for an almost-embeddable graph, we cannot require its embeddable part is 2-cell embedded.
There are two reasons.
First, when we change the embedding of the embeddable part, the structure of the vortices might be lost.
Second, even if the graph itself is connected, the subgraph excluding the apices is not necessarily connected, and thus does not admit a 2-cell embedding.
In order to make our proof work, we define a variant of 2-cell embedding, which we call \emph{minimal embedding}.
Basically, an almost-embeddable graph $G$ whose embeddable part is embedded with a minimal embedding satisfies the following condition.
Let $f$ be a face of the embeddable part of $G$, and $\widetilde{\partial} f$ be the subgraph of $G$ consisting of the boundary of $f$ and all vortices contained in $f$.
Then different connected components of $\widetilde{\partial} f$ belong to different connected components of $G-A$ where $A$ is the set of apices of $G$.
If $G$ does not have vortices and apices (i.e., $G$ is a surface graph), then the above condition is equivalent to saying that the boundary of every face of $G$ is connected.
We show that given an almost-embeddable graph $G$, we can compute (in polynomial time) a new almost-embeddable structure for $G$ in which the embeddable part is embedded with a minimal embedding (which is proved in Lemma~\ref{lem-minimal}).
This result is important for our proof, and might be of independent interest.

\section{Preliminaries}
\label{sec-pre}

\paragraph{Basic Notations.}
Let $G$ be a graph.
We denote by $V(G)$ and $E(G)$ the vertex set and the edge set of $G$, respectively.
For $U \subseteq V(G)$, we denote by $G[U]$ the induced subgraph of $G$ on $U$ and denote by $G - U$ the induced subgraph of $G$ on $V(G) \setminus U$.
Also, we denote by $N_G(U)$ the set of vertices in $V(G) \setminus U$ that are adjacent to at least one vertex in $U$, and write $N_G[U] \coloneqq N_G(U) \cup U$.
If $U = \{u\}$, we also write $N_G(u)$ and $N_G[u]$ instead of $N_G(\{u\})$ and $N_G[\{u\}]$, respectively.

\paragraph{Tree Decomposition and Treewidth.}
A \emph{tree decomposition} of a graph $G$ is a pair $(T,\beta)$ where $T$ is a tree and $\beta\colon V(T) \to 2^{V(G)}$ maps each node $t \in V(T)$ to a set $\beta(t) \subseteq V(G)$, called the \emph{bag} of $t$, such that
\begin{enumerate}[label = (\roman*)]
 \item $\bigcup_{t \in V(T)} \beta(t) = V(G)$,
 \item for every edge $uv \in E(G)$, there exists $t \in V(T)$ with $u,v \in \beta(t)$, and
 \item for every $v \in V(G)$, the set $\{t \in V(T) \mid v \in \beta(t)\}$ forms a connected subset in $T$.
\end{enumerate}
The \emph{width} of $(T,\beta)$ is $\max_{t \in V(T)} |\beta(t)| - 1$.
The \emph{treewidth} of a graph $G$, denoted by $\tw(G)$, is the minimum width of a tree decomposition of $G$.
It is sometimes more convenient to consider \emph{rooted} trees.
Throughout this paper, we always view the underlying tree of a tree decomposition as a rooted tree.

Let $(T,\beta)$ be a tree decomposition of a graph $G$.
The \emph{adhesion} of a non-root node $t \in V(T)$, denoted by $\sigma(t)$, is defined as $\sigma(t) \coloneqq \beta(t) \cap \beta(t')$ where $t'$ is the parent of $t$.
For convenience, we also define the adhesion of the root of $T$ as the empty set.
For each node $t \in V(T)$, we define the \emph{$\gamma$-set} of $t$ as $\gamma(t) \coloneqq \bigcup_{s \in V(T_t)} \beta(s)$ where $T_t$ is the subtree of $T$ rooted at $t$.
The \emph{torso} of $t$, denoted by $\tor(t)$, is the graph obtained from $G[\beta(t)]$ by making $\sigma(s)$ a clique for all children $s$ of $t$, i.e., adding edges between any two vertices $u,v \in \beta(t)$ such that $u,v \in \sigma(s)$ for some child $s$ of $t$.
The following two facts are well-known.

\begin{fact}\label{fact-partition}
 Let $(T,\beta)$ be a tree decomposition of $G$.
 Then $\{\beta(t) \setminus \sigma(t) \mid t \in V(T)\}$ is a partition of $V(G)$.
\end{fact}

\begin{proof}
 Let $x$ be a vertex of $G$ and define $T'$ to be the subtree induced by $\beta^{-1}(x)$.
 Let $r$ denote the highest element of $T'$.
 We claim that for every $t \neq r$ such that $t \in V(T')$, we have that $x \in \sigma(t)$. This would imply the desired result because then $x$ belongs to $\beta(t) \backslash \sigma(t)$ if and only if $t = r$.

 So let $t$ be a node of $T'$ different from $r$.
 In particular, it means that the parent $t'$ of $t$ also belongs to $T'$ and thus $x \in \beta(t) \cap \beta(t') = \sigma(t)$.
\end{proof}

\begin{fact}\label{fact-connintorso}
 Let $(T,\beta)$ be a tree decomposition of $G$.
 If $U \subseteq V(G)$ is a subset of vertices such that $G[U]$ is connected, then for all $t \in V(T)$, either $\tor(t)[U \cap \beta(t)]$ is connected or every connected component of $\tor(t)[U \cap \beta(t)]$ intersects $\sigma(t)$.
\end{fact}

\begin{proof}
 Let $U_1, \dots, U_r$ denote the connected components of $\tor(t)[U \cap \beta(t)]$ and let us assume $r \geq 2$ or nothing needs to be proved.
 Let $t'$ denote the parent of $t$ in $T$ and $s_1, \dots, s_\ell$ the children of $t$ in $T$.
 Because $G[U]$ is connected, it means that for every $i \in [r]$, there exists a path $P$ from $U_i$ to some other component $U_j$ in $G[U]$.
 Without loss of generality, we can assume that $P$ is a shortest such path.
 This means that, among the vertices of $P$ only the first vertex $v_i$ belongs to $U_i$, only the last vertex $v_j$ belongs to $U_j$ and the inner vertices of $P$ belong to $U \setminus \beta(t)$.
 In particular it means that all the inner vertices of $P$ belong to bags of nodes in a subtree obtained from $T$ by removing $t$.
 Because $v_i$ and $v_j$ are adjacent to some vertices on that path, they also belong to bags of nodes of that same subtree.
 If this subtree is the one containing $t'$, then $v_i \in \sigma(t)$.
 Otherwise, the subtree containing the inner vertices is the one attached to $t$ with $s_a$, for some $a \in [\ell]$.
 But then $v_i,v_j \in \sigma(s_a)$.
 However, $\sigma(s_a)$ is a clique in $\tor(t)$ which is a contradiction since $v_i$ and $v_j$ belong to different components of $\tor(t)[U \cap \beta(t)]$.
 This shows that, for every $i \in [r]$, $U_i$ intersects $\sigma(t)$.
\end{proof}

We shall also need the following useful lemma that relates the treewidth of a graph to the treewidth of the torsos in an arbitrary tree decomposition of the graph.

\begin{lemma} \label{lem-torsotw}
 If a graph $G$ admits a tree decomposition in which the treewidth of every torso is at most $k$, then $\tw(G) \leq 2k+1$.
\end{lemma}
\begin{proof}
 Let $(T, \beta)$  be a tree decomposition of $G$. For every $t \in V(T)$, let $(T_t, \beta_t)$ denote a tree decomposition of $\tor(t)$ of width $k$.
 Recall that for every child $s$ of $t$ in $T$, $\sigma(s) = \beta(t) \cap \beta(s)$ is a clique in $\tor(t)$.
 In particular, it means that there exists a node $t_s \in V(T_t)$ such that $\sigma(s) \subseteq \beta_t(t_s)$.
 The intuition here is that we can replace each bag $\beta(t)$ of $T$ with the tree decomposition $T_t$ by using the fact that $\sigma(s)$ is a clique in $\tor(t)$ for every child $s$ of $t$ to connect $T_t$ in a tree-like fashion.

 More formally, using the $(T_t, \beta_t)$, we can construct a tree decomposition $(T^*, \beta^*)$ of $G$ as follows: $T^*$ is the tree with vertex set the disjoint union of $V(T_t)$ for $t \in V(T)$.
 If $u$ and $v$ are two nodes of $V(T^*)$ such that $u \in V(T_t)$ and $v \in V(T_s)$, then we put an edge between $u$ and $v$ in $T^*$ if either (i) $s = t$ and $u$ is adjacent to
 $v$ in $T_t$ or (ii) $s$ is a child of $t$ in $T$, $u = t_s$ and $v$ is the root of $T_s$.
 In other words, $T^*$ is the tree obtained from the disjoint union of the $T_t$ where we add, for every node $t$ and child $s$ of $t$, the edge between the node $t_s \in V(T_t)$ and the root of $T_s$.

 For every $u \in V(T^*)$, such that $u \in (T_s)$ for some $s \in V(T)$, we define $\beta^*(u) \coloneqq \beta_s(u) \cup \sigma(s)$.
 Note that it immediately means that $|\beta^*(u)| \leq 2k+2$, since $(T_t, \beta_t)$ has width $k$ for all $t \in V(T)$.
 We show that $(T^*, \beta^*)$ is a tree decomposition of $G$ of width $2k+1$.

 First let us show that $T^*$ is a tree.
 Let $T'$ denote a subtree of $T$.
 We show by induction on the depth of $T'$ that the union of the $T_t$ for $t \in V(T')$ forms a subtree in $T^*$.
 If $T'$ has depth $1$, then $T'$ consists of a single vertex $t'$, and it follows from the fact that $T_{t'}$ is a tree and we kept all the edges in $T^*$.
 Let $t'$ denote the highest node of $T'$ and let $T_1,\dots, T_r$ denote the subtrees of $T'$ attached to $t'$.
 By the induction hypothesis, for every $i \in [r]$, the union $Q_i$ of the $T_t$ for $t \in V(T_i)$ is a tree.
 Moreover, by construction, we only add edges between nodes of $T_z$ and $T_y$ if $z$ is an ancestor of $y$ in $T$, or the opposite.
 In particular, it means that the $Q_i$ for $i \in [r]$ are disjoint subtrees.
 Moreover, $T_{t'}$ is a tree and if $s_i$ denotes the root of $T_i$ for some $i \in [r]$, then the edge between $t_{s_i}$ and the root of $T_{s_i}$ is the only edge between nodes of $T_{t'}$ and $Q_i$, which ends the induction.

 Moreover, it is easy to see that two adjacent vertices $x$ and $y$ of $G$ must share at least one bag of the tree decomposition.
 Indeed, since they share at least one bag $t \in V(T)$ and are adjacent in $\tor(t)$, they share one bag in $T_t$ and thus in $T^*$.

 Finally, let $x$ be a vertex of $G$ and $T'$ denote the subtree of $T$ such that $\beta^{-1}(x) = V(T')$.
 First note that $(\beta^*)^{-1}(x)$ is contained in the union of the sets $V(T_s)$ for $s \in V(T')$ by construction, since a vertex only gets added to some bag of $T_s$ if it belongs to $\sigma(s)$.
 Let $t$ denote the highest node of $T'$. We first remark that for every $s' \in V(T')$ different from $t$ and $u \in V(T_{s'})$, it holds that $x \in \beta^*(u)$. Indeed, since $s'$ is not the highest node of $T'$, it means that $x$ also belongs to the parent $t'$ of $s'$, and thus is in $\sigma(s')$. This means that $x$ gets added to all the bags of $T_{s'}$ by construction of $\beta^*$. Moreover, the set of bags of $T_t$ which contain $x$ is also connected and we know that if $x$ belongs to $\beta(t)$ and $\beta(s)$, then it belongs to $\tor(t)$ and there is an edge between a bag containing $x$ and the root of every $T_s$ for a child $s$ of $t$ in $T'$. This means that $(\beta^*)^{-1}(x)$ is a connected subtree of $T^*$.
\end{proof}

\paragraph{Edge Contractions.}
Edge contraction is a basic operation on graphs.
When an edge $uv$ of a graph $G$ is contracted, the two endpoints $u$ and $v$ are merged into one vertex whose neighbors are those in $N_G(\{u,v\})$.
For $U \subseteq V(G)$ we write $G/U$ to denote the graph obtained from $G$ by contracting all edges in $G[U]$, or equivalently, contracting every connected component of $G[U]$ to a single vertex.
If a graph $G'$ is obtained from $G$ by contracting edges, then each vertex of $G'$ corresponds to a subset of vertices of $G$ and these subsets form a partition of $V(G)$.
In this case, there is a naturally defined map $\pi\colon V(G) \rightarrow V(G')$ which maps each vertex $v \in V(G)$ to the vertex of $G'$ whose corresponding subset of $V(G)$ contains $v$.
We call $\pi$ the \emph{quotient map} of the contraction of $G$ to $G'$.
We need the following simple facts.

\begin{fact}[\cite{BandyapadhyayLLSX24}]\label{fact-inducedtd}
 Let $(T,\beta)$ be a tree decomposition of $G$, and $G'$ be a graph obtained from $G$ by edge contraction with the quotient map $\pi\colon V(G) \rightarrow V(G')$.
 Then $(T,\beta^*)$ is a tree decomposition of $G'$, where $\beta^*(t) = \pi(\beta(t))$ for all $t \in V(T)$.
\end{fact}

\begin{fact}\label{fact-quotient}
 Let $G$ be a graph and $G'$ be a graph obtained from $G$ by edge contraction with the quotient map $\pi\colon V(G) \rightarrow V(G')$.
 Also let $U' \subseteq V(G')$.
 Then $G'[U']$ is connected if and only if $G[\pi^{-1}(U')]$ is connected.
\end{fact}

\begin{proof}
 If $P$ is a path of $G$, then contracting an edge on this path still gives a path.
 So if $G[\pi^{-1}(U')]$ is connected, then $G'[U']$ is also connected.

 So Suppose that $P' = (x_1,\dots,x_p)$ is a path in $G'$.
 For every vertex $x_i$ on $P'$, the set $C_i \coloneqq \pi^{-1}(x_i)$ is connected in $G$.
 So $G[\pi^{-1}(P')]$ can be seen as a sequence of connected sets $C_1,\dots, C_p$, where each edge $x_ix_{i+1}$ of $P'$ corresponds to an edge between $C_i$ and $C_{i+1}$ in $G$.
 Hence, $G[\pi^{-1}(P')]$ is connected.
 This¸ implies that if $G'[U']$ is connected, then $G[\pi^{-1}(U')]$ is also connected.
\end{proof}

\paragraph{Graph Minors.}
A graph $H$ is a \emph{minor} of a graph $G$ (or $G$ contains $H$ as a minor) if $H$ can be obtained from $G$ by deleting vertices, deleting edges, and contracting edges.
A graph $G$ is \emph{$H$-minor-free} if $H$ is not a minor of $G$.
The following fact gives us an alternative criterion for determining whether a graph is a minor of another graph.

\begin{fact}[\cite{CyganFKLMPPS15}]\label{fact-minormap}
 A graph $G$ contains another graph $H$ as a minor if and only if there exists $U \subseteq V(G)$ and a surjective map $\rho\colon U \to V(G')$ satisfying the following condition: for all $U' \subseteq V(G')$ such that $G'[U']$ is connected, the graph $G[\rho^{-1}(U')]$ is connected.
\end{fact}

\paragraph{Almost-Embeddable Graphs.}
The class of almost-embeddable graphs is a generalization of the class of bounded-genus graphs, and is related to $H$-minor-free graphs due to the profound work of Robertson and Seymour \cite{RobertsonS03a}.
A graph $G$ is $h$-\emph{almost-embeddable} if it admits an $h$-almost-embeddable structure described below.

\begin{definition}\label{def-almost}
 An $h$-\emph{almost-embeddable structure} of a graph $G$ consists of a set $A \subseteq V(G)$ with $|A| \leq h$, a decomposition $G-A = G_0 \cup G_1 \cup \cdots \cup G_r$ for $r \leq h$, an embedding $\eta$ of $G_0$ to a surface $\varSigma$ of (Euler) genus $g$ with $g \leq h$, $r$ disjoint disks $D_1,\dots,D_r$ each of which is contained in a face (can possibly intersect the boundary of the face) of the $\varSigma$-embedded graph $(G_0,\eta)$, and $r$ pairs $(\tau_1,\mathcal{P}_1),\dots,(\tau_r,\mathcal{P}_r)$ such that the following conditions hold for all $i \in [r]$:
 \begin{itemize}
  \item $G_1,\dots,G_r$ are mutually disjoint.
  \item $V(G_0) \cap_\eta D_i = V(G_0) \cap V(G_i)$, where $V(G_0) \cap_\eta D_i$ consists of the vertices in $V(G_0)$ whose image under $\eta$ is contained in the disk $D_i$.
   Set $q_i = |V(G_0) \cap_\eta D_i| = |V(G_0) \cap V(G_i)|$.
  \item $\tau_i = (v_{i,1},\dots,v_{i,q_i})$ is a permutation of the vertices in $V(G_0) \cap_\eta D_i$ that is compatible with the clockwise or counterclockwise order along the boundary of $D_i$.
  \item $\mathcal{P}_i = (\pi_i,\beta)$ is a path decomposition\footnote{Path decomposition is a special case of tree decomposition in which the tree is a path.} of $G_i$ of width at most $h$ where the path $\pi_i = (u_{i,1},\dots,u_{i,q_i})$ of the decomposition is of length $q_i - 1$ and satisfies $v_{i,j} \in \beta(u_{i,j})$ for all $j \in [q_i]$.
 \end{itemize}
 Conventionally, we call $A$ the \emph{apex set}, $G_0$ the \emph{embeddable part}, $\eta:G_0 \rightarrow \varSigma$ the \emph{partial embedding}, $G_1,\dots,G_r$ the \emph{vortices} attached to the disks $D_1,\dots,D_r$.
 Also, we call each pair $(\tau_i,\mathcal{P}_i)$ the \emph{witness pair} of the vortex $G_i$, for $i \in [r]$.
 The \emph{vortex vertices} of $G$ refers to the vertices in $\bigcup_{i=1}^r V(G_i)$.
\end{definition}

\section{Proof of Theorem~\ref{thm-contraction}}
\label{sec-proof}

In this section, we prove our main theorem, which is restated below.

\main*

\subsection{Useful results for surface-embedded graphs}
We begin with introducing some basic notions and results about surface-embedded graphs.
For the purpose of our proof, sometimes it is more convenient to consider graphs embedded in a surface with a reference point.
Let $(\varSigma,x_0)$ be a pair where $\varSigma$ is a connected closed surface, and $x_0 \in \varSigma$ is a reference point on $\varSigma$.
A $(\varSigma,x_0)$-embedded graph refers to a $\varSigma$-embedded graph whose image on $\varSigma$ is disjoint from $x_0$, i.e., a graph embedded on $\varSigma \backslash \{x_0\}$.
For example, plane graphs are exactly $(\varSigma,x_0)$-embedded graphs for $\varSigma = \mathbb{S}^2$, where the reference point $x_0 \in \mathbb{S}^2$ is the ``point at infinity'' of the plane.
We typically represent a $(\varSigma,x_0)$-embedded graph by a pair $(G,\eta)$, where $G$ is the graph itself and $\eta\colon G \rightarrow \varSigma$ is an embedding of $G$ to $\varSigma$ whose image avoids $x_0$, which we call a $(\varSigma,x_0)$-embedding.
Let $(G,\eta)$ be a $(\varSigma,x_0)$-embedded graph.
For any subgraph $G'$ of $G$, $\eta$ induces an $(\varSigma,x_0)$-embedding of $G'$; for convenience, we usually use the same notation ``$\eta$'' to denote this subgraph embedding.
A \emph{face} of $(G,\eta)$ refers to (the closure of) a connected component of $\varSigma \backslash \eta(G)$, where $\eta(G)$ is the image of $G$ on $\varSigma$ under the embedding $\eta$.
We denote by $F_\eta(G)$ the set of faces of $(G,\eta)$.
With the reference point $x_0$, we can define the \emph{outer face} of $(G,\eta)$, which is the (unique) face of $(G,\eta)$ that contains $x_0$.
The other faces of $(G,\eta)$ are called \emph{inner faces}.
The \emph{boundary} of a face $f \in F_\eta(G)$, denoted by $\partial f$, is the subgraph of $G$ consisting of all vertices and edges that are incident to $f$ (under the embedding $\eta$).
Note that $f$ itself is a face in its boundary subgraph $(\partial f, \eta)$, i.e., $f \in F_\eta(\partial f)$.
We say a face $f \in F_\eta(G)$ is \emph{singular} if its boundary $\partial f$ is not connected.
The following lemma gives useful properties of the boundary subgraph of a face.

\begin{lemma}\label{lem-degree}
 Let $o \in F_\eta(G)$ be a face of a $(\varSigma,x_0)$-embedded graph $(G,\eta)$, where $g$ is the genus of $\varSigma$.
 Then $(\partial o,\eta)$ has $O(g)$ singular faces.
 Furthermore, every face $f \in F_\eta(\partial o) \backslash \{o\}$ satisfies the following properties.
 \begin{enumerate}[label = (\arabic*)]
  \item\label{item:degree-1} $\partial f$ has $O(g)$ connected components.
  \item\label{item:degree-2} The degree of every vertex of $\partial f$ is $O(g)$.
  \item\label{item:degree-3} There are $O(g)$ vertices of $\partial f$ whose degrees are at least 3.
 \end{enumerate}
\end{lemma}

\begin{proof}
Consider a face $f \in F_\eta(\partial o) \backslash \{o\}$.
We first figure out what the faces of $(\partial f,\eta)$ are.
Clearly, $f$ itself is a face of $(\partial f,\eta)$.
In addition, since $\partial f$ is a subgraph of $\partial o$, there must be another face $f_0 \in F_\eta(\partial o)$ such that $o \subseteq f_0$.
By definition, each edge of $\partial f$ is incident to $f$.
Also, each edge of $\partial f$ is incident to $f_0$, because it is incident to $o$ in $(\partial o,\eta)$.
Note that one edge can be incident to at most two faces in a surface-embedded graph, and thus the edges of $\partial f$ are only incident to $f$ and $f_0$.
It follows that $f$ and $f_0$ are the only two faces of $(\partial f,\eta)$, i.e., $F_\eta(\partial f) = \{f,f_0\}$, because any face of $(\partial f,\eta)$ must be incident to some edge of $\partial f$.
We further argue that $\partial f$ has no vertex of degree 0 or 1.
Similarly to the edges, each vertex of $\partial f$ is incident to $f$ (by definition) and also incident to $f_0$ since it is incident to $o$ in $(\partial o,\eta)$.
But vertices of degree 0 or 1 can only be incident to one face in a surface-embedded graph.
Thus, every vertex of $\partial f$ has degree at least 2.

Now we are ready to prove the lemma.
We first show $(\partial o,\eta)$ only has $O(g)$ singular faces.
It suffices to bound the number of singular faces other than $o$.
Let $f \in F_\eta(\partial o) \backslash \{o\}$ be a singular face.
As argued above, every vertex of $\partial f$ has degree at least 2.
Thus, every connected component of $\partial f$ contains a cycle.
Since $f$ is singular, $\partial f$ has at least two connected components, which implies that $\partial f$ contains two disjoint cycles.
For each singular face in $F_\eta(\partial o) \backslash \{o\}$, we pick such two disjoint cycles on its boundary.
Let $\gamma$ denote the number of singular faces in $F_\eta(\partial o) \backslash \{o\}$.
Then there are in total $2 \gamma$ cycles picked.
Note that these cycles are \emph{edge-disjoint}.
Indeed, for any face $f \in F_\eta(\partial o) \backslash \{o\}$, the edges in $E(\partial f)$ must have one side incident to $o$ and the other side incident to $f$, and thus are not incident to any other face in $F_\eta(\partial o)$.
So the boundaries of faces in $F_\eta(\partial o) \backslash \{o\}$ are edge-disjoint and in particular the cycles picked are edge-disjoint.
Now we delete an arbitrary edge on each cycle we pick.
Let $(\partial o)^-$ be the resulting graph.
Note that $\partial o$ and $(\partial o)^-$ share the same vertex set, i.e., $(\partial o)^-$ is a spanning subgraph of $\partial o$.
Furthermore, if two vertices are in the same connected component of $\partial o$, they are also in the same connected component of $(\partial o)^-$, because the $2 \gamma$ cycles picked are edge-disjoint and we only delete one edge from each cycle (so the cycle is still connected by the remaining edges).
This implies that $\partial o$ and $(\partial o)^-$ have the same number of connected components.
Define $\Delta_E \coloneqq |E(\partial o)| - |E((\partial o)^-)|$ and $\Delta_F \coloneqq |F_\eta(\partial o)| - |F_\eta((\partial o)^-)|$.
Clearly, $\Delta_E = 2\gamma$.
Also, we have $\Delta_F \leq -\gamma$, because all non-singular faces in $F_\eta(\partial o) \backslash \{o\}$ are preserved in $((\partial o)^-,\eta)$ (as we do not delete any edge on their boundaries), while $o$ and all singular faces in $F_\eta(\partial o) \backslash \{o\}$ are merged into one face in $((\partial o)^-,\eta)$ (due to the edges we delete).
Since $\partial o$ and $(\partial o)^-$ have the same numbers of vertices and connected components, Euler's formula implies $|\Delta_E + \Delta_F| = O(g)$.
Therefore, $\gamma = O(g)$.

Next, we show that every face $f \in F_\eta(\partial o) \backslash \{o\}$ satisfies the properties \ref{item:degree-1}-\ref{item:degree-3} in the lemma.
Let $\#_V$, $\#_E$, $\#_F$, $\#_C$ be the numbers of vertices, edges, faces, connected components of $(\partial f,\eta)$.
By Euler's formula, $|\#_V-\#_E+\#_F-\#_C| = O(g)$, where $g$ is the genus of $\varSigma$.
We have shown that $\#_F = |F_\eta(\partial f)| = 2$.
Therefore, $|\#_V - (\#_E+\#_C)| = O(g)$.
Since every vertex of $\partial f$ has degree at least 2, we have $\#_E \geq \#_V$, which implies $\#_C = O(g)$ and $\#_E - \#_V = O(g)$.
The former is property \ref{item:degree-1} of the lemma.
The latter further implies $\sum_{v \in V(\partial f)} (\mathsf{deg}(v)-2) = 2 \#_E - 2 \#_V = O(g)$, where $\mathsf{deg}(v)$ denotes the degree of $v$ in $\partial f$.
As each vertex of $\partial f$ has degree at least 2, $\mathsf{deg}(v)-2 \geq 0$ for all $v \in V(\partial f)$.
Therefore, $\mathsf{deg}(v)-2 = O(g)$ for all $v \in V(\partial f)$ and there are $O(g)$ vertices $v \in V(\partial f)$ such that $\mathsf{deg}(v)-2 > 0$.
This proves \ref{item:degree-2} and \ref{item:degree-3} of the lemma.
\end{proof}

The \emph{vertex-face incidence} (VFI) graph of $(G,\eta)$ is a bipartite graph with vertex set $V(G) \cup F_\eta(G)$ and edges connecting every pair $(v,f) \in V(G) \times F_\eta(G)$ such that $v$ is incident to $f$ (or equivalently, $v$ is a vertex in $\partial f$).
It is known that the VFI graph of $(G,\eta)$ is always connected \cite{BandyapadhyayLLSJ22}.
Let $G^*$ be the VFI graph of $(G,\eta)$.
A \emph{vertex-face alternating} (VFA) path in $(G,\eta)$ refers to a path in $G^*$; it is called ``vertex-face alternating'' because vertices and faces of $(G,\eta)$ appear alternately on the path.
For two vertices $v,v' \in V(G)$, the \emph{vertex-face distance} between $v$ and $v'$ in $(G,\eta)$ is defined as the shortest-path distance between $v$ and $v'$ in $G^*$.
In the same way, we can define the \emph{vertex-face distance} between a vertex and a face of $(G,\eta)$, or between two faces of $(G,\eta)$.
The \emph{vertex-face diameter} of $(G,\eta)$, denoted by $\mathsf{diam}^*(G,\eta)$, is the maximum vertex-face distance between vertices/faces of $(G,\eta)$, or equivalently, the graph diameter of the VFI graph of $(G,\eta)$.
Following from the classical work \cite{Eppstein00} of Eppstein, it is known that $\tw(G) = O_g(\mathsf{diam}^*(G,\eta))$.
Recently, Bandyapadhyay et al.\ \cite{BandyapadhyayLLSJ22} generalized this result in the following way.
Let $w\colon F_\eta(G) \rightarrow \mathbb{N}$ be a weight function on the faces of $(G,\eta)$.
Consider a simple path $\pi = (a_0,a_1,\dots,a_m)$ in $G^*$, and write $F_\pi = F_\eta(G) \cap \{a_0,a_1,\dots,a_m\}$.
The \emph{cost} of $\pi$ under the weight function $w$ is defined as $m+\sum_{f \in F_\pi} w(f)$, i.e., the length of $\pi$ plus the total weights of the faces that $\pi$ goes through.
We define the \emph{$w$-weighted vertex-face distance} between vertices/faces $a,a' \in V(G) \cup F_\eta(G)$ in $(G,\eta)$ as the minimum cost of a path connecting $a$ and $a'$ in $G^*$ under the weight function $w$.
Note that $w$-weighted vertex-face distances satisfy the triangle inequality, although they do not necessarily form a metric because the $w$-weighted vertex-face distance from a face $f \in F_\eta(G)$ to itself is $w(f)$ rather than 0.
The \emph{$w$-weighted vertex-face diameter} of $(G,\eta)$, denoted by $\mathsf{diam}_w^*(G,\eta)$, is the maximum $w$-weighted vertex-face distance between vertices/faces of $(G,\eta)$.
Clearly, when $w$ is the zero function, the $w$-weighted vertex-face distance/diameter coincides with the ``unweighted'' vertex-face distance/diameter defined before.
Bandyapadhyay et al.\ \cite{BandyapadhyayLLSJ22} generalizes the bound $\tw(G) = O_g(\mathsf{diam}^*(G,\eta))$ by proving the following lemma.

\begin{lemma}[\cite{BandyapadhyayLLSJ22}]\label{lem-twdiam2}
Let $(G,\eta)$ be a $(\varSigma,x_0)$-embedded graph where the genus of the surface $\varSigma$ is $g$ and $\kappa: F_\eta(G) \rightarrow 2^{V(G)}$ be a map satisfying $\kappa(f) \subseteq V(\partial f)$.
Then we have $\tw(G^\kappa) = O(\mathsf{diam}_{w_\kappa}^*(G,\eta))$, where $G^\kappa$ is the graph obtained from $G$ by making $\kappa(f)$ a clique for all $f \in F_\eta(G)$ and $w_\kappa: F_\eta(G) \rightarrow \mathbb{N}$ is a weighted function on the faces of $(G,\eta)$ defined as $w_\kappa(f) = |\kappa(f)|$.
\end{lemma}

The \emph{radial layering} of a $(\varSigma,x_0)$-embedded graph $(G,\eta)$ is a partition $L_1,\dots,L_m$ of $V(G)$ where $L_i$ consists of the vertices which have vertex-face distance $2i-1$ to the outer face of $(G,\eta)$.
Alternatively, $L_i$ can be defined as the vertices incident to the outer face of $(G - \bigcup_{j=1}^{i-1} L_j, \eta)$.
The notion of radial layering generalizes the well-known outerplanar layering of plane graphs.
We have the following simple observation for the radial layering.

\begin{fact}\label{fact-diff1}
Let $L_1,\dots,L_m$ be the radial layering of a $(\varSigma,x_0)$-embedded graph $(G,\eta)$.
Then, for any $f \in F_\eta(G)$, $V(\partial f) \subseteq L_{i-1} \cup L_i$ for some $i \in [m]$.
Also, $N_G[L_i] \subseteq L_{i-1} \cup L_i \cup L_{i+1}$ for all $i \in [m]$.
\end{fact}

\begin{proof}
Let $v,v' \in V(\partial f)$ for some $f \in F_\eta(G)$.
Denote by $d_v$ (resp., $d_{v'}$) the vertex-face distance between the outer face of $(G,\eta)$ and $v$ (resp., $v'$).
As both $v$ and $v'$ are incident to $f$, the vertex-face distance between $v$ and $v'$ in $(G,\eta)$ is $2$, which implies $|d_v - d_{v'}| \leq 2$.
Thus, the indices of the layers containing $v$ and $v'$ differ by at most 1.
It follows that $V(\partial f) \subseteq L_{i-1} \cup L_i$ for some $i \in [m]$.

To see $N_G[L_i] \subseteq L_{i-1} \cup L_i \cup L_{i+1}$, let $v \in L_i$ and $v' \in N_G(v)$.
Since $vv' \in E(G)$ is an edge, it must be incident to some face $f \in F_\eta(G)$, which implies that $v$ and $v'$ are incident to $f \in F_\eta(G)$.
By the observation above, we have either $V(\partial f) \subseteq L_{i-1} \cup L_i$ or $V(\partial f) \subseteq L_i \cup L_{i+1}$.
Thus, $v' \in L_{i-1} \cup L_i \cup L_{i+1}$.
\end{proof}

For notational convenience, if $L_1,\dots,L_m$ are the radial layers, then we write $L_{\geq a} = \bigcup_{i=a}^m L_i$, $L_{> a} = \bigcup_{i=a+1}^m L_i$, $L_{\leq a} = \bigcup_{i=1}^a L_i$, $L_{< a} = \bigcup_{i=1}^{a-1} L_i$, and $L_a^b = \bigcup_{i=a}^b L_i$.
We shall use these notations throughout this section.
Also, we shall assume $L_i = \emptyset$ for all indices $i \leq 0$ and $i >m$.

An \emph{extension} of a $(\varSigma,x_0)$-embedded graph $(G,\eta)$ is another $(\varSigma,x_0)$-embedded graph $(G',\eta')$ where $G'$ is a supergraph of $G$ and $(G,\eta') = (G,\eta)$, i.e., the restriction of $\eta'$ to $G$ is the same as $\eta$.
We say the extension is \emph{outer-preserving} if the images of all vertices in $V(G') \backslash V(G)$ and all edges in $E(G') \backslash E(G)$ under $\eta'$ are in the inner faces of $(G,\eta)$.
Thus, an outer-preserving extension has the same outer face and the same outer-face boundary as the original graph.

An embedding $\eta\colon G \rightarrow \varSigma$ is \emph{minimal} if for every face $o \in F_\eta(G)$, the boundary subgraph $(\partial o,\eta)$ satisfies the following condition: $V(\partial f)$ lies in one connected component of $\partial o$ for every $f \in F_\eta(\partial o) \backslash \{o\}$.
We have the following simple observation.

\begin{fact}\label{fact-minimal}
Let $\eta\colon G \rightarrow \varSigma$ be a minimal embedding.
If $G'$ is an induced subgraph of $G$ that is the disjoint union of some connected components of $G$, then the induced embedding $\eta\colon G' \rightarrow \varSigma$ is also a minimal embedding.
Also, if $G'$ is a disjoint union of $G$ and some isolated vertices and $\eta'\colon G' \rightarrow \varSigma$ is an extension of $\eta$, then $\eta'$ is a minimal embedding.
\end{fact}

\begin{proof}
To see the first statement, let $G'$ be the disjoint union of some connected components of $G$.
Then $G$ is the disjoint union of $G'$ and another graph $G''$.
Suppose $\eta\colon G' \rightarrow \varSigma$ is not minimal.
So there exist $o' \in F_\eta(G')$ and $f' \in F_\eta(\partial o') \backslash \{o'\}$ such that $V(\partial f')$ is not contained in one connected component of $\partial o'$.
Let $C_1$ and $C_2$ be two connected components of $\partial o'$ that contain at least one vertex in $V(\partial f')$.
Suppose $v \in V(\partial f')$ is a vertex contained in $C_1$.
As $G'$ is a subgraph of $G$, each face of $(G,\eta)$ is contained in one face of $(G',\eta)$.
There exists a face $o \in F_\eta(G)$ such that $o \subseteq o'$ and $o$ is incident to $v$, i.e., $v \in F_\eta(\partial o)$.
As $o \subseteq o'$, $f'$ is contained in some face $f \in F_\eta(\partial o)$ of $(\partial o, \eta)$.
We claim that $V(\partial f)$ is not contained in one connected component of $\partial o$, which contradicts the fact that $\eta\colon G \rightarrow \varSigma$ is a minimal embedding.
First notice that $v \in V(\partial f)$.
Indeed, since (the image of) $v$ is contained in $f'$, it is also contained in $f$ (but it is not in the interior of $f$ as $v \in V(\partial o)$).
So $v$ is incident to $f$ and $v \in V(\partial f)$.
Now it suffices to find another vertex $v' \in V(\partial f)$ that is not contained in the same connected component of $\partial o$ as $v$.
If $\partial f$ contains some vertex in $G''$, we are done, because any vertex in $G''$ cannot be in the same connected component of $\partial o$ as $v$ (as $G$ is the disjoint union of $G'$ and $G''$, and $v \in V(G')$).
Otherwise, $\partial f$ is a subgraph of $G'$.
In this case, we pick an arbitrary vertex $v' \in V(\partial f')$ that is contained in $C_2$.
We observe that $v' \in V(\partial f)$.
To see this, let $x \in \varSigma$ be a point in the interior of $o$ (and thus in the interior $o'$).
Since $v' \in V(\partial o')$, there exists a curve $\gamma$ on $\varSigma$ connecting $x$ and (the image of) $v'$ that lies in the interior of $o'$ (except the endpoints).
Note that $\gamma$ must intersect the boundary of $f$, because $x \in o$ and $\eta(v') \in f' \subseteq f$.
So either $v' \in V(\partial f)$ or the interior of $\sigma$ intersects (the image of) $\partial f$.
The latter is impossible because the interior of $\sigma$ does not intersect (the images of) any vertex/edges of $G'$ (as $\gamma$ lies in the interior of $o$) but $\partial f$ is a subgraph of $G'$ by assumption.
Thus, $v' \in V(\partial f)$.
To see that $v$ and $v'$ lie in different connected components of $\partial o$, simply notice that $\partial o \cap G'$ is a subgraph of $\partial o' \cap G'$.
Indeed, as $o \subseteq o'$ and $o'$ is a face of $(G',\eta)$, any vertex/edge of $\partial o \cap G'$ should be embedded in $o'$ and thus incident to $o'$.
Since $v$ and $v'$ lie in different connected components of $\partial o'$ (namely, $C_1$ and $C_2$), they also lie in different connected components of $\partial o$.

The second statement is somewhat trivial.
If $G'$ is a disjoint union of $G$ and some isolated vertices and $\eta'\colon G' \rightarrow \varSigma$ is an extension of $\eta$, then $(G,\eta)$ and $(G',\eta')$ have the same faces (with possibly different boundaries because of the isolated vertices).
Consider a face $o' \in F_{\eta'}(G')$.
There exists a face $o \in F_\eta(G)$ such that $\partial o'$ is the disjoint union of $\partial o$ and some isolated vertices whose images are in the interior of $o$.
Then $F_{\eta'}(\partial o')$ and $F_\eta(\partial o)$ are totally the same.
Any face $f' \in F_{\eta'}(\partial o') \backslash \{o'\}$ can also be viewed as a face in $F_\eta(\partial o) \backslash \{o\}$.
Note that $V(\partial f')$ does not contain the isolated vertices in $V(\partial o') \backslash V(\partial o)$ and thus is contained in $V(\partial o)$.
Furthermore, since $\eta$ is minimal, all vertices in $V(\partial f')$ are contained in the same connected component of $\partial o$, and thus in the same connected component of $\partial o'$.
\end{proof}

Finally, we need to establish an important lemma which allows us the transform a given almost-embeddable structure of a graph to another almost-embeddable structure in which the partial embedding is minimal.

\begin{lemma}\label{lem-minimal}
Given a graph $G$ with an $h$-almost-embeddable structure, one can compute in polynomial time a new $O_h(1)$-almost-embeddable structure of $G$ in which the partial embedding is minimal.
Furthermore, every apex in the given structure is still an apex in the new structure, while every vortex vertex in the given structure is either a vortex vertex or an apex in the new structure.
\end{lemma}

\begin{proof}
Let $G_0$ be the embeddable part of $G$ and $\eta\colon G_0 \rightarrow \varSigma$ be the partial embedding.
If $\eta$ is minimal, we are done.
If $\eta$ is not minimal, we show that one can compute a new $O_h(1)$-almost-embeddable structure of $G$ in which the embeddable part is embedded in a surface of genus strictly smaller than $\varSigma$.
In addition, every apex in the old structure is still an apex in the new structure, while every vortex vertex in the old structure is either a vortex vertex or an apex in the new structure.
Note that this implies the lemma.
Indeed, we can keep doing this until the partial embedding in the almost-embeddable structure is minimal.
As the genus of the surface is always decreasing and the original surface $\varSigma$ is of genus at most $h$, this procedure will terminate in $h$ steps.

Suppose $\eta$ is not minimal.
Let $G_1,\dots,G_r$ be the vortices of $G$ attached to disjoint facial disks $D_1,\dots,D_r$ in $(G_0,\eta)$ with witness pairs $(\tau_1,\mathcal{P}_1),\dots,(\tau_r,\mathcal{P}_r)$, where $r \leq h$.
To get some intuition, we first consider an ideal case, in which we have a \emph{non-separating} simple closed curve $\gamma$ on $\varSigma$ that is contained in the interior of some face of $(G_0,\eta)$ and is disjoint from all the disks $D_1,\dots,D_r$.
Here ``non-separating'' means that removing $\gamma$ from $\varSigma$ does not split $\varSigma$ into multiple connected components, i.e., $\varSigma \backslash \gamma$ is connected.
Then we can apply the well-known cut-and-paste operation (see \cite{Diestel05} for example) to reduce the genus of $\varSigma$ while keeping the partial embedding $\eta$ and the vortex disks $D_1,\dots,D_r$.
Specifically, we cut $\varSigma$ along $\gamma$, resulting in a surface with one or two boundary circles (depending on whether $\gamma$ has one side or two sides in $\varSigma$), and attach disks to the boundary circles to make it a surface (without boundary), which we denote by $\varSigma'$.
It is known that the genus of the new surface $\varSigma'$ is strictly smaller than the genus of $\varSigma$.
Furthermore, as $\gamma$ is inside a face of $(G_0,\eta)$ and is disjoint from the disks $D_1,\dots,D_r$, the images of the vertices/edges of $G_0$ and the disks $D_1,\dots,D_r$ preserve in $\varSigma'$.
So this gives us a new almost-embeddable structure of $G$ whose underlying surface is of smaller genus than $\varSigma$.
It is not difficult to see that when $\eta$ is not minimal, one can always find such a non-separating simple closed curve $\gamma$ in the interior of a face of $(G_0,\eta)$.
However, it is not always possible to require $\gamma$ being disjoint from the disks $D_1,\dots,D_r$.
Therefore, our main idea below is to find a non-separating simple closed curve $\gamma$ in the interior of a face of $(G_0,\eta)$ that intersects the boundary of each disk $D_i$ at most \emph{twice}.
Then we can split each $D_i$ intersected by $\gamma$ and the corresponding vortex $G_i$ into two by moving $O_h(1)$ vertices to the apex set of $G$.
After this, $\gamma$ no longer intersect any vortex disk and thus we are able to apply the above cut-and-paste argument to reduce the genus of $\varSigma$.

We begin with defining an extension of $(G_0,\eta)$ as follows.
Consider an index $i \in [r]$.
Suppose $\tau_i = (v_{i,1},\dots,v_{i,\ell_i})$.
By definition, $v_{i,1},\dots,v_{i,\ell_i}$ are the vertices of $(G_0,\eta)$ that lie on the boundary of $D_i$, sorted in clockwise or counterclockwise order.
For convenience, we write $v_{i,0} = v_{i,\ell_i}$.
We then add the edges $(v_{i,j-1},v_{i,j})$ for all $j \in [\ell_i]$ to $G_0$, and call them \emph{virtual} edges.
Furthermore, we draw these virtual edges along the boundary of the disk $D_i$ (this is possible because $v_1,\dots,v_{\ell_i}$ are sorted along the boundary of $D_i$).
The images of these virtual edges then enclose the disk $D_i$.
We do this for all indices $i \in [r]$.
Let $G_0'$ denote the resulting graph after adding the virtual edges.
Since $D_1,\dots,D_r$ are disjoint facial disks in $(G_0,\eta)$, the images of the virtual edges do not cross each other or cross the original edges in $(G_0,\eta)$.
Therefore, the drawing of the virtual edges extends $\eta$ to an embedding of $G_0'$ to $\varSigma$; for simplicity, we still use the notation $\eta$ to denote this embedding.
By construction, it is clear that $D_1,\dots,D_r$ are faces of $(G_0',\eta)$ and these faces are (homeomorphic to) disks.
The other faces of $(G_0',\eta)$ are not necessarily disks.
For convenience of illustration, we add some ``dummy'' vertices/edges in each non-disk face of $(G_0',\eta)$ to subdivide the face into disks.
Let $\varGamma$ denote the resulting graph after adding the dummy vertices/edges, and we still use $\eta$ to denote the embedding of $\varGamma$ in $\varSigma$.
Now $(\varGamma,\eta)$ is an extension of $(G_0,\eta)$.
Furthermore, every face of $(\varGamma,\eta)$ is a disk (surface embeddings with this property are known as \emph{cellular} embeddings).
Note that $D_1,\dots,D_r$ are faces of $(\varGamma,\eta)$.

Next, we shall use the well-known \emph{duality} theory for surface-embedded graphs \cite{MoharT01}.
The graph $(\varGamma,\eta)$ has a \emph{dual} graph $(\varGamma^*,\eta^*)$, which is also a $\varSigma$-embedded graph.
The vertices of $\varGamma^*$ correspond to the faces of $(\varGamma,\eta)$.
For each face $f \in F_\eta(\varGamma)$, denote by $f^* \in V(\varGamma^*)$ the vertex of $\varGamma^*$ corresponding to $f$ and call it the \emph{dual vertex} of $f$.
The edges of $\varGamma^*$ correspond to the edges of $(\varGamma,\eta)$ in the following way: every edge $e$ of $(\varGamma,\eta)$ has a dual edge $e^*$ in $\varGamma^*$ which connects the dual vertices of the two faces of $(\varGamma,\eta)$ incident to $e$.
There is a canonical embedding $\eta^*\colon\varGamma^* \rightarrow \varSigma$ in which every vertex $f^* \in V(\varGamma^*)$ is embedded in the face $f \in F_\eta(\varGamma)$ dual to $f^*$ and every edge $e^* = (f_1^*,f_2^*) \in E(\varGamma^*)$ is embedded in $f_1 \cup f_2$ such that its image only crosses (the image of) the edge $e \in E(\varGamma)$ dual to $e^*$.
It is known that the faces of $(\varGamma^*,\eta^*)$ exactly correspond to the vertices of $(\varGamma,\eta)$.
So for each vertex $v \in V(\varGamma)$, we denote by $v^* \in F_{\eta^*}(\varGamma^*)$ the face of $(\varGamma^*,\eta^*)$ corresponding to $v$ and call it the \emph{dual face} of $v$.

Since $\eta$ is not minimal, there exist $o \in F_\eta(G_0)$ and $f \in F_\eta(\partial o)$ such that $V(\partial f)$ is not contained in one connected component of $\partial o$.
We pick two vertices $s,t \in V(\partial f)$ that are contained in different connected components of $\partial o$.
Every edge in $E(\varGamma) \backslash E(\partial o)$ is embedded by $\eta$ either inside $o$ or outside $o$; for convenience, we call the ones inside $o$ \emph{in-edges} and the ones outside $o$ \emph{out-edges}.
As $s$ and $t$ belong to different connected components of $\partial o$, the edges in $E(\varGamma) \backslash E(\partial o)$ form an $(s,t)$-\emph{cut} in $\varGamma$.
Thus, there exists a minimal $(s,t)$-cut $E \subseteq E(\varGamma) \backslash E(\partial o)$, which can clearly be computed in polynomial time.
We can partition $E$ into two subsets $E_\mathsf{in}$ and $E_\mathsf{out}$, which consist of the in-edges and the out-edges, respectively.

\begin{claim}
 \label{claim:in-edge-non-empty}
 $E$ contains at least one in-edge, i.e., $E_\mathsf{in} \neq \emptyset$.
\end{claim}
\begin{claimproof}
 Let $F \subseteq F_\eta(\varGamma)$ consist of the faces of $(\varGamma,\eta)$ contained in $o$.
 The union of the boundaries of the faces in $F$ is a subgraph $B$ of $\varGamma$ containing $s$ and $t$ whose edges are either in-edges or those in $E(\partial o)$.
 To see $E_\mathsf{in} \neq \emptyset$, it suffices to show that $B$ is connected; indeed, if $B$ is connected, then $E$ has to contain at least one edge of $B$ in order to be an $(s,t)$-cut, which must be an in-edge because $E$ does not contain any edge in $E(\partial o)$.
 As all faces of $(\varGamma,\eta)$ are disks, their boundaries are connected.
 Therefore, the boundary of every face in $F$ is contained in one connected component of $B$.
 Furthermore, if two faces in $F$ are incident to a common vertex in $(\varGamma,\eta)$, their boundaries must be contained in the same connected component of $B$.
 It follows that, if $B$ is not connected, then we can partition $F$ into two subsets $F_1$ and $F_2$ such that for any $f_1 \in F_1$ and $f_2 \in F_2$, $f_1$ and $f_2$ are not incident to a common vertex in $(\varGamma,\eta)$.
 However, this is impossible because the union of the faces in $F$ is $o$, which is a connected region in $\varSigma$.
 Thus, $B$ is connected and $E_\mathsf{in} \neq \emptyset$.
\end{claimproof}

Let $E^* \subseteq E(\varGamma^*)$ be the set of edges of $\varGamma^*$ dual to the edges in $E$, and $(\varGamma_E^*,\eta^*)$ be the subgraph of $(\varGamma^*,\eta^*)$ consisting of all edges in $E^*$ and their endpoints.
Similarly, we can partition $E^*$ into two subsets $E_\mathsf{in}^*$ and $E_\mathsf{out}^*$, which are the dual of $E_\mathsf{in}$ and $E_\mathsf{out}$, respectively.
Note that (the images of) all edges in $E_\mathsf{in}^*$ are in the interior of $o$, while all edges in $E_\mathsf{out}^*$ are outside $o$.
Therefore, each connected component of $\varGamma_E^*$ cannot contain edges in both $E_\mathsf{in}^*$ and $E_\mathsf{out}^*$.
We call a connected component of $\varGamma_E^*$ an \emph{in-component} (resp., \emph{out-component}) if its edges are in $E_\mathsf{in}^*$ (resp., $E_\mathsf{out}^*$).
Since $E_\mathsf{in} \neq \emptyset$, we have $E_\mathsf{in}^* \neq \emptyset$ and thus $\varGamma_E^*$ has at least one in-component.
Furthermore, it is known that a minimal cut in a surface-embedded graph is always dual to an \emph{even subgraph} of the dual graph \cite{EricksonFN12}, that is, a subgraph in which the degree of every vertex is even.
Therefore, $\varGamma_E^*$ is an even subgraph of $\varGamma^*$, and hence every connected component of $\varGamma_E^*$ contains a simple cycle.
In particular, we can find a simple cycle $\gamma$ in an in-component of $\varGamma^*$.
The image of $\gamma$ under $\eta^*$ is a simple closed curve on $\varSigma$ contained in the interior of $o$.
With a bit abuse of notation, we also use $\gamma$ to denote this curve.

\begin{claim}
 \label{claim:gamma-non-separating}
 $\gamma$ is non-separating.
\end{claim}
\begin{claimproof}
 Assume $\gamma$ separates $\varSigma$.
 The two sides of $\gamma$ correspond to the two connected components of $\varSigma \backslash \gamma$.
 Observe that $s^*$ and $t^*$ are in the same connected component of $\varSigma \backslash \gamma$.
 Indeed, $s$ and $t$ are both on the boundary of the face $f$, and $\gamma$ is disjoint from $f$ (as it is contained in the interior of $o$).
 Thus, the images of $s$ and $t$ under $\eta$ are both contained in the connected component of $\varSigma \backslash \gamma$ containing $f$.
 Since $s$ is embedded inside $s^*$, $s^*$ must be contained in the same connected component $\varSigma \backslash \gamma$ as $s$.
 For the same reason, $t^*$ is contained in the same connected component $\varSigma \backslash \gamma$ as $t$, and thus $s^*$ and $t^*$ are in the same connected component of $\varSigma \backslash \gamma$.
 Let $\varSigma_1$ be the connected component of $\varSigma \backslash \gamma$ that contains $s^*$ and $t^*$, and $\varSigma_2$ be the other connected component of $\varSigma \backslash \gamma$.
 Every face of $(\varGamma^*,\eta^*)$ is either contained in $\varSigma_1$ or contained in $\varSigma_2$.
 Denote by $V_1 \subseteq V(\varGamma)$ and $V_2 \subseteq V(\varGamma)$ the subsets of vertices whose dual faces in $F_{\eta^*}(\varGamma^*)$ are contained $\varSigma_1$ and $\varSigma_2$, respectively.
 We have $s,t \in V_1$.
 Note that the edges of $\varGamma$ between $V_1$ and $V_2$ are exactly those dual to the edges of $\gamma$.
 But all edges of $\gamma$ are in $E_\mathsf{in}^*$ and thus are dual to the edges in $E_\mathsf{in}$.
 Therefore, $E_\mathsf{in}$ contains all edges of $\varGamma$ between $V_1$ and $V_2$.
 However, this implies $E$ is not a minimal $(s,t)$-cut in $\varGamma$.
 To see this, pick an arbitrary edge $e \in E(\varGamma)$ between $V_1$ and $V_2$.
 As argued before, $e \in E_\mathsf{in} \subseteq E$.
 We claim that $E \backslash \{e\}$ is also an $(s,t)$-cut in $\varGamma$.
 Since $E$ is an $(s,t)$-cut, if $E \backslash \{e\}$ is not an $(s,t)$-cut, then the two endpoints of $e$ must lie in the connected components of $\varGamma - E$ containing $s$ and $t$, respectively.
 But one endpoint of $e$ is in $V_2$, and the connected component $C$ of $\varGamma - E$ this endpoint lies in must be entirely contained in $V_2$, simply because $E$ contains all edges of $\varGamma$ between $V_1$ and $V_2$.
 Since $s,t \in V_1$, we know that $C$ contains neither $s$ nor $t$.
 Therefore, $E \backslash \{e\}$ is also an $(s,t)$-cut in $\varGamma$, contradicting with the minimality of $E$.
 It follows that $\gamma$ is non-separating.
\end{claimproof}

\begin{claim}
 For every $i \in [r]$, either $\gamma$ is disjoint from $D_i$ or $\gamma$ intersects (the images of) two virtual edges on the boundary of $D_i$.
\end{claim}
\begin{claimproof}
 Recall that $D_i$ is a face of $\varGamma$.
 Thus, the only edges of $(\varGamma^*,\eta^*)$ that intersect the boundary of $D_i$ are those dual to the edges in $E(\partial D_i)$, which are exactly the edges of $\varGamma^*$ incident to $D_i^*$, the dual vertex of $D_i$ in $\varGamma^*$.
 As $\gamma$ is a cycle in $\varGamma^*$, it either contains no edges incident to $D_i^*$ (when $\gamma$ does not contain $D_i^*$) or contains two edges incident to $D_i^*$ (if $\gamma$ contains $D_i^*$).
 Therefore, either $\gamma$ is disjoint from $D_i$ or $\gamma$ intersects at most two virtual edges on the boundary of $D_i$.
\end{claimproof}

In what follows, we modify the vortices of $G$ (and the corresponding disks for attaching them) to make the curve $\gamma$ disjoint from all vortex disks.
Without loss of generality, we consider the vortex $G_1$ attached in $D_1$.
If $\gamma$ is disjoint from $D_1$, we do not change $G_1$ and $D_1$.
Otherwise, $\gamma$ intersects two edges on the boundary of $D_1$.
Recall that $(\tau_1,\mathcal{P}_1)$ is the witness pair of $G_1$, where $\tau_1 = (v_{1,1},\dots,v_{1,\ell_1})$.
Let $(u_1,\dots,u_{\ell_1})$ be the underlying path of the path decomposition $\mathcal{P}_1$ of $G_1$.
We use $\beta(u_i)$ to denote the bag of $u_i$ in $\mathcal{P}_1$.
By definition, we have $v_{1,i} \in \beta(u_i)$ for all $i \in [\ell_1]$.
Suppose $\gamma$ intersects the edges $(v_{1,i-1},v_{1,i})$ and $(v_{1,j-1},v_{1,j})$, where $i<j$.
Then $\gamma$ splits $D_1$ into two parts, one contains $v_{1,i},\dots,v_{1,j-1}$ and the other contains $v_{1,j},\dots,v_{1,\ell_1},v_{1,1},\dots,v_{1,i-1}$.
Inside $D_1$, we draw a curve $\gamma'$ connecting $v_{1,i}$ and $v_{1,j-1}$ that is disjoint from $\gamma$; this is possible because $v_{1,i}$ and $v_{1,j-1}$ are on the same side of $\gamma$.
Similarly, we draw a curve $\gamma''$ inside $D_1$ connecting $v_{1,j}$ and $v_{1,i-1}$ that is disjoint from $\gamma$.
The portion of the boundary of $D_1$ from $v_{1,i}$ to $v_{1,j-1}$ together with $\gamma'$ encloses a disk $D_1' \subseteq D_1$, and the vertices $v_{1,i},\dots,v_{1,j-1}$ lie on the boundary of $D_1'$ in clockwise (or counterclockwise) order.
Similarly, the portion of the boundary of $D_1$ from $v_{1,j}$ to $v_{1,i-1}$ together with $\gamma''$ encloses a disk $D_1'' \subseteq D_1$, and the vertices $v_{1,j},\dots,v_{1,\ell_1},v_{1,1},\dots,v_{1,i-1}$ lie on the boundary of $D_1''$ in clockwise (or counterclockwise) order.
Next, we split the vortex $G_1$ into two vortices $G_1'$ and $G_1''$, attached in $D_1'$ and $D_1''$, respectively.
We first construct $G_1'$, and $G_1''$ can be constructed in the same way.
Set $X = \beta(u_i) \cup \beta(u_{j-1})$.
Note that $|X| \leq 2h+2$.
We move the vertices in $X$ to the apex set of $G$, and define $G_1'$ as the subgraph of $G_1$ induced by the vertices in $(\bigcup_{k=i}^{j-1} \beta(u_k)) \backslash X$.
In order to make $G_1'$ a vortex attached to $D_1'$, we need a witness pair $(\sigma_1',\mathcal{P}_1')$ for $G_1'$.
As $\mathcal{P}_1$ is a path decomposition, if a vertex appears in $\bigcup_{k=i}^{j-1} \beta(u_k)$ and also appears in one of the bags $\beta(u_j),\dots,\beta(u_{\ell_1}),\beta(u_1),\dots,\beta(u_{i-1})$, then it must be contained in $X$.
Thus, no vertex of $G_1'$ can be contained in the bags $\beta(u_j),\dots,\beta(u_{\ell_1}),\beta(u_1),\dots,\beta(u_{i-1})$.
So the path $(u_i,\dots,u_{j-1})$ with the bags $\beta(u_i) \backslash X,\dots,\beta(u_{j-1}) \backslash X$ gives a path decomposition of $G_1'$.
However, we cannot directly use this as the path decomposition $\mathcal{P}_1'$ in the witness pair, because some vertices in $\{v_{1,i},\dots,v_{1,j-1}\}$ might be ``missing'', i.e., they might be contained in $X$ and are hence moved to the apex set.
This issue can be handled by merging consecutive bags in the path decomposition as follows.
Let $K = \{k \mid i \leq k \leq j-1 \text{ and } v_{1,k} \notin X\}$ and suppose $K = \{k_1,\dots,k_t\}$ where $k_1 < \cdots < k_t$.
For convenience, set $k_0 = i-1$ and $k_{t+1} = j$.
Since $|X| \leq 2h$, all but at most $2h$ indices in $\{i,\dots,j-1\}$ are contained in $K$, which implies $|k_\alpha - k_{\alpha-1}| \leq 2h+1$ for all $\alpha \in [t+1]$.
We use $(u_{k_1},\dots,u_{k_t})$ as the path of $\mathcal{P}_1'$, and define the bag $\beta'(u_{k_\alpha})$ of $u_{k_\alpha}$ as the union of $\beta(u_{k_{\alpha-1}+1}),\dots,\beta(u_{k_{\alpha+1}-1})$ excluding the vertices in $X$.
One can easily verify that $\mathcal{P}_1'$ is a path decomposition of $G_1'$.
Furthermore, the size of each bag of $\mathcal{P}_1'$ is $O(h^2)$.
Now set $\sigma_1' = (v_{1,k_1},\dots,v_{1,k_t})$.
Then $(\sigma_1',\mathcal{P}_1')$ is a valid witness pair for $G_1'$, as $v_{1,k_\alpha} \in \beta'(u_{k_\alpha})$ for all $\alpha \in [t]$ and $v_{1,k_1},\dots,v_{1,k_t}$ are sorted in clockwise or counterclockwise order around the boundary of $D_1'$.
Similarly, we can construct the vortex $G_1''$, and attach it in the disk $D_1''$ with a witness pair $(\sigma_1'',\mathcal{P}_1'')$.
The two new disks $D_1'$ and $D_1''$ are disjoint from $\gamma$.
We do this for every vortex disk $D_i$.
After this, $\gamma$ is disjoint from all vortex disks.
Note that the total number of the new vortices is at most $2h$, and the width of the path decomposition of the new vortices is $O(h^2)$.
Also, the total number of vertices moved to the apex set is $O(h^2)$.
Thus, the modified almost-embeddable structure of $G$ is an $O_h(1)$-almost-embeddable structure.
As $\gamma$ is in the interior of $o$ and is now disjoint from all vortex disks, we can apply the above cut-and-paste operation along $\gamma$ to obtain an $O_h(1)$-almost-embeddable structure of $G$ whose underlying surface is of genus strictly smaller than $\varSigma$, which completes the proof.
\end{proof}

\subsection{A technical lemma}
In this section, we prove the following key lemma, which serves as a technical core in our proof.

\begin{lemma} \label{lem-key}
 Let $(G,\eta)$ be a $(\varSigma,x_0)$-embedded graph, and $L_1,\dots,L_m$ be the radial layering of $(G,\eta)$.
 Given any $t \in [m]$ with $t \geq 2$ and a set $\varPhi \subseteq L_t$ of size $c$, if all faces of $(G,\eta)$ incident to $L_t$ are non-singular, then one can compute in polynomial time a subset $X \subseteq L_t$ containing $\varPhi$ and another subset $L_t^+ \subseteq L_t \backslash N_G(L_{t-1})$ satisfying the following four conditions (where $o_t$ is the outer face of $(G[L_{\geq t}],\eta)$ and $g$ is the genus of $\varSigma$).
 \begin{enumerate}[label = (\arabic*)]
  \item\label{item:rcd-key-1} Every connected component of $G[X]$ intersects $\varPhi$ and is neighboring to $L_{t-1}$.
  \item\label{item:rcd-key-2} For any outer-preserving extension $(G',\eta')$ of $(\partial o_t,\eta)$, every connected component $C$ of $G'-(L_t \backslash L_t^+)$ satisfies $N_{G'}(C) \subseteq V(\partial f)$ for some $f \in F_\eta(\partial o_t) \backslash \{o_t\}$.
  \item\label{item:rcd-key-3} $V(\partial f) \cap (X \backslash L_t^+) = O_{g,c}(1)$ for all $f \in F_\eta(\partial o_t) \backslash \{o_t\}$.
  \item\label{item:rcd-key-4} $\partial f - L_t^+$ only has $O_{g,c}(1)$ connected components for all $f \in F_\eta(\partial o_t) \backslash \{o_t\}$.
 \end{enumerate}
\end{lemma}

Let $(G,\eta)$ be a connected $(\varSigma,x_0)$-embedded graph satisfying the properties of the lemma, where $\varSigma$ is a surface of genus $g$.
Let $L_1,\dots,L_m$ be the radial layering of $(G,\eta)$.
Fix an index $t \in [m]$ with $t \geq 2$, and assume all faces of $(G,\eta)$ incident to $L_t$ are non-singular.
Denote by $o_t$ the outer face of $(G[L_{\geq t}],\eta)$.

We say a vertex in $L_t = V(\partial o_t)$ is an \textit{exit vertex} if it is adjacent to a vertex in $L_{t-1}$.
In the graph $(\partial o_t,\eta)$, each edge $e$ is incident to two (not necessarily different) faces in $F_\eta(\partial o_t)$, which correspond to the two ``sides'' of (the embedding of) $e$.
Note that one of these two faces must be $o_t$ because all edges in $(\partial o_t,\eta)$ are incident to $o_t$, while the other one can be either $o_t$ or an inner face of $(\partial o_t,\eta)$, which we denote by $f(e)$.

Let $\mathcal{P}$ be the set of all pairs $(f,v)$ where $f \in F_\eta(\partial o_t) \backslash \{o_t\}$ is an inner face of $(\partial o_t, \eta)$ and $v \in V(\partial f)$.
In what follows, we classify the pairs in $\mathcal{P}$ into three classes: \textit{singular pairs}, \textit{critical pairs}, and \textit{normal pairs}.
A pair $(f,v) \in \mathcal{P}$ is \textit{singular} if (at least) one of the following conditions holds.
\begin{enumerate}[label = (\roman*)]
 \item There is a path in $\partial o_t$ from $v$ to another vertex of $\partial f$ that does not use any edge of $\partial f$, i.e., the path only uses the edges in $E(\partial o_t) \backslash E(\partial f)$.
 \item There is a path in $\partial o_t$ from $v$ to a vertex of $\partial f'$ for some singular face $f' \in F_\eta(\partial o_t) \backslash \{o_t,f\}$ that only uses the edges in $E(\partial o_t) \backslash E(\partial f)$.
\end{enumerate}
A pair $(f,v) \in \mathcal{P}$ is \textit{critical} if it is not singular and there exists a path in $\partial o_t$ from $v$ to an exit vertex that only uses the edges in $E(\partial o_t) \backslash E(\partial f)$.
In particular, if $(f,v)$ is not singular and $v$ is an exit, then $(f,v)$ is critical.
Finally, a pair in $\mathcal{P}$ is \textit{normal} if it is neither singular nor critical.
We first observe that there are only a constant number of singular pairs in $\mathcal{P}$.

\begin{observation} \label{obs-O1singular}
For each inner face $f \in F_\eta(\partial o_t)$, there are only $O_g(1)$ vertices $v \in V(\partial f)$ such that the pair $(f,v) \in \mathcal{P}$ is singular.
\end{observation}
\begin{proof}
We say a singular pair $(f,v) \in \mathcal{P}$ is of \textit{type-(i)} if it satisfies condition (i) above, and of \textit{type-(ii)} if it does not satisfy condition (i) but satisfies condition (ii).
Fix an inner face $f \in F_\eta(\partial o_t)$.

\paragraph{Bounding type-(i) singular pairs.}
We first bound the number of vertices $v \in V(\partial f)$ such that $(f,v)$ is a type-(i) singular pair.
By definition, for each such vertex $v$, there exists a \textit{witness} path in $\partial o_t$ from $v$ to another vertex of $\partial f$ that only uses the edges in $E(\partial o_t) \backslash E(\partial f)$.
The union $U$ of all witness paths is a subgraph of $\partial o_t$.
Each connected component of $U$ contains at least two vertices in $V(\partial f)$, because it contains at least one witness path whose two endpoints are different vertices in $V(\partial f)$.
In each connected component $C$ of $U$, we take a Steiner tree whose terminals are the vertices in $V(\partial f)$ contained in $C$, that is, a tree in $C$ containing all terminals in which all leaves are terminals.
Let $T$ be the forest that consists of these Steiner trees.
Now consider the graph $(\partial f \cup T, \eta)$.
We prove the following three properties of this graph.
\begin{itemize}
 \item $(\partial f \cup T, \eta)$ has only two faces.
 \item All vertices in $\partial f \cup T$ are of degree at least $2$.
 \item For every type-(i) singular pair $(f,v)$, $v$ is of degree at least $3$ in $\partial f \cup T$.
\end{itemize}

To bound the number of faces of $(\partial f \cup T, \eta)$, we notice that $\partial f \cup T$ is a subgraph of $\partial o_t$, and thus every face of $(\partial o_t,\eta)$ is contained in exactly one face of $(\partial f \cup T,\eta)$.
Let $f^* \in F_\eta(\partial f \cup T)$ be the face that contains $o_t$.
Besides, $f$ itself is also a face of $(\partial f \cup T,\eta)$ as all edges of $\partial f$ are also edges of $\partial f \cup T$.
We show that $F_\eta(\partial f \cup T) = \{f^*,f\}$.
The graph $(T,\eta)$ only has one face which is the entire surface $\varSigma$, because a forest embedded in a surface does not cut the surface.
Therefore, for all $f' \in F_\eta(\partial o_t) \backslash \{o_t,f\}$, we have $E(\partial f') \nsubseteq E(T)$, for otherwise $f'$ is a face of $(T,\eta)$.
Also note that $E(\partial f') \cap E(\partial f) = \emptyset$ for all $f' \in F_\eta(\partial o_t) \backslash \{o_t,f\}$, because each edge in $E(\partial f')$ must have one side incident to $f'$ and the other side incident to $o_t$, and is thus not incident to $f$.
It follows that $E(\partial f') \nsubseteq E(\partial f \cup T)$ for all $f' \in F_\eta(\partial o_t) \backslash \{o_t,f\}$, and thus there exists an edge $e_{f'} \in E(\partial f')$ that is not in $\partial f \cup T$.
Because of the absence of $e_{f'}$ in $\partial f \cup T$, $f'$ and $o_t$ are contained in the same face in $(\partial f \cup T,\eta)$, which is $f^*$.
In other words, $f^*$ contains all faces in $F_\eta(\partial o_t) \backslash \{f\}$, which implies $f^* \cup f = \varSigma$ and hence $F_\eta(\partial f \cup T) = \{f^*,f\}$.

To see all vertices in $\partial f \cup T$ are of degree at least $2$, we first notice that all vertices in $\partial f$ are of degree at least $2$.
Indeed, if a vertex in $(\partial f,\eta)$ is of degree 0 or 1, then it can only be incident to one face, while all vertices in $(\partial f,\eta)$ are incident to both $f$ and $o_t$.
Then it suffices to consider the vertices in $T$ but not in $\partial f$.
Recall that the leaves of the trees in $T$ are all vertices in $\partial f$.
Thus, the vertices in $T$ but not in $\partial f$ are of degree at least $2$ in $T$ (and thus in $\partial f \cup T$).

Finally, let us consider the degree of a vertex $v$ in $\partial f \cup T$ where $(f,v)$ is a type-(i) singular pair.
Consider a type-(i) singular pair $(f,v)$.
As argued above, the degree of $v$ in $\partial f$ is at least $2$.
On the other hand, the degree of $v$ in $T$ is at least $1$, because $U$ contains the witness path for $v$ and thus $v$ is the terminal of one tree in $T$.
Note that $\partial f$ and $T$ do not have common edges.
Therefore, the degree of $v$ in $\partial f \cup T$ is at least $3$.

Using the three properties of $(\partial f \cup T, \eta)$, we are ready to bound the number of type-(i) singular pairs.
We denote by $\gamma$ the number of vertices $v \in V(\partial f)$ such that $(f,v)$ is a type-(i) singular pair.
Let $\#_V,\#_E,\#_F,\#_C$ be the numbers of vertices, edges, faces, and connected components of $(\partial f \cup T, \eta)$.
We have $\#_F = 2$ as $(\partial f \cup T, \eta)$ has two faces.
Also, the last two properties of $(\partial f \cup T, \eta)$ imply $2\#_E \geq 2\#_V + \gamma$, or equivalently, $\#_V - \#_E \leq -\gamma/2$.
Because $\#_C \geq 0$, we have $\#_V-\#_E+\#_F-\#_C \leq 2-\gamma/2$.
By Euler's formula, $|\#_V-\#_E+\#_F-\#_C| = O(g)$.
So we immediately have $\gamma = O(g)$.

\paragraph{Bounding type-(ii) singular pairs.}
Next, we bound the number of vertices $v \in V(\partial f)$ such that $(f,v)$ is a type-(ii) singular pair.
By definition, for each such vertex $v$, there exists a \textit{witness} path in $\partial o_t$ from $v$ to a vertex on the boundary of a singular face $f' \in F_\eta(\partial o_t) \backslash \{o_t,f\}$ that only uses the edges in $E(\partial o_t) \backslash E(\partial f)$; we call $f'$ the \textit{witness face} of $v$ for convenience.
By Lemma~\ref{lem-degree}, $(\partial o_t,\eta)$ has only $O(g)$ singular faces.
We claim that each singular face can only witness $O(g)$ vertices, which implies the number of type-(ii) singular pairs $(f,v)$ is $O(g^2)$.
Consider a singular face $f' \in F_\eta(\partial o_t) \backslash \{o_t,f\}$.
Let $(f,v)$ and $(f,v')$ be two type-(ii) singular pairs whose witness face is $f'$.
Denote by $\pi_v$ (resp., $\pi_{v'}$) the witness path of $v$ (resp., $v'$) and by $u$ (resp., $u'$) the endpoint of $\pi_v$ (resp., $\pi_{v'}$) other than $v$ (resp., $v'$).
Observe that $u$ and $u'$ must lie in different connected components of $\partial f'$.
Indeed, if $u$ and $u'$ lie in  the same connected components of $\partial f'$, then there exists a path $\pi$ from $u$ to $u'$ in $\partial f'$.
By concatenating $\pi_v,\pi,\pi_{v'}$, we obtain a path from $v$ to $v'$ in $\partial o_t$ that only uses the edges of $E(\partial o_t) \backslash E(\partial f)$.
This implies $(f,v)$ and $(f,v')$ are type-(i) singular pairs, which contradicts with our assumption.
Therefore, $u$ and $u'$ must lie in different connected components of $\partial f'$.
In other words, there cannot be two witness paths ending in the same connected component of $\partial f'$.
So the number of vertices with witness face $f'$ is bounded by the number of connected components of $\partial f'$.
By \ref{item:degree-1} of Lemma~\ref{lem-degree}, $\partial f'$ can only have $O(g)$ connected components.
Thus, $f'$ can witness at most $O(g)$ vertices and the number of type-(ii) singular pairs $(f,v)$ is $O(g^2)$.
\end{proof}

Next, we define a notion called \textit{legal paths}.
Consider a path $\pi = (v_0,v_1,\dots,v_r)$ in $\partial o_t$ and write $e_i = v_{i-1}v_i$ for $i \in [r]$.
We say $\pi$ is \textit{legal} if for all $i \in [r-1]$ such that $(f(e_i),v_i)$ is a critical pair in $\mathcal{P}$, we have $f(e_{i+1}) \neq f(e_i)$.
We have the following observation.

\begin{observation}\label{obs-legal}
Let $\pi$ be a legal path in $\partial o_t$ and denote by $V_\pi \subseteq V(\partial o_t)$ the vertices on $\pi$.
Then for each face $f \in F_\eta(\partial o_t) \backslash \{o_t\}$, the graph $\partial f [V_\pi \cap V(\partial f)]$ has $O_g(1)$ connected components and there are at most two vertices $v \in V_\pi \cap V(\partial f)$ such that $(f,v)$ is a critical pair.
\end{observation}

\begin{proof}
Let $\pi = (v_0,v_1,\dots,v_r)$ be a legal path in $\partial o_t$ and write $e_i = v_{i-1}v_i$ for $i \in [r]$.
Then $V_\pi = \{v_0,v_1,\dots,v_r\}$.
Suppose $\partial f [V_\pi \cap V(\partial f)]$ has $k$ connected components $C_1,\dots,C_k$.
For each $i \in [k]$, define $\alpha_i$ as the largest index $\alpha \in \{0\} \cup [r]$ such that $v_\alpha \in C_i$.
By definition, the vertices $v_{\alpha_1},\dots,v_{\alpha_k}$ are distinct and we can assume $\alpha_1 < \cdots < \alpha_k$ without loss of generality.
We observe that $e_{\alpha_i+1} \notin E(\partial f)$ for all $i \in [k]$ such that $\alpha_i < r$.
Indeed, if $e_{\alpha_i+1} \in E(\partial f)$, then $v_{\alpha_i+1}$ is in the same connected component as $v_{\alpha_i}$ in $\partial f [V_\pi \cap V(\partial f)]$, i.e., $v_{\alpha_i+1} \in C_i$, which contradicts the definition of $\alpha_i$.
Based on this observation, we further claim that $(f,v_{\alpha_i})$ is a singular pair for any $i<k$.
To see this, let $\beta > \alpha_i$ be the smallest index such that $v_\beta \in V(\partial f)$; note that such an index $\beta$ always exists because $\alpha_k > \alpha_i$ and $v_{\alpha_k} \in V(\partial f)$.
Now the subpath $(v_{\alpha_i},\dots,v_\beta)$ of $\pi$ is from $v_{\alpha_i}$ to another vertex in $V(\partial f)$, i.e., $v_\beta$, and only consists of the edges in $E(\partial o_t) \backslash E(\partial f)$ since $v_{\alpha_i+1},\dots,v_{\beta-1} \notin V(\partial f)$.
Therefore, $(f,v_{\alpha_i})$ is singular.
However, by Observation~\ref{obs-O1singular}, there are only $O_g(1)$ vertices $v \in V(\partial f)$ such that $(f,v)$ is singular.
As the vertices $v_{\alpha_1},\dots,v_{\alpha_k}$ are distinct, we immediately have $k = O_g(1)$.

Next, we bound the number of vertices $v \in V_\pi \cap V(\partial f)$ such that $(f,v)$ is a critical pair.
Assume there are three such vertices $v_{i^-},v_i,v_{i^+} \in V_\pi \cap V(\partial f)$ where $i^- < i < i^+$.
We observe that either $f(e_i) \neq f$ or $f(e_{i+1}) \neq f$.
Indeed, if $f(e_i) = f$, then $(f(e_i),v_i) = (f,v_i)$ is critical, which implies $f(e_{i+1}) \neq f(e_i) = f$ as $\pi$ is legal.
Without loss of generality, assume $f(e_{i+1}) \neq f$.
We use the same argument as above to show that $(f,v_i)$ is a singular pair.
Let $j > i$ be the smallest index such that $v_j \in V(\partial f)$; note that such an index $j$ always exists because $i^+ > i$ and $v_{i^+} \in V(\partial f)$.
The subpath $(v_i,\dots,v_j)$ of $\pi$ is from $v_i$ to another vertex in $V(\partial f)$, i.e., $v_j$, and only consists of the edges in $E(\partial o_t) \backslash E(\partial f)$ as $v_{i+1},\dots,v_{j-1} \notin V(\partial f)$.
Thus, $(f,v_i)$ is a singular pair, which contradicts the fact that $(f,v_i)$ is critical.
As a result, there are at most two vertices $v \in V_\pi \cap V(\partial f)$ such that $(f,v)$ is a critical pair.
\end{proof}

\begin{observation}\label{obs-wbpath}
For any vertex $v \in L_t$, there exists a legal path in $\partial o_t$ from $v$ to an exit vertex of $L_t$, which can be computed in polynomial time in $|L_t|$.
\end{observation}
\begin{proof}
For convenience, we say a path $\pi = (v_0,v_1,\dots,v_r)$ in $\partial o_t$ is $k$-legal if $f(e_{i+1}) \neq f(e_i)$ for all $i \in [\min{(k,r-1)}]$ such that $(f(e_i),v_i)$ is a critical pair in $\mathcal{P}$, where $e_i = v_{i-1},v_i$.
Let $v \in L_t$.
We show how to compute a simple $(k+1)$-legal path in $\partial o_t$ from $v$ to an exit vertex of $L_t$, given a simple $k$-legal path in $\partial o_t$ from $v$ to an exit vertex of $L_t$.
Let $\pi = (v_0,v_1,\dots,v_r)$ be the given $k$-legal path where $v_0 = v$ and $v_r$ is an exit vertex.
Write $e_i = v_{i-1}v_i$ for $i \in [r]$.
If $k \geq r-1$, we are done, because $\pi$ is already legal and thus $(k+1)$-legal.
Assume $k < r-1$.
If $(f(e_k),v_k)$ is not a critical pair, then $\pi$ is $(k+1)$-legal.
So we only need to consider the case where $(f(e_k),v_k)$ is a critical pair.
Let $f = f(e_k)$.
As $(f,v_k)$ is critical, there exists a path $\pi'$ from $v_k$ to an exit vertex that only uses the edges in $E(\partial o_t) \backslash E(\partial f)$.
Without loss of generality, we can assume $\pi'$ is simple, and such a path can be easily computed in polynomial time.
We claim that the concatenation $\pi^*$ of $(v_0,v_1,\dots,v_k)$ and $\pi'$ is a simple $(k+1)$-legal path.
Clearly, $\pi^*$ is $k$-legal as its first half is $(v_0,v_1,\dots,v_k)$.
Let $e$ be the first edge of $\pi'$.
We have $e \in E(\partial o_t) \backslash E(\partial f)$ and thus $f(e) \neq f = f(e_k)$.
Therefore, $\pi^*$ is $(k+1)$-legal.
To see $\pi^*$ is simple, we need to show that $\pi'$ does not visit $v_0,v_1,\dots,v_{k-1}$.
It suffices to show that $v_0,v_1,\dots,v_{k-1}$ are not reachable from $v_k$ using only the edges in $E(\partial o_t) \backslash E(\partial f)$.
Assume this is not the case.
Then there exists a largest index $i \in [k-1]$ such that $v_i$ is reachable from $v_k$ using only the edges in $E(\partial o_t) \backslash E(\partial f)$.
Note that $v_i \notin V(\partial f)$, for otherwise $(f,v_k)$ is a singular pair.
This implies $i < k-1$, because $v_{k-1} \in V(\partial f)$.
By the choice of $i$, we must have $e_{i+1} \in E(\partial f)$ and thus $v_i \in V(\partial f)$, contradicting with the fact $v_i \notin V(\partial f)$.
Thus, $\pi^*$ is a simple $(k+1)$-legal path in $\partial o_t$ from $v$ to an exit vertex of $L_t$.
Now we have seen given a simple $k$-legal path, how to compute a simple $(k+1)$-legal path with the same endpoints.
Iteratively doing this $|L_t|$ times gives us a simple legal path from $v$ to an exit vertex.
\end{proof}

Let $\varPhi \subseteq L_t$ be the given set of size $c$ as in Lemma~\ref{lem-key}.
Suppose $\varPhi = \{v_1,\dots,v_c\}$.
For each $i \in [c]$, we use Observation~\ref{obs-wbpath} to compute a legal path $\pi_i$ in $\partial o_t$ from $v_i$ to an exit vertex of $L_t$.
Then we define $X \subseteq L_t$ as the set of all vertices on the paths $\pi_1,\dots,\pi_c$.
Clearly, $v_1,\dots,v_c \in X$ and every connected component of $G[X]$ is neighboring to $L_{t-1}$ because each of the paths $\pi_1,\dots,\pi_c$ contains an exit vertex of $L_t$.
In particular, condition \ref{item:rcd-key-1} of Lemma~\ref{lem-key} holds.
We then define the set $L_t^+$ as follows.
For every pair $(f,v) \in \mathcal{P}$, we define $Y_{f,v} \subseteq L_t$ as the set of vertices that is reachable from $v$ using only the edges in $E(\partial o_t) \backslash E(\partial f)$.
Then $L_t^+$ is defined by the union of all $Y_{f,v}$ where $(f,v)$ is a normal pair and $v \in X$.
By definition, if $(f,v)$ is a normal pair, then none of the vertices in $Y_{f,v}$ is an exit vertex, and therefore $L_t^+ \subseteq L_t \backslash N_G(L_{t-1})$.
We need to verify that conditions \ref{item:rcd-key-2}-\ref{item:rcd-key-4} of Lemma~\ref{lem-key} are satisfied.

\paragraph{Verifying Condition \ref{item:rcd-key-2} of Lemma~\ref{lem-key}.}
To see condition \ref{item:rcd-key-2}, we first observe some basic properties of the sets $Y_{f,v}$ defined above for normal pairs $(f,v) \in \mathcal{P}$.

\begin{observation} \label{obs-normalprop}
Let $(f,v) \in \mathcal{P}$ be a normal pair.
Then the following properties hold.
\begin{enumerate}[label = (\alph*)]
 \item\label{item:normal-properties-1} $Y_{f,v}$ does not contain any exit vertex of $L_t$.
 \item\label{item:normal-properties-2} $Y_{f,v} \cap V(\partial f) = \{v\}$.
 \item\label{item:normal-properties-3} For all $f' \in F_\eta(o_t) \backslash \{o_t,f\}$, either $Y_{f,v} \cap V(\partial f') = \emptyset$ or $V(\partial f') \subseteq Y_{f,v}$.
 \item\label{item:normal-properties-4} For any outer-preserving extension $(G',\eta')$ of $(\partial o_t,\eta)$, we have $N_{G'}(Y_{f,v}) \cap L_t \subseteq V(\partial f)$ and $N_{G'}(Y_{f,v} \backslash \{v\}) \cap L_t = \{v\}$.
\end{enumerate}
\end{observation}

\begin{proof}
To see \ref{item:normal-properties-1}, assume $Y_{f,v}$ contains an exit vertex of $L_t$.
Then there exists a path from $v$ to an exit vertex that only uses the edges in $E(\partial o_t) \backslash E(\partial f)$.
In this case, if $(f,v)$ is not singular, then $(f,v)$ is critical by definition.
This contradicts with the fact that $(f,v)$ is normal.

To see \ref{item:normal-properties-2}, assume $Y_{f,v}$ contains a vertex in $V(\partial f)$ other than $v$.
Then there exists a path from $v$ to another vertex in $V(\partial f)$ that only uses the edges in $E(\partial o_t) \backslash E(\partial f)$, which implies $(f,v)$ is singular and contradicts with the fact that $(f,v)$ is normal.

To see \ref{item:normal-properties-3}, consider a face $f' \in F_\eta(o_t) \backslash \{o_t,f\}$ such that $Y_{f,v} \cap V(\partial f') \neq \emptyset$.
Note that $f'$ cannot be a singular face.
Indeed, there exists a path in from $v$ to a vertex in $V(\partial f')$ that only uses the edges in $E(\partial o_t) \backslash E(\partial f)$.
If $f'$ is a singular face, then $(f,v)$ is singular, contradicting with the fact that $(f,v)$ is normal.
Therefore, $f'$ is not singular, i.e., $\partial f'$ is connected.
It follows that every vertex in $V(\partial f')$ is reachable from $v$ using only the edges in $E(\partial o_t) \backslash E(\partial f)$ and hence $V(\partial f') \subseteq Y_{f,v}$.

Finally, we show \ref{item:normal-properties-4} holds.
Let $(G',\eta')$ be an outer-preserving extension of $(\partial o_t,\eta)$.
Consider a vertex $u \in N_{G'}(Y_{f,v}) \cap L_t$ and assume $u \notin V(\partial f)$.
Let $u' \in Y_{f,v}$ be a neighbor to $u$ in $G'$.
Since $u \notin V(\partial f)$, $u'u \notin E(\partial f)$.
If $u'u \in E(\partial o_t)$, then $u$ is reachable from $v$ using only the edges in $E(\partial o_t) \backslash E(\partial f)$ (as we can reach $u'$ from $v$ and go via the edge $(u',u)$ to reach $u$) and thus $u \in Y_{f,v}$, which contradicts with the fact $u \in N_G(Y_{f,v})$.
Thus, $u'u \notin E(\partial o_t)$.
By the definition of outer-preserving extension, the image of $u'u$ under $\eta'$ lies in an inner face $f' \in F_\eta(\partial o_t) \backslash \{o_t\}$, and thus the images of $u',u$ lie in $f'$.
It follows that $u',u \in V(\partial f')$ as $u',u \in L_t$.
Note that $f' \neq f$ since $u \notin V(\partial f)$ by our assumption.
Then by \ref{item:normal-properties-3} shown above, either $Y_{f,v} \cap V(\partial f') = \emptyset$ or $V(\partial f') \subseteq Y_{f,v}$.
However, we have $Y_{f,v} \cap V(\partial f') \neq \emptyset$ as $u' \in Y_{f,v} \cap V(\partial f')$, and have $V(\partial f') \nsubseteq Y_{f,v}$ as $u \notin Y_{f,v}$.
This contradiction proves $u \in V(\partial f)$.

To see $N_{G'}(Y_{f,v} \backslash \{v\}) \cap L_t = \{v\}$, consider a vertex $u \in N_{G'}(Y_{f,v} \backslash \{v\}) \cap L_t$ and assume $u \neq v$.
Since $N_{G'}(Y_{f,v}) \cap L_t \subseteq V(\partial f)$, we have $u \in V(\partial f)$.
Let $u' \in Y_{f,v} \backslash \{v\}$ be a neighbor of $u$ in $G'$.
By \ref{item:normal-properties-2} shown above, $Y_{f,v} \cap V(\partial f) = \{v\}$ and thus $u' \notin V(\partial f)$.
So $u'u \notin E(\partial f)$.
Then we can apply the same argument as above.
If $u'u \in E(\partial o_t)$, then $u$ is reachable from $v$ using only the edges in $E(\partial o_t) \backslash E(\partial f)$ (as we can reach $u'$ from $v$ and go via the edge $(u',u)$ to reach $u$) and thus $u \in Y_{f,v}$, which contradicts with the fact $Y_{f,v} \cap V(\partial f) = \{v\}$.
Thus, $u'u \notin E(\partial o_t)$, and the image of $u'u$ under $\eta$ lies in an inner face $f' \in F_\eta(\partial o_t) \backslash \{o_t\}$.
It follows that $u',u \in V(\partial f')$.
Note that $f' \neq f$ as $u' \notin V(\partial f)$.
Then by \ref{item:normal-properties-3} shown above, either $Y_{f,v} \cap V(\partial f') = \emptyset$ or $V(\partial f') \subseteq Y_{f,v}$.
However, we have $Y_{f,v} \cap V(\partial f') \neq \emptyset$ as $u' \in Y_{f,v} \cap V(\partial f')$, and have $V(\partial f') \nsubseteq Y_{f,v}$ as $u \notin Y_{f,v}$.
This contradiction proves $u = v$ and thus $N_{G'}(Y_{f,v} \backslash \{v\}) \cap L_t \subseteq \{v\}$.
To further see $N_{G'}(Y_{f,v} \backslash \{v\}) \cap L_t = \{v\}$, we simply observe that $\partial o_t[Y_{f,v}]$ is connected (by the construction of $Y_{f,v}$).
Therefore, in $\partial o_t$ (and thus in $G'$), $v$ is adjacent to some vertex in $Y_{f,v} \backslash \{v\}$.
\end{proof}

Let $(G',\eta')$ be an outer-preserving extension of $(\partial o_t,\eta)$.
For every connected component $C$ of $G' - (L_t \backslash L_t^+)$, we have $N_{G'}(C) \subseteq L_t \backslash L_t^+$.
We claim that $N_{G'}(C) \subseteq \bigcup_{f \in F_\eta(\partial o_t) \backslash \{o_t\}} V(\partial f)$.
As $N_{G'}(C) \subseteq L_t$ and $C \subseteq L_t^+ \cup (V(G') \backslash L_t)$, we have $N_{G'}(C) \subseteq (N_{G'}(L_t^+) \cap L_t) \cup (N_{G'}(V(G') \backslash L_t) \cap L_t)$.
Recall that $L_t^+$ is the union of all $Y_{f,v}$ where $(f,v)$ is a normal pair.
Thus, by Observation~\ref{obs-normalprop}\ref{item:normal-properties-4}, we have $N_{G'}(L_t^+) \cap L_t \subseteq \bigcup_{f \in F_\eta(\partial o_t) \backslash \{o_t\}} V(\partial f)$.
For any vertex $v \in V(G') \backslash L_t$, the image of $v$ under $\eta'$ lies in the interior of some face $f \in F_\eta(\partial o_t) \backslash \{o_t\}$, and thus $N_{G'}(\{v\}) \cap L_t \subseteq V(\partial f)$.
It then follows that $N_{G'}(C) \subseteq \bigcup_{f \in F_\eta(\partial o_t) \backslash \{o_t\}} V(\partial f)$.

Based on this, we further show that $N_{G'}(C) \subseteq V(\partial f)$ for some $f \in F_\eta(\partial o_t) \backslash \{o_t\}$.
If $N_G(C) = \emptyset$, we are done.
Also, if $C \cap L_t^+ = \emptyset$, then the images of all vertices in $C$ under $\eta'$ must lie in one face $f \in F_\eta(\partial o_t) \backslash \{o_t\}$ and thus $N_G(C) \subseteq V(\partial f)$.
So suppose $N_G(C) \neq \emptyset$ and $C \cap L_t^+ \neq \emptyset$.
We observe the following.

\begin{observation}
There exists a face $f^* \in F_\eta(\partial o_t) \backslash \{o_t\}$ such that $V(\partial f^*) \nsubseteq C$ and $V(\partial f^*) \cap C \neq \emptyset$.
\end{observation}
\begin{proof}
Let $u \in N_{G'}(C)$ be a vertex and $u' \in C$ be a neighbor of $u$ in $G'$.
If $u' \in L_t^+$, then $u' \in Y_{f^*,v}$ for some normal pair $(f^*,v)$.
We have $Y_{f^*,v} \subseteq C$, since $Y_{f^*,v} \subseteq L_t^+$ and $G'[Y_{f^*,v}]$ is connected.
Therefore, $u \in N_{G'}(Y_{f^*,v})$, which implies $u \in V(\partial f^*)$ by Observation~\ref{obs-normalprop}\ref{item:normal-properties-4} and hence $V(\partial f^*) \nsubseteq C$.
To see $V(\partial f^*) \cap C \neq \emptyset$, note that $u' = v$.
Indeed, if $u' \in Y_{f^*,v} \backslash \{v\}$, then Observation~\ref{obs-normalprop}\ref{item:normal-properties-4} implies $u = v$, contradicting with the fact $u \notin C$.
So $u' = v \in V(\partial f^*) \cap C$ and $V(\partial f^*) \cap C \neq \emptyset$.

If $u' \in V(G') \backslash L_t$, the image of $u'$ under $\eta'$ lies in the interior of a face in $F_\eta(\partial o_t) \backslash \{o_t\}$.
We then let $f^*$ be this face.
It follows that $u \in V(\partial f^*)$ as $u$ is neighboring to $u'$ and $u \in L_t$.
So we have $V(\partial f^*) \nsubseteq C$ because $u \notin C$.
To see $V(\partial f^*) \cap C \neq \emptyset$, we use the assumption $C \nsubseteq L_{>t}$, which implies $C \cap L_t^+ \neq \emptyset$.
Since $C$ is connected, there exists a path $\pi$ in $C$ from $u'$ to a vertex $w \in C \cap L_t^+$.
Note that $\pi$ intersects $V(\partial f^*)$, because $u'$ lies in the interior of $f^*$ while $w$ does not.
Thus, $V(\partial f^*) \cap C \neq \emptyset$.
\end{proof}

Let $f^* \in F_\eta(\partial o_t) \backslash \{o_t\}$ be the face in the above observation.
We show that $N_{G'}(C) \subseteq V(\partial f^*)$, which implies \ref{item:rcd-key-2} of Lemma~\ref{lem-key}.
Define $C^* = V(\partial f^*) \cap C$, which is nonempty by the construction of $f^*$.
Note that $C^* \subseteq L_t^+$, because $V(\partial f^*) \subseteq L_t$ and $C \subseteq V(G') \backslash (L_t \backslash L_t^+)$.
We claim that $(f^*,u)$ is normal for all $u \in C^*$.
Consider a vertex $u \in C^* \subseteq L_t^+$.
Then $u \in Y_{f,v}$ for some normal pair $(f,v)$.
If $f \neq f^*$, by Observation~\ref{obs-normalprop}\ref{item:normal-properties-3}, either $Y_{f,v} \cap V(\partial f^*) = \emptyset$ or $V(\partial f^*) \subseteq Y_{f,v}$.
The former is false because $u \in Y_{f,v} \cap V(\partial f^*)$.
The latter is false because $V(\partial f^*) \nsubseteq C$ (note that $Y_{f,v} \subseteq C$ as $u \in C$ and $G[Y_{f,v}]$ is connected).
Thus, $f = f^*$.
Then by Observation~\ref{obs-normalprop}\ref{item:normal-properties-2}, $Y_{f,v} \cap V(\partial f^*) = Y_{f,v} \cap V(\partial f) = \{v\}$, which implies $u = v$.
It follows that $(f^*,u) = (f,v)$ is a normal pair.
Now let $C'$ be the union of $\bigcup_{u \in C^*} Y_{f^*,u}$ and all connected components of $G' - L_t$ that are neighboring to $\bigcup_{u \in C^*} Y_{f^*,u}$ in $G'$.
Clearly, we have $C' \subseteq C$, because $\bigcup_{u \in C^*} Y_{f^*,u} \subseteq C$ and thus any connected component of $G' - L_t$ neighboring to $\bigcup_{u \in C^*} Y_{f^*,u}$ should be contained in $C$.
The following observation implies $N_{G'}(C) \subseteq V(\partial f^*)$.

\begin{observation}
$N_{G'}(C') \subseteq V(\partial f^*)$ and $C' = C$.
\end{observation}
\begin{proof}
Let $v \in N_{G'}(C')$ and $v' \in C'$ be a neighbor of $v$ in $G'$.
Observe that $v \in L_t$.
Indeed, if $v \in V(G') \backslash L_t$, then we have $v' \in C'$ by the construction of $C'$ (as it is neighboring to $v$), contradicting with the fact $v \in N_{G'}(C')$.
As $v' \in C'$, there are two cases: $v' \in \bigcup_{u \in C^*} Y_{f^*,u}$ or $v'$ lies in a connected component of $G' - L_t$ that is neighboring to $\bigcup_{u \in C^*} Y_{f^*,u}$ in $G'$.
In the first case, $v \in N_{G'}(Y_{f^*,u})$ for some $u \in C^*$.
Because $v \in L_t$, we have $v \in V(\partial f^*)$ by Observation~\ref{obs-normalprop}\ref{item:normal-properties-4}.
In the second case, (the image of) the connected component of $G' - L_t$ containing $v$ is contained in a face $f \in F_\eta(\partial o_t) \backslash \{o_t\}$.
Since that connected component is neighboring to $\bigcup_{u \in C^*} Y_{f^*,u}$, there exists $u \in C^*$ such that $Y_{f^*,u} \cap V(\partial f) \neq \emptyset$.
By Observation~\ref{obs-normalprop}\ref{item:normal-properties-3}, this implies either $V(\partial f) \subseteq Y_{f^*,u}$ or $f = f^*$.
The former is not true because $v \in V(\partial f) \backslash C' \subseteq V(\partial f) \backslash Y_{f^*,u}$.
So $f = f^*$ and $v \in V(\partial f^*)$.

To further show $C' = C$, we notice that $V(\partial f^*) \cap C = C^* = V(\partial f^*) \cap C'$, and hence $N_{G'}(C') \cap (V(\partial f^*) \cap C) = \emptyset$.
This implies $N_{G'}(C') \cap C = \emptyset$, as $N_{G'}(C') \subseteq V(\partial f^*)$.
But $C' \subseteq C$ and $C$ is connected.
Therefore, $C' = C$.
\end{proof}

\paragraph{Verifying Condition \ref{item:rcd-key-3} of Lemma~\ref{lem-key}.}
Let $f \in F_\eta(\partial o_t) \backslash \{o_t\}$ be a face.
By the construction of $L_t^+$, we have $v \in L_t^+$ for all vertices $v \in \partial (f) \cap X$ such that $(f,v)$ is normal.
Therefore, for every $v \in V(\partial f) \cap (X \backslash L_t^+)$, either $(f,v)$ is singular or $(f,v)$ is critical.
By Observation~\ref{obs-O1singular}, there are only $O_g(1)$ vertices $v \in V(\partial f)$ such that $(f,v)$ is singular.
So it suffices to bound the number of vertices $v \in V(\partial f) \cap (X \backslash L_t^+)$ such that $(f,v)$ is critical.
Recall that $X$ consists of the vertices on the legal paths $\pi_1,\dots,\pi_c$.
By Observation~\ref{obs-legal}, each legal path contains at most two vertices $v \in V(\partial f)$ such that $(f,v)$ is critical.
So the number of vertices $v \in V(\partial f) \cap (X \backslash L_t^+)$ such that $(f,v)$ is critical is at most $2c$.
This further implies $V(\partial f) \cap (X \backslash L_t^+) = O_{g,c}(1)$.

\paragraph{Verifying Condition \ref{item:rcd-key-4} of Lemma~\ref{lem-key}.}
Let $f \in F_\eta(\partial o_t) \backslash \{o_t\}$ be a face.
We first claim that either $V(\partial f) \subseteq L_t^+$ or $V(\partial f) \cap L_t^+ \subseteq X$.
Indeed, if there exists a normal pair $(f',v')$ with $f' \neq f$ and $v' \in X$ such that $V(\partial f) \subseteq Y_{f',v'}$, then $V(\partial f) \subseteq L_t^+$.
Otherwise, by Observation~\ref{obs-normalprop}\ref{item:normal-properties-3}, $Y_{f',v'} \cap V(\partial f) = \emptyset$ for all normal pair $(f',v')$ with $f' \neq f$ and $v' \in X$.
This implies that any vertex $u \in V(\partial f) \cap L_t^+$ must belong to $Y_{f,v}$ for some normal pair $(f,v)$ with $v \in X$.
But by Observation~\ref{obs-normalprop}\ref{item:normal-properties-2}, we have $Y_{f,v} \cap V(\partial f) = \{v\}$ and thus $u = v \in X$.

If $V(\partial f) \subseteq L_t^+$, then $V(\partial f) \backslash L_t^+ = \emptyset$ and we are done.
If $V(\partial f) \cap L_t^+ \subseteq X$, then $V(\partial f) \backslash L_t^+ = (V(\partial f) \backslash X) \cup (V(\partial f) \cap (X \backslash L_t^+))$.
As shown above, $|V(\partial f) \cap (X \backslash L_t^+)| = O_{g,c}(1)$.
Thus, to show \ref{item:rcd-key-4} of Lemma~\ref{lem-key}, we only need to show $\partial f - X$ has $O_{g,c}(1)$ connected components.
Since $X$ consists of the paths $\pi_1,\dots,\pi_c$ and by Observation~\ref{obs-legal} the intersection of each $\pi_i$ and $V(\partial f)$ has $O_g(1)$ connected components (in $\partial f$), we know that $\partial f[V(\partial f) \cap X]$ has $O_{g,c}(1)$ connected components.
By \ref{item:degree-2} and \ref{item:degree-3} of Lemma~\ref{lem-degree}, the maximum degree of $\partial f$ is $O(g)$ and there are $O(g)$ vertices in $\partial f$ of degree at least 3.
The following fact then implies that $\partial f - (V(\partial f) \cap X)$ has $O_{g,c}(1)$ connected components (setting $G = \partial f$ and $V = V(\partial f) \cap X$).

\begin{fact}\label{fact-components}
Let $G$ be a graph of maximum degree $\alpha$ in which there are $\beta$ vertices of degree at least 3.
Then for every $V \subseteq V(G)$, the graph $G-V$ has at most $O_{\alpha,\beta,\gamma}(1)$ more connected components than $G$, where $\gamma$ is the number of connected components of $G[V]$.
In particular, deleting $\gamma$ vertices from $G$ can increase the number of connected components by at most $O_{\alpha,\beta,\gamma}(1)$.
\end{fact}

\begin{proof}
 First let us make the easy remark that in any graph $G$, removing a set $X \subseteq V(G)$ such that $G[X]$ is connected and $|N(X)| = a$ increases the number of connected components by at most $a$.
 Indeed, consider $C$ the connected component of $G$ containing $X$ and $C_1, \dots C_r$ the connected components of $C - X$.
 Since $C$ is connected, there exists a path between $C_i$ and $C_j$ for every $i \neq j$ in $C$ and this path has to use some vertex of $X$.
 In particular, it means that $X$ is adjacent to every of the $C_i$ and thus $|N(X)| \geq r$.
 Since all the other connected components of $G$ are untouched, this means that removing $X$ indeed increases the number of connected components by at most $a$.

 We now claim that, in a graph of maximum degree $\alpha$, if $S$ is a set of vertices such that $G[S]$ is connected and $S$ contains $k$ vertices of degree at least $3$ in $G$, then the number of vertices of $S$ with a neighbor outside of $S$ is bounded by $(2 \alpha)^{k+1}$.
    
 The proof is by induction on $k$. If $k = 0$, then $G[S]$ is a connected graph where every vertex has degree 1 or 2.
 This means that $G[S]$ is either a cycle or a path.
 In the first case all the vertices have degree 2 in $G$ and in $G[S]$, which means no vertices of $S$ has a neighbor outside of $S$.
 In the second case, only the extremities of the path can have a neigbhor outside.

 Suppose now that the result is true up to $k-1$ and let $x$ be a vertex of $S$ of degree at least 3 in $G$.
 Let $S_1, \dots, S_t$ be the connected components of $G[S]$ after removing $x$.
 Note that because $x$ has degree at most $\alpha$ and $G[S]$ is connected, it meanst that $x$ has a neighbor in each $S_i$ and thus $t \leq \alpha$.
 Moreover, the number of vertices of degree $3$ in each $S_i$ is at most $k-1$, which means by induction that the number of vertices with neighbors outside $S_i$ is at most $(2 \alpha)^{k} $.
 Therefore the number of vertices of $S$ with a neighbor outside of $S$ is bounded by $1 + \alpha \cdot (2 \alpha)^{k} \leq (2 \alpha)^{k+1} $, which ends the proof of our claim.

 Let $V$ be a set of vertices and $C_1,\dots,C_{\gamma}$ be the set of connected components of $G[V]$.
 Because there is at most $\beta$ vertices of degree at least 3 in $G$ and the maximum degree is $\alpha$, the previous claim implies that $N(C_i) \leq (2 \alpha)^{\beta+1} \cdot \alpha $ for every $i \in [\gamma]$.
 This means that removing the sets $C_i$ one after the other, we only increase the number of connected components by $O(\alpha^{O(\beta)})$ in each step.
 Since there are only $\gamma$ steps, we get that the $G-V$ has at most $O( \gamma \cdot \alpha^{O(\beta)})$ more connected components than $G$.
\end{proof}

\subsection{Decomposing almost-embeddable graphs}
In this section, based on Lemma~\ref{lem-key}, we prove a result on almost-embeddable graphs (Corollary~\ref{cor-decomp2} below), which roughly states that one can decompose a connected almost-embeddable graph into $p+1$ parts in which the first $p$ parts satisfy the (robust) ``bounded-treewidth contraction'' property and the last part connects a small set of given vertices.
We begin with apex-free almost-embeddable graphs, and establish the following decomposition lemma.

\begin{lemma} \label{lem-decomp1}
Given a connected graph $G$ with an apex-free $h$-almost-embeddable structure in which the partial embedding is minimal, a number $p$, and a set $\varPhi \subseteq V(G)$ of size $c$, one can compute in polynomial time $p$ disjoint sets $Z_1,\dots,Z_p \subseteq V(G) \backslash N_G[\mathsf{Vort} \cup \varPhi]$ satisfying the following conditions (where $\mathsf{Vort} \subseteq V(G)$ consists of the vortex vertices in $G$).
\begin{enumerate}[label = (\arabic*)]
    \item\label{item:decomp1-1} $\tw(G/(Z_i \backslash Z')) = O_{h,c}(p+|Z'|)$ for all $i \in [p]$ and $Z' \subseteq Z_i$.
    \item\label{item:decomp1-2} $\varPhi$ is contained in one connected component of $G - \bigcup_{i=1}^p Z_i$.
\end{enumerate}
\end{lemma}

Before proving the above lemma, we first observe that it implies a decomposition for (general) almost-embeddable graphs satisfying similar conditions, which is the following corollary.

\begin{corollary} \label{cor-decomp2}
Given a connected graph $G$ with an $h$-almost-embeddable structure in which the partial embedding is minimal, a number $p$, and a set $\varPhi \subseteq V(G)$ of size $c$, one can compute in polynomial time $p$ disjoint sets $Z_1,\dots,Z_p \subseteq V(G) \backslash (N_G[\mathsf{Vort}] \cup A)$ satisfying the following conditions (where $\mathsf{Vort} \subseteq V(G)$ consists of the vortex vertices and $A \subseteq V(G)$ is the apex set of $G$).
\begin{enumerate}[label = (\arabic*)]
    \item\label{item:decomp2-1} $\tw(G/(Z_i \backslash Z')) = O_{h,c}(p+|Z'|)$ for all $i \in [p]$ and $Z' \subseteq Z_i$.
    \item\label{item:decomp2-2} $\varPhi$ is contained in one connected component of $G'$, where $G'$ denotes the graph obtained from $G - \bigcup_{i=1}^p Z_i$ by deleting all edges $(a,v)$ where $a \in A$ and $v \in (\bigcup_{i=1}^p N_G(Z_i)) \backslash A$.
    \item\label{item:decomp2-3} $Z_i \cap N_G[(\varPhi \backslash A)] = \emptyset$ for all $i \in [p]$.
\end{enumerate}
\end{corollary}

\begin{proof}
Let $G$ be the connected graph given in the corollary.
Denote by $\mathsf{Vort} \subseteq V(G)$ the set of vortex vertices in $G$ and $A \subseteq V(G)$ the apex set of $G$.
Suppose $G_0$ is the embeddable part of $G$, and $\eta: G_0 \rightarrow \varSigma$ is the minimal partial embedding into a surface $\varSigma$ of genus $g \leq h$.
Also, let $p$ and $\varPhi \subseteq V(G)$ be as in the corollary.

Let $C_1,\dots,C_r$ be the connected components of $G-A$, which are apex-free $h$-almost-embeddable graphs.
We say an edge $(a,a')$ of $G[A]$ is \emph{redundant} if $a,a' \in N_G(C_i)$ for some $i \in [r]$.
Note that if we remove all redundant edges from $G$, the resulting graph is still connected.
Now consider an index $j \in [r]$.
Since $C_j$ is a subgraph of $G$, as observed in \cite{BandyapadhyayLLSJ22} (last paragraph of the proof of Lemma~4 in the arxiv version), we can obtain a (apex-free) $3h$-almost-embeddable structure of $C_j$ (from the $h$-almost-embeddable structure of $G$) in which
\begin{enumerate}[label = (\roman*)]
    \item the underlying surface is still $\varSigma$;
    \item the embeddable part is $C_j \cap G_0$ plus some isolated vertices;
    \item the partial embedding restricted to $C_j \cap G_0$ is the same as $\eta$;
    \item the vortex vertices are those in $C_j \cap \mathsf{Vort}$.
\end{enumerate}
Conditions (ii) and (iii) guarantee that the partial embedding in the almost-embeddable structure of $C_j$ is minimal by Fact~\ref{fact-minimal}.
For each vertex $a \in N_G(C_j) \subseteq A$, we take a vertex $v \in C_j$ neighboring to $a$ in $G$; for convenience, we call $v$ the \emph{projection image} of $a$ in $C_j$ and $av$ the \emph{projection edge} of $a$ to $C_j$.
Let $\varPi_j \subseteq C_j$ be the set of all projection images in $C_j$ and $\varPhi_j = (\varPhi \cap C_j) \cup \varPi_j$.
Since $N_G(C_j) \subseteq A$, we have $|\varPi_j| \leq |A| \leq h$ and thus $|\varPhi_j| \leq c+h$.
Using Lemma~\ref{lem-decomp1}, we can compute disjoint sets $Z_1^{(j)},\dots,Z_p^{(j)} \subseteq C_j \backslash N_G[\mathsf{Vort} \cup \varPhi_j]$ such that\
\begin{enumerate}[label = (\arabic*)]
 \item $\tw(C_j/(Z_i^{(j)} \backslash Z')) = O_{h,c}(p+|Z'|)$ for all $i \in [p]$ and $Z' \subseteq Z_i$, and
 \item $\varPhi$ is in one connected component of $C_j - \bigcup_{i=1}^p Z_i^{(j)}$.
\end{enumerate}
We then define $Z_i = \bigcup_{j=1}^r Z_i^{(j)}$ for all $i \in [p]$.
It is clear from our construction that $Z_1,\dots,Z_p \subseteq V(G) \backslash (N_G[\mathsf{Vort}] \cup A)$ and these sets are disjoint.
So it suffices to verify that $Z_1,\dots,Z_p$ satisfy conditions \ref{item:decomp2-1}-\ref{item:decomp2-3} in the corollary.

To verify \ref{item:decomp2-1}, consider an index $i \in [p]$ and a subset $Z' \subseteq Z_i$.
We have $\tw(G/(Z_i \backslash Z')) \leq \tw((G-A)/(Z_i \backslash Z')) + |A|$.
As $|A| \leq h$, we only need to bound $\tw((G-A)/(Z_i \backslash Z'))$.
Note that $(G-A)/(Z_i \backslash Z')$ is the disjoint union of the graphs $C_1/(Z_i^{(1)} \backslash Z'),\dots,C_r/(Z_i^{(r)} \backslash Z')$, each of which is of treewidth $O_{h,c}(p+|Z'|)$ by \ref{item:decomp1-1} of Lemma~\ref{lem-decomp1}.
Therefore, $\tw((G-A)/(Z_i \backslash Z')) = O_{h,c}(p+|Z'|)$.

To verify \ref{item:decomp2-2}, let $G'$ be the graph obtained from $G - \bigcup_{i=1}^p Z_i$ by deleting all edges $(a,v)$ where $a \in A$ and $v \in N_G(Z_i) \backslash A$.
For each $j \in [r]$, let $C_j'$ be the connected component of $C_j - \bigcup_{i=1}^p Z_i^{(j)}$ containing the vertices in $\varPhi_j$.
Note that if a vertex $a \in A$ is neighboring to $C_j$, then it is neighboring to $C_j'$ as $\varPhi_j$ contains the projection image of $a$ in $C_j$.
Also, the projection edge of $a$ to $C_j$ preserves in $G'$, because $Z_i^{(j)} \cap N_{C_j}[\varPhi_j] = \emptyset$ by Lemma~\ref{lem-decomp1} and thus the projection image of $a$ in $C_j$ is not contained in $N_G(Z_i) \backslash A$.
Therefore, the graph $G'[A \cup (\bigcup_{j=1}^r C_j')]/(\bigcup_{j=1}^r C_j')$ is isomorphic to $G/(\bigcup_{j=1}^r C_j)$.
Since $G$ is connected, $G/(\bigcup_{j=1}^r C_j)$ is also connected by Fact~\ref{fact-quotient} and thus $G'[A \cup (\bigcup_{j=1}^r C_j')]/(\bigcup_{j=1}^r C_j')$ is connected.
Using Fact~\ref{fact-quotient} again, we can deduce that $G'[A \cup (\bigcup_{j=1}^r C_j')]$ is connected, which implies that $\varPhi$ is contained in one connected component of $G'$.

To verify \ref{item:decomp2-3} is straightforward.
Consider a vertex $v \in \varPhi \backslash A$, which belongs to $C_j$ for some $j \in [r]$.
Then $v \in \varPhi^{(j)}$.
By Lemma~\ref{lem-decomp1}, for all $i \in [p]$, $Z_i^{(j)} \cap N_{C_j}[v] = \emptyset$ and thus $Z_i^{(j)} \cap N_G[v] = \emptyset$.
For $j' \in [r]$ other than $j$, it is clear that $Z_i^{(j')} \cap N_G[v] = \emptyset$.
Therefore, $Z_i \cap N_G[v] = \emptyset$.
\end{proof}

The rest of this section is dedicated to proving Lemma~\ref{lem-decomp1}.
Let $G$ be the graph as in the lemma.
Suppose $G_0$ is the embeddable part of $G$ and $\eta\colon G_0 \rightarrow \varSigma$ is the minimal partial embedding where $\varSigma$ is a surface of genus $g \leq h$.
We view $(G_0,\eta)$ as a $(\varSigma,x_0)$-embedded graph by (arbitrarily) picking a reference point $x_0 \in \varSigma$ not in the image of $\eta$.
Denote by $\mathsf{Vort} \subseteq V(G)$ the set of vortex vertices of $G$.
Also, let $p$ and $\varPhi$ be as in the lemma.
For each face $f \in F_\eta(G_0)$, we denote by $\tilde{\partial} f$ the union of $\partial f$ and the vortices of $G$ contained in $f$.
Clearly, if $f$ is not a vortex face, then $\tilde{\partial} f = \partial f$.

\begin{observation} \label{obs-vortexfaceconn}
$\tilde{\partial} f$ is connected for all $f \in F_\eta(G_0)$.
\end{observation}

\begin{proof}
Assume $\tilde{\partial} f$ is not connected.
We partition $V(\tilde{\partial} f)$ into two subsets $V_1$ and $V_2$ such that there exists no edge in $\tilde{\partial} f$ between $V_1$ and $V_2$.
Next, we classify all vertices of $G$ as \emph{type-1} vertices and \emph{type-2} vertices as follows.
A vertex in $V(\tilde{\partial} f)$ is of type-1 (resp., type-2) if it is contained in $V_1$ (resp., $V_2$).
Consider a vertex $v \in V(G) \backslash V(\tilde{\partial} f)$.
If $v$ is a vertex of $G_0$, then its image under $\eta$ is contained in some $f' \in F_\eta(\partial f) \backslash \{f\}$.
Since $\eta$ is a minimal embedding, $V(\partial f')$ is contained in one connected component $C$ of $\partial f$, and $C$ in turn contained in either $V_1$ or $V_2$.
We classify $v$ as type-1 (resp., type-2) if $C$ is contained in $V_1$ (resp., $V_2$).
If $v$ is a vortex vertex, then the vortex $v$ belongs to is contained in a face of $(G,\eta)$ other than $f$, which is in turn contained in some $f' \in F_\eta(\partial f) \backslash \{f\}$.
Similarly, we classify $v$ as type-1 (resp., type-2) if the connected component of $\partial f$ containing $V(\partial f')$ is contained in $V_1$ (resp., $V_2$).
One can easily check that there is no edge in $G$ between a type-1 vertex and a type-2 vertex.
Indeed, the edges of $\tilde{\partial} f$ do not connect type-1 vertices with type-2 vertices, as there exists no edge in $\tilde{\partial} f$ between $V_1$ and $V_2$.
The other edges of $G$ are either in $E(G_0) \backslash E(\partial f)$ or in vortices not contained in $f$.
An edge $e \in E(G_0) \backslash E(\partial f)$ has its image in some face $f' \in F_\eta(\partial f) \backslash \{f\}$.
If the connected component of $\partial f$ containing $V(\partial f')$ is contained in $V_1$ (resp., $V_2$), then the two endpoints $e$ are both type-1 (resp., type-2) vertices.
If $e$ is an edge in a vortex not contained in $f$, then that vortex is contained in a face of $(G,\eta)$ other than $f$, which is in turn contained in some $f' \in F_\eta(\partial f) \backslash \{f\}$.
Again, if the connected component of $\partial f$ containing $V(\partial f')$ is contained in $V_1$ (resp., $V_2$), then the two endpoints $e$ are both type-1 (resp., type-2) vertices.
Therefore, no edge in $G$ is between a type-1 vertex and a type-2 vertex.
Note that there exist at least one type-1 vertex and one type-2 vertex, because $V_1 \neq \emptyset$ and $V_2 \neq \emptyset$.
It follows that $G$ is not connected, contradicting with our assumption.
As such, $\tilde{\partial} f$ must be connected.
\end{proof}

In order to construct the sets $Z_1,\dots,Z_p$, consider the radial layering $L_1,\dots,L_m$ of $(G_0,\eta)$.
We classify these layers as bad and good layers as follows.
The first layer $L_1$ is always bad.
For $i \geq 2$, the layer $L_i$ is bad if $N_G[\varPhi] \cap L_i \neq \emptyset$ or $N_G[\mathsf{Vort}] \cap L_i \neq \emptyset$, and is good otherwise.
By Fact~\ref{fact-diff1}, if $N_G[\varPhi] \cap L_i \neq \emptyset$, then $L_{i-1} \cup L_i \cup L_{i+1}$ contains at least one vertex in $\varPhi$.
Therefore, there are at most $3|\varPhi| = 3c$ indices $i \in [m]$ such that $N_G[\varPhi] \cap L_i \neq \emptyset$.
Again by Fact~\ref{fact-diff1}, each vortex face $f \in F_\eta(G_0)$ is incident to at most two consecutive radial layers $L_i,L_{i+1}$, and hence the vertices in the vortices contained in $f$ are neighboring to at most four radial layers $L_{i-1},L_i,L_{i+1},L_{i+2}$.
Since the number of vortex faces is at most $h$, there are at most $4h$ indices $i \in [m]$ such that $N_G[\mathsf{Vort}] \cap L_i \neq \emptyset$.
It follows that the number of bad layers is at most $3c+4h+1$.

Before defining $Z_1,\dots,Z_p$, we first figure out which vertices of $V(G)$ are not included in any $Z_i$ and guarantee that these vertices connect $\varPhi$.
To this end, we shall iteratively define a sequence of sets $X_m,\dots,X_1$ of vertices in $G$, where $X_i \subseteq L_i$.
Suppose $X_{t+1},\dots,X_m$ have already been constructed, and we now construct $X_t$.
If $L_t$ is a bad layer, we simply set $X_t = L_t$.
Otherwise, $L_t$ is a good layer and we construct $X_t$ as follows.
Let $\mathcal{C}$ denote the set of connected components of $G[\mathsf{Vort} \cup (\bigcup_{i=t+1}^m X_i)]$.
For $C \in \mathcal{C}$ with $C \cap \varPhi \neq \emptyset$ and $N_G(C) \cap L_t \neq \emptyset$, we pick a vertex $v_C \in N_G(C) \cap L_t$.
Set $\varPhi_t = \{v_C \mid C \in \mathcal{C}\}$.
Observe that $|\varPhi_t| \leq |\varPhi| = c$.
Indeed, since $L_t$ is good, $L_t \cap \mathsf{Vort} = \emptyset$ and thus $L_t \cap (\mathsf{Vort} \cup (\bigcup_{j=i+1}^m X_j)) = \emptyset$.
Thus, the number of components $C \in \mathcal{C}$ with $C \cap \varPhi \neq \emptyset$ is at most $|\varPhi \backslash L_t|$, which implies $|\varPhi_t| \leq |\varPhi \backslash L_t| \leq |\varPhi| = c$.
As $L_t$ is a good layer, all faces of $(G_0,\eta)$ incident to $L_t$ do not contain vortices.
Since $\eta$ is a minimal embedding, by Observation~\ref{obs-vortexfaceconn}, this further implies that every face of $(G_0,\eta)$ incident to $L_t$ has a connected boundary.
Thus, we can apply Lemma~\ref{lem-key} on $G_0$ with $\varPhi = \varPhi_t$ to obtain a set $X \subseteq L_t$ containing $\varPhi_t$ and another set $L_t^+ \subseteq L_t$ satisfying the properties in the lemma.
We then set $X_t = X$.

By our construction, it is clear that $\varPhi \subseteq \mathsf{Vort} \cup (\bigcup_{i=1}^m X_i)$.
We then prove that all vertices in $\varPhi$ are contained in the same connected component of $G[\mathsf{Vort} \cup (\bigcup_{i=1}^m X_i)]$.
This is achieved by two steps.
In the first step, we show that every connected component of $G[\mathsf{Vort} \cup (\bigcup_{i=1}^m X_i)]$ which contains at least one vertex in $\varPhi$ must intersect $L_1$.
Note that $L_1$ is a bad layer and hence $L_1 = X_1 \subseteq \mathsf{Vort} \cup (\bigcup_{i=1}^m X_i)$.
In the second step, we then show that all vertices in $L_1$ lie in the same connected component of $G[\mathsf{Vort} \cup (\bigcup_{i=1}^m X_i)]$.

To do the first step, consider a connected component $C$ of $G[\mathsf{Vort} \cup (\bigcup_{i=1}^m X_i)]$ which contains at least one vertex in $\varPhi$.
We want to show $L_1 \cap C \neq \emptyset$.
Let $t$ be the smallest index such that $L_t \cap C \neq \emptyset$.
Assume $t>1$ for a contradiction.

\begin{observation}
$N_G(C) \cap L_{t-1} \neq \emptyset$.
\end{observation}
\begin{proof}
If $L_t$ is a good layer, then $\mathsf{Vort} \cap L_t = \emptyset$.
Therefore, we have $L_t \cap C = X_t \cap C$, which implies $X_t \cap C \neq \emptyset$.
So $C$ contains at least one connected component of $G[X_t] = G_0[X_t]$.
By \ref{item:rcd-key-1} of Lemma~\ref{lem-key}, every connected component of $G_0[X_t]$ is neighboring to $L_{t-1}$, which implies $N_G(C) \cap L_{t-1} \neq \emptyset$.
If $L_t$ is a bad layer, then $X_t = L_t$.
Pick an arbitrary vertex $u \in L_t \cap C$.
As $t>1$, there must exist a face $f \in F_\eta(G_0)$ incident to $u$ and $L_{t-1}$.
Define $S = V(\tilde{\partial} f) \cap C$.
Note that $S \neq \emptyset$ as $u \in S$, and $V(\tilde{\partial} f) \backslash S \neq \emptyset$ as $L_{t-1} \cap V(\tilde{\partial} f) \neq \emptyset$ and $L_{t-1} \cap S \subseteq L_{t-1} \cap C = \emptyset$.
By Observation~\ref{obs-vortexfaceconn}, $\tilde{\partial} f$ is connected.
So there exists $u' \in V(\tilde{\partial} f) \backslash S$ that is neighboring to $S$ in $\tilde{\partial} f$ (and thus in $G$).
It follows that $u' \in N_G(C)$.
But $C$ is a connected component of $G[\mathsf{Vort} \cup (\bigcup_{i=1}^m X_i)]$, which implies $u' \notin \mathsf{Vort} \cup (\bigcup_{i=1}^m X_i)$ and in particular $u' \notin \mathsf{Vort} \cup X_t = \mathsf{Vort} \cup L_t$.
Since $f$ is incident to $u$ and $L_{t-1}$, we have $V(\partial f) \subseteq L_{t-1} \cup L_t$ by Fact~\ref{fact-diff1}, and hence $V(\tilde{\partial} f) \subseteq \mathsf{Vort} \cup L_t \cup L_{t-1}$.
This implies that $u' \in L_{t-1}$ and $N_G(C) \cap L_{t-1} \neq \emptyset$.
\end{proof}

Based on the observation $N_G(C) \cap L_{t-1} \neq \emptyset$, we deduce a contradiction as follows.
If $L_{t-1}$ is a bad layer, then $X_{t-1} = L_{t-1}$ and thus $N_G(C) \cap X_{t-1} \neq \emptyset$, which contradicts with the fact that $C$ is a connected component of $G[\mathsf{Vort} \cup (\bigcup_{i=1}^m X_i)]$.
If $L_{t-1}$ is a good layer, recall the set $\varPhi_{t-1}$ we define when constructing $X_{t-1}$.
As $C$ is disjoint from $L_1,\dots,L_{t-1}$, $C$ is in fact a connected component of $G[\mathsf{Vort} \cup (\bigcup_{i=t}^m X_i)]$.
Also, we have $C \cap \varPhi \neq \emptyset$ by assumption and $N_G(C) \cap L_{t-1} \neq \emptyset$ as shown above.
Therefore, $\varPhi_{t-1}$ includes a vertex $v_C \in N_G(C) \cap L_{t-1}$.
By construction, we have $\varPhi_{t-1} \subseteq X_{t-1}$ and thus $v_C \in X_{t-1}$.
This implies $N_G(C) \cap X_{t-1} \neq \emptyset$, which contradicts with the fact that $C$ is a connected component of $G[\mathsf{Vort} \cup (\bigcup_{i=1}^m X_i)]$.
In both cases, we have contradictions, so the assumption $t>1$ is wrong.
It follows that every connected component of $G[\mathsf{Vort} \cup (\bigcup_{i=1}^m X_i)]$ containing at least one vertex in $\varPhi$ intersects $L_1$.

We now do the second step, where we want to show that all vertices in $L_1$ lie in the same connected component of $G[\mathsf{Vort} \cup (\bigcup_{i=1}^m X_i)]$.
We have $L_1 = V(\partial o)$, where $o$ is the outer face of $G_0$.
It follows that $V(\tilde{\partial} o) \subseteq \mathsf{Vort} \cup L_1 = \mathsf{Vort} \cup X_1$.
Thus, $\tilde{\partial} o$ is a subgraph of $G[\mathsf{Vort} \cup (\bigcup_{i=1}^m X_i)]$, which is connected by Observation~\ref{obs-vortexfaceconn}.
Combining this with the first step, we can conclude that all vertices in $\varPhi$ are contained in the same connected component of $G[\mathsf{Vort} \cup (\bigcup_{i=1}^m X_i)]$.

Next, we are ready to construct the sets $Z_1,\dots,Z_p$.
Set $\Delta = p+3c+4h+1$.
We say a number $q \in [\Delta]$ is \emph{bad} if $q$ is congruent to the index of a bad layer modulo $\Delta$, i.e., $q \equiv i \pmod{\Delta}$ for some $i \in [m]$ such that $L_i$ is a bad layer, and is \emph{good} otherwise.
As argued before, there are at most $3c+4h+1$ bad layers, and therefore there are at most $3c+4h+1$ bad numbers in $[\Delta]$.
So we can always find $p$ good numbers $q_1,\dots,q_p \in [\Delta]$.
We then construct the sets $Z_1,\dots,Z_p$ by simply defining $Z_i$ as the union of $L_q \backslash (X_q \cup L_q^+)$ for all indices $q \in [m]$ congruent to $q_i$ modulo $\Delta$, i.e.,
\begin{equation}
 \label{eq:zi}
 Z_i \coloneqq \bigcup_{j=0}^{\lfloor(m-q_i)/\Delta\rfloor} (L_{j\Delta+q_i} \backslash (X_{j\Delta+q_i} \cup L_{j\Delta+q_i}^+)).
\end{equation}
Since the layers $L_1,\dots,L_m$ are disjoint, the sets $Z_1,\dots,Z_p$ are also disjoint.
Also, $Z_1,\dots,Z_p \subseteq V(G) \backslash N_G[\mathsf{Vort}]$, because $Z_1,\dots,Z_p$ only contain vertices in good layers.
Furthermore, $\varPhi$ is contained in one connected component of $G - \bigcup_{i=1}^p Z_i$, because $\mathsf{Vort} \cup (\bigcup_{i=1}^m X_i) \subseteq V(G) \backslash \bigcup_{i=1}^p Z_i$ and all vertices in $\varPhi$ are contained in the same connected component of $G[\mathsf{Vort} \cup (\bigcup_{i=1}^m X_i)]$.

Finally, we show that $\tw(G/(Z_i \backslash Z')) = O_{h,c}(p+|Z'|)$ for all $i \in [p]$ and $Z' \subseteq Z_i$, which is the most complicated part in this proof.
Without loss of generality, it suffices to show $\tw(G/(Z_1 \backslash Z')) = O_{h,c}(p+|Z'|)$ for all $Z' \subseteq Z_1$.
Let $G_1,\dots,G_r$ be the vortices of $G$ attached to disjoint facial disks $D_1,\dots,D_r$ in $(G_0,\eta)$ with witness pairs $(\tau_1,\mathcal{P}_1),\dots,(\tau_r,\mathcal{P}_r)$, where $r \leq h$.
We call the faces of $(G_0,\eta)$ containing $D_1,\dots,D_r$ \emph{vortex faces}.
We first add some ``virtual'' edges to $G_0$ as follows.
Consider an index $i \in [r]$.
Suppose $\tau_i = (v_{i,1},\dots,v_{i,\ell_i})$.
By definition, $v_{i,1},\dots,v_{i,\ell_i}$ are the vertices of $(G_0,\eta)$ that lie on the boundary of $D_i$, sorted in clockwise or counterclockwise order.
For convenience, we write $v_{i,0} = v_{i,\ell_i}$.
We then add the edges $(v_{i,j-1},v_{i,j})$ for all $j \in [\ell_i]$ to $G_0$, and call them \emph{virtual} edges.
Furthermore, we draw these virtual edges along the boundary of the disk $D_i$ (this is possible because $v_1,\dots,v_{\ell_i}$ are sorted along the boundary of $D_i$).
The images of these virtual edges then enclose the disk $D_i$.
We do this for all indices $i \in [r]$.
Let $G_0'$ denote the resulting graph after adding the virtual edges.
Since $D_1,\dots,D_r$ are disjoint facial disks in $(G_0,\eta)$, the images of the virtual edges do not cross each other or cross the original edges in $(G_0,\eta)$.
Therefore, the drawing of the virtual edges extends $\eta$ to an embedding of $G_0'$ to $\varSigma$; for simplicity, we still use the notation $\eta$ to denote this embedding.
By construction, it is clear that $D_1,\dots,D_r$ are faces of $(G_0',\eta)$.
Note that every face of $(G_0,\eta)$ that does not contain any vortices is preserved in $(G_0',\eta)$, while every vortex face is subdivided into multiple faces in $(G_0',\eta)$ by the virtual edges.
The next observation states that the part of $(G_0',\eta)$ inside each vortex face has a constant vertex-face diameter.

\begin{observation}\label{obs-invortex}
Let $f \in F_\eta(G_0)$ be a vortex face and $F \subseteq F_\eta(G_0')$ consist of the faces of $(G_0',\eta)$ contained in $f$.
Then for any two vertices $v,v' \in V(\partial f)$, there exists a VFA path in $(G_0',\eta)$ from $v$ to $v'$ of length $O(h)$ that does not visit any face in $F_\eta(G_0') \backslash F$.
\end{observation}

\begin{proof}
Let $G^*$ be the VFI graph of $(G_0',\eta)$ restricted to $V(\partial f) \cup F$.
One can easily verify that $G^*$ is connected, because the union of the faces in $F$ is a connected region in $\varSigma$ (which is $f$).
We show that $\mathsf{diam}(G^*) = O(h)$.
Every face in $F$ is incident to some virtual edge added in $f$, and thus shares a common boundary vertex with some disk $D_i \in F$.
It follows that in $G^*$ every face in $F$ is within distance $2$ from some $D_i \in F$.
Therefore, $G^*$ can be covered by $O(h)$ subgraphs of diameter at most $4$, which implies $\mathsf{diam}(G^*) = O(h)$.
Now consider a shortest path $\pi$ from $v$ to $v'$ in $G^*$, which is a VFA path in $(G_0',\eta)$ that does not visit any face in $F_\eta(G_0') \backslash F$.
As $\mathsf{diam}(G^*) = O(h)$, the length of $\pi$ is $O(h)$.
\end{proof}

The reason for why we construct $G_0'$ is to relate the treewidth of $G/(Z_1 \backslash Z')$ with that of $G_0'/(Z_1 \backslash Z')$.
In this way, we can reduce the task to bounding $\tw(G_0'/(Z_1 \backslash Z'))$ without considering the vortices/apices of $G$.
Similar tricks are also used in previous work \cite{BandyapadhyayLLSJ22,DemaineFHT05}.

\begin{observation}\label{obs-GtoG0'}
$\tw(G/(Z_1 \backslash Z')) = O_h(\tw(G_0'/(Z_1 \backslash Z')))$.
\end{observation}

\begin{proof}
We first show that $G[Z_1 \backslash Z'] = G_0[Z_1 \backslash Z'] = G_0'[Z_1 \backslash Z']$.
As $G$ is apex-free, we have $G = G_0 \cup (\bigcup_{i=1}^h G_i)$.
The vertices in the intersection of $G_0$ and the vortices $G_1,\dots,G_h$ are those on the boundaries of $D_1,\dots,D_r$.
Hence, these vertices all lie in the bad layers of $G_0$, while $Z_1$ is the union of several good layers.
This implies $Z_1 \subseteq V(G_0) \backslash (\bigcup_{i=1}^h V(G_i))$ and thus $G[Z_1 \backslash Z'] = G_0[Z_1 \backslash Z']$.
To see $G_0[Z_1 \backslash Z'] = G_0'[Z_1 \backslash Z']$, recall that $G_0'$ is obtained from $G_0$ by adding some virtual edges.
The virtual edges are all on the boundaries of the vortex faces $D_1',\dots,D_h'$, and thus are disjoint from $Z_1$.
So we have $G_0[Z_1 \backslash Z'] = G_0'[Z_1 \backslash Z']$.
Based on this fact, we see that $G_0/(Z_1 \backslash Z')$ is a spanning subgraph of $G_0'/(Z_1 \backslash Z')$, and $G_0/(Z_1 \backslash Z')$ is a subgraph of $G/(Z_1 \backslash Z')$.
Therefore, the vertex set of $G_0'/(Z_1 \backslash Z')$ is a subset of the vertex set of $G/(Z_1 \backslash Z')$.

Now let $(T,\beta)$ be a tree decomposition of $G_0'/(Z_1 \backslash Z')$ of width $w = \tw(G_0'/(Z_1 \backslash Z'))$.
We are going to modify $(T,\beta)$ to a tree decomposition of $G/(Z_1 \backslash Z')$ of width $O(hw + h)$, which proves the claim.
To this end, we apply the same argument as in \cite[Lemma 5.8]{DemaineFHT05}.
For convenience, we do not distinguish the vertices in $V(G) \backslash (Z_1 \backslash Z')$ with their images in $G/(Z_1 \backslash Z')$.
For each vertex $v$ of $G_0'/(Z_1 \backslash Z')$, we use $T(v) \subseteq T$ to denote the set of nodes whose bags contain $v$, which is connected in $T$ by the definition of a tree decomposition.
Consider a vortex $G_i$ with the witness pair $(\tau_i,\mathcal{P}_i)$.
Suppose $\tau_i = (v_{i,1},\dots,v_{i,\ell_i})$.
Then $\mathcal{P}_i$ is a path decomposition of $G_i$ with path $P = (u_{i,1},\dots,u_{i,\ell_i})$ such that $v_{i,j} \in \beta(u_{i,j})$ for all $j \in [\ell_i]$.
We then add the bag $\beta(u_{i,j})$ of $\mathcal{P}_i$ to the bags of all nodes in $T(v_{i,j})$.
We do this for all vortices $G_1,\dots,G_r$.
After that, we add the apex set $A$ to the bags of all nodes in $T^*$.
It is easy to verify that after the modification $(T,\beta)$ is a tree decomposition of $G/(Z_1 \backslash Z')$.
Indeed, the bags of $(T,\beta)$ cover every edge of $G/(Z_1 \backslash Z')$: the edges in $G_0/(Z_1 \backslash Z')$ are covered by the original bags of $(T,\beta)$, the edges in each vortex $G_i$ are covered by the bags of the path decomposition $\mathcal{P}_i$ (which are added to the bags of the corresponding nodes in $T$), and the edges adjacent to the apex set $A$ are also covered because we add $A$ to all bags.
Furthermore, for any vertex $v$ of $G/(Z_1 \backslash Z')$, the nodes whose bags containing $v$ are connected in $T$; this follows from the fact that $(v_{i,1},\dots,v_{i,\ell_i})$ forms a path in $G_0'/(Z_1 \backslash Z')$ (which consists of virtual edges) and thus $\bigcup_{j=j^-}^{j^+} T(v_{i,j})$ is connected in $T$ for any $j^-,j^+ \in [\ell_i]$.
Finally, we observe that the width of $(T,\beta)$ is $O(hw+h)$.
Consider a node $t \in T$.
Originally, the size of the bag $\beta(t)$ is at most $w+1$.
If a vortex $G_i$ contains $c_i$ vertices in (the original) $\beta(t)$, then we added $c_i$ bags of the path decomposition $\mathcal{P}_i$ to $\beta(t)$.
Since the vortices are disjoint, each vertex in $\beta(t)$ can be contained in at most one vortex, which implies $\sum_{i=1}^h c_i \leq w+1$.
Therefore, we added at most $w+1$ bags of the path decompositions $\mathcal{P}_1,\dots,\mathcal{P}_r$ to $\beta(t)$, each of which has size $O(h)$.
So after this step, the size of $\beta(t)$ is $O(hw)$.
Then after we added the apex set $A$ to $\beta(t)$, the size $\beta(t)$ is $O(hw+h)$ because $|A| \leq h$.
\end{proof}

Now it remains to bound $\tw(G_0'/(Z_1 \backslash Z'))$.
Recall $Z_1 = \bigcup_{j=0}^{\lfloor(m-q)/\Delta\rfloor} (L_{j\Delta+q} \backslash (X_{j\Delta+q} \cup L_{j\Delta+q}^+))$ for some good number $q \in [\Delta]$ (see Equation \eqref{eq:zi}).
For notational simplicity, we define $L_i = X_i = L_i^+ = \emptyset$ for all integers $i < 0$ and write $i_j = (j-1) \Delta + q$ for $j \in \mathbb{N}$.
Then we have $Z_1 = \bigcup_{j=1}^{m'} (L_{i_j} \backslash (X_{i_j} \cup L_{i_j}^+))$ for $m' = \lfloor (m-q)/\Delta \rfloor + 1$.
For any $j \in [m']$, $L_{i_j}$ is a good layer.
Hence $X_{i_j}$ and $L_{i_j}^+$ satisfy the conditions in Lemma~\ref{lem-key}.
Let $o_{i_j}$ be the outer face of $(G_0[L_{\geq i_j}],\eta)$, which is also the outer face of $(G_0'[L_{\geq i_j}],\eta)$.
Note that $(G_0'[L_{\geq i_j}],\eta)$ is an outer-preserving extension of $(G_0[L_{i_j}],\eta)$.
So by \ref{item:rcd-key-2} of Lemma~\ref{lem-key}, for every connected component $C$ of $G_0'[L_{i_j}^+ \cup L_{> i_j}]$, $N_{G_0'}(C) \subseteq V(\partial f)$ for some face $f \in F_\eta(\partial o_{i_j}) \backslash \{o_{i_j}\}$.

For each $j \in [m']$ and each face $f \in F_\eta(\partial o_{i_j}) \backslash \{o_{i_j}\}$, we define a set $\kappa(f) \subseteq V(\partial f) \backslash L_{i_j}^+$ as follows.
Denote by $\mathcal{C}_f$ the set of connected components of $\partial f - (L_{i_j}^+ \cup X_{i_j} \cup Z')$.
Note that $V(\partial f) \backslash (L_{i_j}^+ \cup X_{i_j} \cup Z') \subseteq Z_1 \backslash Z'$.
Therefore, every $C \in \mathcal{C}_f$ is contracted into one vertex in $G_0'/(Z_1 \backslash Z')$.
In each $C \in \mathcal{C}_f$, we keep a representative vertex $\xi_C \in C$.
We then define $\kappa(f) = \{\xi_C: C \in \mathcal{C}_f\} \cup ((V(\partial f) \backslash L_{i_j}^+) \cap (X_{i_j} \cup Z'))$.
The following observation bounds the size of $\kappa(f)$.

\begin{observation} \label{obs-kappa}
$|\kappa(f)| = O_{g,c}(|Z'|)$.
\end{observation}
\begin{proof}
We first notice that $|\mathcal{C}_f| = O_{g,c}(|Z'|)$.
By \ref{item:rcd-key-4} of Lemma~\ref{lem-key}, the number of connected components of $\partial f - L_{i_j}^+$ is $O_{g,c}(1)$.
Furthermore, by \ref{item:rcd-key-3} of Lemma~\ref{lem-key}, $|V(\partial f) \cap (X_{i_j} \backslash L_{i_j}^+)| = O_{g,c}(1)$, which implies $|V(\partial f) \cap ((X_{i_j} \cup Z') \backslash L_{i_j}^+)| = O_{g,c}(|Z'|)$.
Thus, $\partial f - (L_{i_j}^+ \cup X_{i_j} \cup Z')$ is a graph obtained from $\partial f - L_{i_j}^+$ by deleting $O_{g,c}(|Z'|)$ vertices.
According to Lemma~\ref{lem-degree}, the maximum degree of $\partial f$ (and thus $\partial f - L_{i_j}^+$) is $O(g)$ and there are at most $O(g)$ vertices in $\partial f$ (and thus $\partial f - L_{i_j}^+$) of degree at least 3.
By Fact~\ref{fact-components}, we further deduce that the number of connected components of $\partial f - (L_{i_j}^+ \cup X_{i_j} \cup Z')$ is $O_{g,c}(|Z'|)$, i.e., $|\mathcal{C}_f| = O_{g,c}(|Z'|)$.
To further bound $|\kappa(f)|$, we observe that $|(V(\partial f) \backslash L_{i_j}^+) \cap (X_{i_j} \cup Z')| = O_{g,c}(|Z'|)$, because $|V(\partial f) \cap (X_{i_j} \backslash L_{i_j}^+)| = O_{g,c}(1)$ by \ref{item:rcd-key-3} of Lemma~\ref{lem-key}.
Based on this fact and the bound $|\mathcal{C}_f| = O_{g,c}(|Z'|)$, we have $|\kappa(f)| = O_{g,c}(|Z'|)$.
\end{proof}

In order to bound $\tw(G_0'/(Z_1 \backslash Z'))$, we shall construct a tree decomposition $(T,\beta)$ for $G_0'/(Z_1 \backslash Z')$ in which each torso is of treewidth $O_{g,c}(p+|Z'|)$.
If such a tree decomposition exists, by Lemma~\ref{lem-torsotw}, we have $\tw(G_0'/(Z_1 \backslash Z')) = O_{g,c}(p+|Z'|)$.
The construction of $(T,\beta)$ is the following.
The depth of the tree $T$ is $m'$.
For all $j \in \{0\} \cup [m']$, the nodes in the $j$-th level of $T$ are one-to-one correspondence to the connected components of $G_0'[L_{i_j}^+ \cup L_{> i_j}]$.
Note that $G_0'[L_{i_0}^+ \cup L_{> i_0}] = G_0'$ is connected, so there is exactly one node in the $0$-th level of $T$, which is the root of $T$.
The parents of the nodes are defined as follows.
Consider a node $t \in V(T)$ in the $j$-th level for $j \geq 1$, and let $C_t$ be the connected component of $G_0'[L_{i_j}^+ \cup L_{> i_j}]$ corresponding to $t$.
Since $G_0'[L_{i_j}^+ \cup L_{> i_j}]$ is a subgraph of $G_0'[L_{i_{j-1}}^+ \cup L_{> i_{j-1}}]$, $C_t$ is contained in a unique connected component of $G_0'[L_{i_{j-1}}^+ \cup L_{> i_{j-1}}]$, which corresponds to a node $t'$ in the $(j-1)$-th level.
We then let $t'$ be the parent of $t$.
For each node $t \in V(T)$, we associate a set $V_t$ to $t$ defined as $V_t \coloneqq C_t \backslash (L_{i_{j+1}}^+ \cup L_{> i_{j+1}})$, where $j \in \{0\} \cup [m']$ is the number such that $t$ is in the $j$-th level of $T$ and $C_t$ is the connected component of $G_0'[L_{i_j}^+ \cup L_{> i_j}]$ corresponding to $t$.
One can easily verify that $\{V_t \mid t \in V(T)\}$ is a partition of $V(G_0')$.
In addition, we associate another set $U_t$ to each node $t \in V(T)$ defined as follows.
If $j = 0$ (i.e., $t$ is the roof of $T$), then set $U_t \coloneqq \emptyset$.
Otherwise, let $t'$ be the parent of $t$.
As $(G_0'[L_{\geq i_j}],\eta)$ is an outer-preserving extension of $(\partial o_{i_j},\eta)$, by \ref{item:rcd-key-2} of Lemma~\ref{lem-key}, there exists a face $f \in F_\eta(\partial o_{i_j}) \backslash \{o_{i_j}\}$ with $N_{G_0'}(C_t) \subseteq V(\partial f) \backslash L_{i_j}^+$.
We then let $U_t \coloneqq V_{t'} \cap \kappa(f)$.
Note that $|U_t| = O_{g,c}(|Z'|)$, as $|\kappa(f)| = O_{g,c}(|Z'|)$ by Observation~\ref{obs-kappa}.
Finally, we define $\beta(t) = \pi(V_t \cup U_t)$ as the bag of each node $t \in V(T)$, where $\pi: V(G_0') \rightarrow V(G_0'/(Z_1 \backslash Z'))$ is the quotient map for the contraction of $G_0'$ to $G_0'/(Z_1 \backslash Z')$.

\begin{observation} \label{obs-treedecomp}
$(T,\beta)$ is a tree decomposition of $G_0'/(Z_1 \backslash Z')$. Furthermore, for all $t \in V(T)$, we have $\sigma(t) = \pi(U_t)$ and thus $|\sigma(t)| = O_{g,c}(|Z'|)$.
\end{observation}
\begin{proof}
As $\{V_t \mid t \in V(T)\}$ is a partition of $V(G_0')$, every vertex $v \in G_0'/(Z_1 \backslash Z')$ is contained in $\beta(t)$ for some $t \in V(T)$ with $V_t \cap \pi^{-1}(\{v\}) \neq \emptyset$.
If $\pi^{-1}(\{v\})$ is a single vertex $u \in V(G_0')$, then there exists a unique node $t \in V(T)$ such that $u \in V_t$.
By our construction, $v$ is only contained in $\beta(t)$ and possibly the bags of the children of $t$.
Thus, the nodes in $T$ whose bags contain $v$ are connected.
If $\pi^{-1}(\{v\})$ is a connected component of $G_0'[Z_1 \backslash Z']$, then it is a connected component of $G_0'[L_{i_j} \backslash (X_{i_j} \cup L_{i_j}^+)]$ for some $j \in [m']$, as the layers $L_{i_1},\dots,L_{i_{m'}}$ are non-adjacent by Fact~\ref{fact-diff1}.
In this case, $\pi^{-1}(\{v\})$ is contained in a connected component of $G_0'[L_{i_{j-1}}^+ \cup L_{>i_{j-1}}]$, which corresponds to a node $t \in T$ in the $(j-1)$-th level, and we have $\pi^{-1}(\{v\}) \subseteq V_t$.
So again, $v$ is contained in $\beta(t)$ and possibly the bags of some children of $t$.
So the nodes in $T$ whose bags contain $v$ are connected.

It suffices to verify that for any edge $vv'$ of $G_0'/(Z_1 \backslash Z')$, there exists a node in $T$ whose bag contains both $v$ and $v'$.
There exists $uu' \in E(G_0')$ such that $v = \pi(u)$ and $v' = \pi(u')$.
If $u,u' \in V_t$ for some $t \in V(T)$, then $v,v' \in \beta(t)$ and we are done.
Assume $\{u,u'\} \nsubseteq V_t$ for all $t \in V(T)$.
Let $j \in \{0\} \cup [m']$ be the largest index such that $\{u,u'\} \cap (L_{i_j}^+ \cup L_{>i_j}) \neq \emptyset$.
Without loss of generality, assume $u \in L_{i_j}^+ \cup L_{>i_j}$.
There exists a node $t \in V(T)$ in the $j$-th level such that $u$ is contained in $C_t$, the connected component of $G_0'[L_{i_j}^+ \cup L_{> i_j}]$ corresponding to $t$.
By the choice of $j$, we have $u,u' \notin L_{i_{j+1}}^+ \cup L_{>i_{j+1}}$, which implies $u \in V_t$ and $v \in \beta(t)$.
It remains to show $v' \in \beta(t)$.
If $u' \in C_t$, then $u' \in V_t$ as $u' \notin L_{i_{j+1}}^+ \cup L_{>i_{j+1}}$, which contradicts with the assumption $\{u,u'\} \nsubseteq V_t$.
So $u' \notin C_t$.
As such, $t$ is not the root of $t$ and has a parent $t'$ in $T$.
Observe that $u' \in C_{t'}$.
Indeed, $u' \in L_{>i_{j-1}}$ as it is a neighbor of $u$ and $u \in L_{\geq i_j}$.
Furthermore, as $uu' \in E(G_0')$ and $u,u' \in L_{>i_{j-1}}$, $u,u'$ belong to the same connected component of $G_0'[L_{i_{j-1}}^+ \cup L_{>i_{j-1}}]$, which is $C_{t'}$ because $u \in C_t \subseteq C_{t'}$.
On the other hand, $u' \notin L_{i_j}^+ \cup L_{>i_j}$, since $u' \in N_{G_0}(C_t)$ and $C_t$ is a connected component of $G_0'[L_{i_j}^+ \cup L_{> i_j}]$.
It follows that $u' \in V_{t'}$.
By construction, $U_t = V_{t'} \cap \kappa(f)$ for some face $f \in F_\eta(\partial o_{i_j}) \backslash \{o_{i_j}\}$ with $N_{G_0'}(C_t) \subseteq V(\partial f) \backslash L_{i_j}^+$.
So we have $u' \in N_{G_0'}(C_t) \subseteq V(\partial f) \backslash L_{i_j}^+$.
If $u' \in \kappa(f)$, then $u' \in V_{t'} \cap \kappa(f) = U_t$ and thus $v' \in \beta(t)$.
If $u' \notin \kappa(f)$, then $u' \in V(\partial f) \backslash (L_{i_j}^+ \cup X_{i_j} \cup Z')$, since $(V(\partial f) \backslash L_{i_j}^+) \cap (X_{i_j} \cup Z') \subseteq \kappa(f)$.
In this case, there exists $C \in \mathcal{C}_f$ such that $u' \in C$.
Recall that $\xi_C \in C$ is the representative of $C$.
Note that $C$ is contracted into one vertex in $G_0'/(Z_1 \backslash Z)$ and thus $\pi(\xi_C) = \pi(u') = v'$.
Since $C$ is connected, $C \subseteq C_{t'}$ and thus $\xi_C \in C_{t'}$.
As $\xi_C \notin L_{i_j}^+$, we further have $\xi_C \in V_{t'}$.
Also, $\xi_C \in \kappa(f)$ by definition.
Therefore, $\xi_C \in V_{t'} \cap \kappa(f) = U_t$, which implies $v' = \pi(\xi_C) \in \beta(t)$.

To see $\sigma(t) = \pi(U_t)$, suppose $t$ is at the $j$-th level of $T$.
If $j = 0$, then $t$ is the root and $\sigma(t) = \emptyset$ and $U_t = \emptyset$.
Otherwise, let $t' \in V(T)$ be the parent of $t$.
We have $\sigma(t) = \beta(t) \cap \beta(t') = (\pi(V_t) \cup \pi(U_t)) \cap \beta(t')$.
Consider a vertex $v \in \sigma(t)$, and assume $v \notin \pi(U_t)$.
Then we must have $v \in \pi(V_t) \cap \beta(t')$.
If $\pi^{-1}(\{v\})$ is a single vertex of $u \in V(G_0')$, then $u \in V_t$ and thus $u \notin V_{t'} \cup U_{t'}$, which implies $v \notin \beta(t')$, contradicting our assumption.
If $\pi^{-1}(\{v\})$ is a connected component of $G_0'[Z_1 \backslash Z']$, then it must be a connected component of $G_0'[L_{i_j} \backslash (X_j \cup L_{i_j}^+)]$, because $v \in \pi(V_t)$.
In this case, since $\pi^{-1}(\{v\}) \subseteq L_{i_j}$, $\pi^{-1}(\{v\}) \cap (V_{t'} \cup U_{t'}) = \emptyset$.
As such, $v \notin \beta(t')$, contradict our assumption.
Therefore, $v \in \pi(U_t)$ and $\sigma(t) \subseteq \pi(U_t)$.
The inclusion $\pi(U_t) \subseteq \sigma(t)$ is obvious, by the fact that $\pi(U_t) \subseteq \beta(t)$ and $\pi(U_t) \subseteq \beta(t')$.
Finally, since $|U_t| = O_{g,c}(|Z'|)$ as argued before, we have $|\sigma(t)| = O_{g,c}(|Z'|)$.
\end{proof}

As $(T,\beta)$ is a tree decomposition of $G_0'/(Z_1 \backslash Z')$, it suffices to show $\tw(\tor(t)) = O_{g,c}(p+|Z'|)$ for all $t \in V(T)$.
Consider a node $t \in V(T)$ in the $j$-th level of $T$.
In the above observation, it has been shown that $|\sigma(t)| = O_{g,c}(|Z'|)$.
So we only need to show $\tw(\tor(t) - \sigma(t)) = O_{g,c}(p+|Z'|)$.

For notational convenience, set $\phi = i_j$ and $\psi = i_{j+1}$.
What we are going to do is to build a graph that contains $\tor(t) - \sigma(t)$ as a minor, and bound the treewidth of that graph.
Let $(\varGamma,\eta)$ be a subgraph of $(G_0'[L_\phi^\psi],\eta)$ obtained by removing all edges in $E(L_\psi) \backslash E(\partial o_\psi)$.
Note that the edges removed are exactly those embedded in the interiors of the faces $f \in F_\eta(\partial o_{\psi}) \backslash \{o_{\psi}\}$.
This implies that every face in $F_\eta(\partial o_{\psi}) \backslash \{o_{\psi}\}$ is also a face of $(\varGamma,\eta)$.
Recall the set $\kappa(f) \subseteq V(\partial f) \backslash L_{\psi}^+$ defined for each $f \in F_\eta(\partial o_{\psi}) \backslash \{o_{\psi}\}$.
For convenience, we set $\kappa(f) = \emptyset$ for all faces $f$ of $(\varGamma,\eta)$ other than the ones in $F_\eta(\partial o_{\psi}) \backslash \{o_{\psi}\}$.
Clearly, $\kappa(f) \subseteq V(\partial f) = V(\varGamma)$ for all $f \in F_\eta(\varGamma)$.
We then define $\varGamma^\kappa$ as the graph obtained from $\varGamma$ by making $\kappa(f)$ a clique for all $f \in F_\eta(\varGamma)$.

\begin{observation} \label{obs-kappaminor}
$\varGamma^\kappa$ contains $\tor(t) - \sigma(t)$ as a minor.
\end{observation}
\begin{proof}
The vertex set of $\tor(t) - \sigma(t)$ is $\beta(t) \backslash \sigma(t)$.
By our construction, we have $\pi^{-1}(\beta(t)) = \pi^{-1}(\pi(V_t) \cup \pi(U_t)) = \pi^{-1}(\pi(V_t)) \cup \pi^{-1}(\pi(U_t))$.
It is clear that $\pi^{-1}(\pi(V_t)) = V_t$ and $V_t \cap \pi^{-1}(\pi(U_t)) = \emptyset$.
As $\sigma(t) = \pi(U_t)$ by Observation~\ref{obs-treedecomp}, we have
\begin{equation*}
    \pi^{-1}(\beta(t) \backslash \sigma(t)) = \pi^{-1}(\beta(t)) \backslash \pi^{-1}(\sigma(t)) = V_t \backslash \pi^{-1}(\pi(U_t)) = V_t.
\end{equation*}
Note that $V_t \subseteq L_\phi^\psi = V(\varGamma^\kappa)$.
Consider the map $\pi_{|V_t}: V_t \rightarrow \beta(t) \backslash \sigma(t)$, which is the restriction of $\pi$ on $V_t$.
By Fact~\ref{fact-minormap}, it suffices to show that $\varGamma^\kappa[\pi^{-1}(V)]$ is connected for all $V \subseteq \beta(t) \backslash \sigma(t)$ such that $(\tor(t) - \sigma(t))[V]$ is connected (note that $\pi^{-1}(V) = \pi_{|V_t}^{-1}(V)$).
Equivalently, this is saying that $\varGamma^\kappa[\pi^{-1}(\{v\})]$ is connected for all vertices $v \in \beta(t) \backslash \sigma(t)$ and $\varGamma^\kappa[\pi^{-1}(\{v,v'\})]$ for all edges $(v,v') \in E(\tor(t) - \sigma(t))$.

Consider a vertex $v \in \beta(t) \backslash \sigma(t)$.
If $\pi^{-1}(\{v\})$ is a single vertex, then $\varGamma^\kappa[\pi^{-1}(\{v\})]$ is connected.
Otherwise, $\pi^{-1}(\{v\})$ is a connected component of $G_0'[Z_1 \backslash Z']$.
In this case, $\pi^{-1}(\{v\}) \subseteq L_{\psi} \backslash L_{\psi}^+$.
As $G_0'[\pi^{-1}(\{v\})]$ is connected, to show $\varGamma^\kappa[\pi^{-1}(\{v\})]$ is connected, it suffices to show that for any edge $uu' \in E(G_0')$ such that $u,u' \in \pi^{-1}(\{v\})$, $u$ and $u'$ are contained in the same connected component of $\varGamma^\kappa[\pi^{-1}(\{v\})]$.
If $uu' \in E(\varGamma^\kappa)$, we are done.
Otherwise, $uu' \in E(L_{\psi}) \backslash E(\partial o_{\psi})$.
In this case, $uu'$ is embedded in some face $f \in F_\eta(\partial o_{\psi}) \backslash \{o_{\psi}\}$, and thus $u,u' \in V(\partial f) \backslash (L_{\psi}^+ \cup X_{\psi} \cup Z')$.
It follows that $u \in C$ and $u' \in C'$ for some $C,C' \in \mathcal{C}_f$.
We have the representative vertices $\xi_C$ and $\xi_{C'}$.
By definition, $\pi(C) = \pi(C') = \{v\}$.
So $u,u',\xi_C,\xi_{C'} \in \pi^{-1}(\{v\})$.
Note that $(\xi_C,\xi_{C'}) \in E(\varGamma^\kappa)$ because $\xi_C,\xi_{C'} \in \kappa(f)$.
So it remains to show that $u$ (resp., $u'$) is in the same connected component as $\xi_C$ (resp., $\xi_{C'}$) in $\varGamma^\kappa[\pi^{-1}(\{v\})]$.
This is true because $C$ and $C'$ are connected components of $\partial f - (L_{\psi}^+ \cup X_{\psi} \cup Z')$, which is a subgraph of $\varGamma^\kappa[\pi^{-1}(\{v\})]$.

Next, consider an edge $vv' \in E(\tor(t) - \sigma(t))$.
We have already shown that both $\varGamma^\kappa[\pi^{-1}(\{v\})]$ and $\varGamma^\kappa[\pi^{-1}(\{v'\})]$ are connected.
To see $\varGamma^\kappa[\pi^{-1}(\{v,v'\})]$ is connected, it suffices to show that $\pi^{-1}(\{v\})$ and $\pi^{-1}(\{v'\})$ are neighboring in $\varGamma^\kappa$.
There are two cases: $vv' \in E(G_0'/(Z_1 \backslash Z'))$ and $v,v' \in \sigma(s)$ for some child $s$ of $t$ in $T$.
Suppose $vv' \in E(G_0'/(Z_1 \backslash Z'))$.
Then $\pi^{-1}(\{v\})$ and $\pi^{-1}(\{v'\})$ are neighboring in $G_0'$.
So there exists an edge $uu' \in G_0'[V_t]$ such that $u \in \pi^{-1}(\{v\})$ and $u' \in \pi^{-1}(\{v'\})$.
If $uu' \in E(\varGamma^\kappa)$, we are done.
Otherwise, $uu'$ is embedded in some face $f \in F_\eta(\partial o_{\psi}) \backslash \{o_{\psi}\}$, and thus $u,u' \in V(\partial f) \backslash (L_{\psi}^+ \cup X_{\psi} \cup Z')$.
Now we can apply the same argument as above.
We have $u \in C$ and $u' \in C'$ for some $C,C' \in \mathcal{C}_f$ with representative vertices $\xi_C$ and $\xi_{C'}$.
By definition, $\pi(C) = \{v\}$ and $\pi(C') = \{v'\}$.
Thus, $\xi_C \in \pi^{-1}(\{v\})$ and $\xi_{C'} \in \pi^{-1}(\{v'\})$.
As $(\xi_C,\xi_{C'}) \in E(\varGamma^\kappa)$, we know that $\pi^{-1}(\{v\})$ and $\pi^{-1}(\{v'\})$ are neighboring in $\varGamma^\kappa$.
Next, suppose $v,v' \in \sigma(s)$ for some child $s$ of $t$.
In this case, $v,v' \in \pi(U_s) = \pi(V_t \cap \kappa(f)) \subseteq \pi(\kappa(f))$ for some $f \in F_\eta(\partial o_{\psi}) \backslash \{o_{\psi}\}$.
So there exist $u,u' \in \kappa(f)$ such that $u \in \pi^{-1}(\{v\})$ and $u' \in \pi^{-1}(\{v'\})$.
Note that $uu' \in E(\varGamma^\kappa)$ by the construction of $\varGamma^\kappa$.
So $\pi^{-1}(\{v\})$ and $\pi^{-1}(\{v'\})$ are neighbors in $\varGamma^\kappa$.
\end{proof}

Finally, it suffices to show that $\tw(\varGamma^\kappa) = O_{h,c}(p+|Z'|)$.
By Lemma~\ref{lem-twdiam2}, it suffices to bound $\mathsf{diam}_{w_\kappa}^*(\varGamma,\eta)$, where $w_\kappa: F_\eta(\varGamma) \rightarrow \mathbb{N}$ is the weight function defined as $w_\kappa(f) = |\kappa(f)|$.

\begin{observation}
$\mathsf{diam}_{w_\kappa}^*(\varGamma,\eta) = O_{h,c}(p+|Z'|)$.
\end{observation}
\begin{proof}
For convenience, we say a face of $(\varGamma,\eta)$ is \emph{heavy} if it is in $F_\eta(\partial o_\psi) \backslash \{o_\psi\}$, and \emph{light} otherwise.
The heavy faces of $(\varGamma,\eta)$ are exactly the ones whose $\kappa$-sets are nonempty, i.e., whose weight under $w_\kappa$ is nonzero.
The graphs $(\varGamma,\eta),(G_0'[L_\phi^\psi],\eta),(G_0[L_\phi^\psi],\eta)$ have the same outer face, which is $o_{\phi}$.
We show that for any vertex $v \in L_\phi^\psi$, there exists a VFA path in $(\varGamma,\eta)$ from $o_\phi$ to $v$ of (unweighted) length $O_h(p)$ that does not visit any heavy face.
This implies that the $w_\kappa$-weighted vertex-face distance between $o_\phi$ and any vertex $v \in L_\phi^\psi$ in $(\varGamma,\eta)$ is $O_h(p)$, and the $w_\kappa$-weighted vertex-face distance between $o_\phi$ and any face $f \in F_\eta(\varGamma)$ is $O_{h,c}(p+|Z'|)$ since $w_\kappa(f) = O_{h,c}(|Z'|)$ for all $f \in F_\eta(\varGamma)$.
It in turn follows that $\mathsf{diam}_{w_\kappa}^*(\varGamma,\eta) = O_{h,c}(p+|Z'|)$.

Consider a vertex $v \in L_i$ where $\phi \leq i \leq \psi$.
The vertex-face distance between the outer face $o$ of $(G_0,\eta)$ and $v$ is $2i-1$.
So there exists a VFA path $(f_1,v_1,\dots,f_i,v_i)$ in $(G_0,\eta)$ with $f_1,\dots,f_i \in F_\eta(G_0)$ and $v_1,\dots,v_i \in V(G_0)$ such that $f_1 = o$ and $v_i = v$.
This is a shortest VFA path from $o$ to $v$, and thus each subpath $(f_1,v_1,\dots,f_j,v_j)$ for $j \in [i]$ is also a shortest VFA path from $o$ to $v_j$.
It follows that $v_j \in L_j$ for all $j \in [i]$.
We claim that for all $j \in \{\phi+1,\dots,i\}$, there exists a VFA path in $(\varGamma,\eta)$ from $v_{j-1}$ to $v_j$ of length $O(h)$ which does not visit any heavy face.
By Fact~\ref{fact-diff1}, $\partial f_j \subseteq L_{j-1} \cup L_j$.
Therefore, $f_j$ is also a face of $(G_0[L_\phi^\psi],\eta)$.
Note that $f_j$ is not contained in any heavy face of $(\varGamma,\eta)$, as it is incident to $v_{j-1}$.
If $f_j$ is not a vortex face, then it is also a face of $\varGamma$, which is light.
Thus, $(v_{j-1},f_j,v_j)$ is the VFA path we want.
If $f_j$ is a vortex face, then it contains several faces of $(G_0'[L_\phi^\psi],\eta)$.
Let $F \subseteq F_\eta(G_0'[L_\phi^\psi])$ be the set of these faces.
By Observation~\ref{obs-invortex}, there exists a VFA path $\pi$ in $(G_0'[L_\phi^\psi],\eta)$ from $v_{j-1}$ to $v_j$ of length $O(h)$ that does not visit any face of $(G_0'[L_\phi^\psi],\eta)$ other than those in $F$.
Note that the faces in $F$ preserve in $(\varGamma,\eta)$, and they are light faces of $(\varGamma,\eta)$.
Therefore, $\pi$ is a VFA path in $(\varGamma,\eta)$ from $v_{j-1}$ to $v_j$ of length $O(h)$ which does not visit any heavy face.
It follows that there exists a VFA path in $(\varGamma,\eta)$ from $v_\phi$ to $v_i$ of length $O(h(i-\phi))$ which does not visit any heavy face.
We have $i-\phi \leq \psi - \phi \leq p$, and thus $h(i-\phi) = O_h(p)$.
Finally, since $v_\phi$ is incident to $o_\phi$, we see the existence of a VFA path in $(\varGamma,\eta)$ from $o_\phi$ to $v_i$ of length $O_h(p)$.
\end{proof}

Combining the above observation with Lemma~\ref{lem-twdiam2}, we have $\tw(\varGamma^\kappa) = O_{h,c}(p+|Z'|)$.
By Observation~\ref{obs-kappaminor}, we further have $\tw(\tor(t) - \sigma(t)) = O_{h,c}(p+|Z'|)$ and thus $\tw(\tor(t)) = O_{h,c}(p+|Z'|)$.
It follows that $\tw(G_0'/(Z_1 \backslash Z')) = O_{h,c}(p+|Z'|)$, according to Lemma~\ref{lem-torsotw}.
Finally, using Observation~\ref{obs-GtoG0'}, we have $\tw(G/(Z_1 \backslash Z')) = O_{h,c}(p+|Z'|)$, which completes the proof of Lemma~\ref{lem-decomp1}.

\subsection{Decomposing minor-free graphs}
In this section, we complete the proof of our main theorem (Theorem~\ref{thm-contraction}), based on the result of the previous section.
To this end, we need to introduce the famous Robertson-Seymour decomposition of an $H$-minor-free graph.
By simply combining the results of \cite{BandyapadhyayLLSJ22,DemaineHK05} with Lemma~\ref{lem-minimal}, we can obtain the following strong version of Robertson-Seymour decomposition.

\begin{lemma}[Robertson-Seymour decomposition] \label{lem-rsdecomp}
Let $H$ be a fixed graph.
Any connected $H$-minor-free graph $G$ admits a tree decomposition $(T,\beta)$ satisfying the following two conditions (where $h>0$ is a constant only depending on $H$).
\begin{enumerate}[label = (RS.\arabic*)]
    \item\label{item:rs-decomposition-1} For all $t \in V(T)$, $\tor(t)$ admits an $h$-almost-embeddable structure in which the partial embedding is minimal, and for every child $s$ of $t$ in $T$, all but at most three vertices in $\sigma(s)$ are vortex vertices or apices of $\tor(t)$.
    Also, $|\sigma(t)| \leq h$ and all vertices in $\sigma(t)$ are apices of $\tor(t)$.
    \item\label{item:rs-decomposition-2} For all $t \in V(T)$, $G[\gamma(t) \backslash \sigma(t)]$ is connected, and $\sigma(t) \subseteq N_G(\gamma(t) \backslash \sigma(t))$.
\end{enumerate}
Furthermore, given an $H$-minor-free graph $G$, the tree decomposition (together with the almost-embeddable structures of the torsos) can be computed in polynomial time.
\end{lemma}

\begin{proof}
The existence of a tree decomposition $(T,\beta)$ of $G$ in which $\tor(t)$ is $h$-almost-embeddable and $|\sigma(t)| \leq h$ for all $t \in V(T)$ follows from the profound work of Robertson and Seymour \cite{RobertsonS03a}.
Polynomial-time algorithms for computing a Robertson-Seymour decomposition (and the almost-embeddable structures of the torsos) are also known \cite{DemaineHK05,GroheKR13,KawarabayashiW11}.
The algorithm in \cite{DemaineHK05} guarantees the tree decomposition satisfies condition \ref{item:rs-decomposition-1} except that the partial embeddings of the $h$-almost-embeddable structure of $\tor(t)$ are not necessarily minimal.
The recent work \cite{BandyapadhyayLLSJ22} (Lemma~2.2, or Lemma~4 in the arxiv version) shows how to modify a given tree decomposition of $G$ to make it satisfy condition \ref{item:rs-decomposition-2}.
Each torso $\tor(t)$ of the modified tree decomposition is a subgraph of a torso $\tor(t')$ of the original tree decomposition satisfying $\sigma(t) \subseteq \sigma(t')$, and for every child $s$ of $t$, $\sigma(s) \subseteq \sigma(s')$ for some child $s'$ of $t'$.
Also, it is shown in \cite{BandyapadhyayLLSJ22} (last paragraph of the proof of Lemma~4 in the arxiv version) that given a graph $P$ with an $h$-almost-embeddable structure, for every subgraph $Q$ of $P$, one can compute a $O_h(1)$-almost-embeddable structure in which the vortex vertices (resp., apices) of $Q$ are exactly the vortex vertices (resp., apices) of $P$ that are preserved in $Q$.
Therefore, if we apply the modification of \cite{BandyapadhyayLLSJ22} to a tree decomposition $(T,\beta)$ of $G$ that already satisfies \ref{item:rs-decomposition-1} except the minimality of the partial embeddings, the new tree decomposition satisfies both conditions (except the minimality of the partial embeddings).
Finally, we use Lemma~\ref{lem-minimal} to obtain a new $O_h(1)$-almost-embeddable structure of each torso $\tor(t)$ in which the partial embedding is minimal.
We claim that condition \ref{item:rs-decomposition-1} and \ref{item:rs-decomposition-2} hold using the new almost-embeddable structures.
Condition \ref{item:rs-decomposition-2} still holds, as it is independent of the almost-embeddable structures of the torsos.
By Lemma~\ref{lem-minimal}, the apices in the old structure are also apices in the new structure and the vortex vertices in the old structure become vortex vertices or apices in the new structure.
Therefore, condition \ref{item:rs-decomposition-1} holds with the new almost-embeddable structures of $\tor(t)$.
\end{proof}

Let $G$ be an $H$-minor-free graph.
Without loss of generality, we can assume $G$ is connected.
We first compute the tree decomposition $(T,\beta)$ of $G$ in Lemma~\ref{lem-rsdecomp} as well as the almost-embeddable structures of the torsos.
Condition \ref{item:rs-decomposition-2} of Lemma~\ref{lem-rsdecomp} implies the following property of the tree decomposition $(T,\beta)$.

\begin{observation} \label{obs-noadhtos}
For all $t \in V(T)$, $\tor(t) - \sigma(t)$ is connected and $\sigma(t) \subseteq N_{\tor(t)}(\beta(t) \backslash \sigma(t))$.
\end{observation}
\begin{proof}
By condition \ref{item:rs-decomposition-2} of Lemma~\ref{lem-rsdecomp}, $G[\gamma(t) \backslash \sigma(t)]$ is connected.
So by Fact~\ref{fact-connintorso}, either $\tor(t)[\beta(t) \cap (\gamma(t) \backslash \sigma(t))] = \tor(t) - \sigma(t)$ is connected or every connected component of $\tor(t) - \sigma(t)$ intersects $\sigma(t)$.
The latter is false, and thus $\tor(t) - \sigma(t)$ is connected.
To see $\sigma(t) \subseteq N_{\tor(t)}(\beta(t) \backslash \sigma(t))$, consider a vertex $v \in \sigma(t)$.
Again, by condition \ref{item:rs-decomposition-2} of Lemma~\ref{lem-rsdecomp}, $v \in \sigma(t) \subseteq N_G(\gamma(t) \backslash \sigma(t))$.
This implies $G[(\gamma(t) \backslash \sigma(t)) \cup \{v\}]$ is connected, because $G[\gamma(t) \backslash \sigma(t)]$ is connected.
So by Fact~\ref{fact-connintorso}, either $\tor(t)[(\gamma(t) \backslash \sigma(t)) \cup \{v\}]$ is connected or every connected component of $\tor(t)[(\gamma(t) \backslash \sigma(t)) \cup \{v\}]$ intersects $\sigma(t)$.
If $\tor(t)[(\gamma(t) \backslash \sigma(t)) \cup \{v\}]$ is connected, we directly have $v \in N_{\tor(t)}(\beta(t) \backslash \sigma(t))$.
If every connected component of $\tor(t)[(\gamma(t) \backslash \sigma(t)) \cup \{v\}]$ intersects $\sigma(t)$, then $\tor(t)[(\gamma(t) \backslash \sigma(t)) \cup \{v\}]$ can only have one connected component, because $(\gamma(t) \backslash \sigma(t)) \cup \{v\}$ only intersects $\sigma(t)$ at one vertex, $v$.
Thus, $\tor(t)[(\gamma(t) \backslash \sigma(t)) \cup \{v\}]$ is connected and $v \in N_{\tor(t)}(\beta(t) \backslash \sigma(t))$.
\end{proof}

In order to construct the sets $Z_1,\dots,Z_p \subseteq V(G)$, we consider every node $t \in V(T)$, and compute $p$ sets $Z_1^{(t)},\dots,Z_p^{(t)} \subseteq \beta(t) \backslash \sigma(t)$ as follows.
Let $t \in V(T)$ be a node and $A$ be the apex set of $\tor(t)$.
We say an edge $aa'$ of $\tor(t)[A]$ is \textit{redundant} if there exists a vertex $v \in \beta(t) \backslash A$ adjacent to both $a$ and $a'$ in $\tor(t)$.
Define $G_t$ as the graph obtained from $\tor(t)$ by removing all redundant edges.
The next observation shows that $G_t$ ``inherits'' the nice properties of $\tor(t)$ in Observation~\ref{obs-noadhtos}.

\begin{observation} \label{obs-noadhGt}
For all $t \in V(T)$, $G_t - \sigma(t)$ is connected and $\sigma(t) \subseteq N_{G_t}(\beta(t) \backslash \sigma(t))$.
\end{observation}
\begin{proof}
By Observation~\ref{obs-noadhtos}, $\tor(t) - \sigma(t)$ is connected.
Because $G_t$ is obtained by removing redundant edges from $\tor(t)$, to show $G_t - \sigma(t)$ is connected, it suffices to show that for any redundant edge $(a,a')$ where $a,a' \in \beta(t) \backslash \sigma(t)$, $a$ and $a'$ are in the same connected component of $G_t - \sigma(t)$.
If $aa'$ is redundant, by definition there exists $v \in \beta(t) \backslash A$ neighboring to both $a$ and $a'$ in $\tor(t)$.
By condition \ref{item:rs-decomposition-1} of Lemma~\ref{lem-rsdecomp}, $\sigma(t) \subseteq A$ and thus $v \in \beta(t) \backslash \sigma(t)$.
Since $v \notin A$, the edges $av$ and $a'v$ are both not redundant and thus preserve in $G_t - \sigma(t)$.
Therefore, $a$ and $a'$ lie in the same connected component of $G_t - \sigma(t)$.

To see $\sigma(t) \subseteq N_{G_t}(\beta(t) \backslash \sigma(t))$, consider a vertex $a \in \sigma(t)$.
By Observation~\ref{obs-noadhtos}, $\sigma(t) \subseteq N_{\tor(t)}(\beta(t) \backslash \sigma(t))$.
Thus, $a$ has a neighbor $a' \in \beta(t) \backslash \sigma(t)$ in $\tor(t)$.
If $aa'$ is not redundant, then it preserves in $G_t$, which implies $a \in N_{G_t}(\beta(t) \backslash \sigma(t))$.
Otherwise, there exists $v \in \beta(t) \backslash A$ such that $av \in E(\tor(t))$.
Note that $av$ is not redundant as $v \notin A$.
By condition \ref{item:rs-decomposition-1} of Lemma~\ref{lem-rsdecomp}, $\sigma(t) \subseteq A$ and thus $v \in \beta(t) \backslash \sigma(t)$.
Therefore, $a \in N_{G_t}(\beta(t) \backslash \sigma(t))$.
\end{proof}

For each $v \in \sigma(t)$, we pick an arbitrary vertex $v' \in \beta(t) \backslash \sigma(t)$ adjacent to $v$ in $G_t$ and call $v'$ the \textit{witness neighbor} of $v$ (by the above observation, such a witness neighbor exists).
Define $\varPhi^{(t)} \subseteq \beta(t) \backslash \sigma(t)$ as the set of witness neighbors of all vertices in $\sigma(t)$.
We have $|\varPhi^{(t)}| \leq |\sigma(t)| \leq h$ by \ref{item:rs-decomposition-1} of Lemma~\ref{lem-rsdecomp}.
Next, we apply Corollary~\ref{cor-decomp2} on $G = G_t - \sigma(t)$ with the set $\varPhi = \varPhi^{(t)}$, and define $Z_1^{(t)},\dots,Z_p^{(t)}$ as the sets $Z_1,\dots,Z_p$ obtained by the corollary.
Note that $G_t - \sigma(t)$ is obtained from $\tor(t)$ by deleting some apices and some edges between apices, so it inherits the almost-embeddable structure of $\tor(t)$, in which the partial embedding is minimal.
Also, $G_t - \sigma(t)$ is connected by Observation~\ref{obs-noadhtos}.
Therefore, Corollary~\ref{cor-decomp2} is applicable on $G_t - \sigma(t)$.

Next, we define a coloring $\mathsf{col}\colon V(T) \rightarrow [p]$ which assigns each node of $T$ one of the colors $1,\dots,p$.
The coloring is defined in a top-down manner in $T$.
If $t$ is the root, define $\mathsf{col}(t)$ as an arbitrary number in $[p]$.
Otherwise, let $t'$ be the parent of $t$ and suppose $\mathsf{col}(t')$ is already defined.

\begin{observation}
There exists at most one index $i \in [p]$ such that the clique $\tor(t')[\sigma(t)]$ contains an edge of $G_{t'}[Z_i^{(t')}]$.
\end{observation}
\begin{proof}
By \ref{item:rs-decomposition-1} of Lemma~\ref{lem-rsdecomp}, $\tor(t')[\sigma(t)]$ is a clique in $\tor(t')$ in which all but at most three vertices are vortex vertices and apices of $\tor(t')$.
On the other hand, by Corollary~\ref{cor-decomp2}, $Z_1^{(t')},\dots,Z_p^{(t')}$ do not contain vortex vertices or apices of $G_{t'}$, and thus do not contain vortex vertices or apices of $\tor(t')$.
Therefore, $\bigcup_{i=1}^p Z_i^{(t')}$ contains at most three vertices in $\sigma(t)$.
As $Z_1^{(t')},\dots,Z_p^{(t')}$ are disjoint, it follows that there exists at most one index $i \in [p]$ such that $|\sigma(t) \cap Z_i^{(t')}| \geq 2$.
This proves the observation since $\tor(t')[\sigma(t)]$ contains an edge of $G_{t'}[Z_i^{(t')}]$ only if $|\sigma(t) \cap Z_i^{(t')}| \geq 2$.
\end{proof}
If there exists some $i \in [p]$ such that $\tor(t')[\sigma(t)]$ contains an edge of $G_{t'}[Z_i^{(t')}]$, we then set $\mathsf{col}(t) = i$; by the above observation, such an index $i$ is unique (if it exists).
Otherwise, $\tor(t')[\sigma(t)]$ is edge-disjoint from all of $\tor(t')[Z_1^{(t')}],\dots,\tor(t')[Z_p^{(t')}]$.
In this case, we set $\mathsf{col}(t) = \mathsf{col}(t')$.

Finally, we construct the partition $Z_1,\dots,Z_p$ in Theorem~\ref{thm-contraction} as follows.
Let $T_i = \mathsf{col}^{-1}(\{i\})$ be the set of nodes of $T$ colored with $i$.
For each node $t \in V(T)$, let $R^{(t)}$ be the complement of $\bigcup_{i=1}^p Z_i^{(t)}$ in $\beta(t) \backslash \sigma(t)$, i.e., $R^{(t)} \coloneqq (\beta(t) \backslash \sigma(t)) \backslash (\bigcup_{i=1}^p Z_i^{(t)})$.
We then define
\begin{equation}
 Z_i \coloneqq \left(\bigcup_{t \in T} Z_i^{(t)}\right) \cup \left(\bigcup_{t \in T_i} R^{(t)}\right)
\end{equation}
for every $i \in [p]$.
From our construction, one can easily verify the following fact.
\begin{observation}
$Z_1,\dots,Z_p$ is a partition of $V(G)$.
\end{observation}
\begin{proof}
Consider a vertex $v \in V(G)$.
We need to show that there exists a unique index $i \in [p]$ such that $v \in Z_i$.
By Fact~\ref{fact-partition}, $\{\beta(t) \backslash \sigma(t) \mid t \in V(T)\}$ forms a partition of $V(G)$.
So there exists a unique node $t \in V(T)$ such that $v \in \beta(t) \backslash \sigma(t)$.
By our construction, $Z_1^{(t)},\dots,Z_p^{(t)},R^{(t)}$ is a partition of $\beta(t) \backslash \sigma(t)$.
If $v \in Z_i^{(t)}$, then $v \in Z_i$ and $v \notin Z_j$ for all $j \in [p]$ other than $i$.
If $v \in R^{(t)}$, then $v \in Z_i$ for $i = \mathsf{col}(t)$ and $v \notin Z_j$ for all $j \in [p]$ other than $i$ (simply because $t \notin T_j$).
\end{proof}

\begin{observation} \label{obs-connectedbottom}
Let $i \in [p]$ and $t \in V(T)$.
For any edge $vv'$ of $\tor(t)[Z_i^{(t)}]$, $v$ and $v'$ lie in the same connected component of $G[(\bigcup_{s \in S} (\gamma(s) \backslash \sigma(s)) \cap Z_i) \cup \{v,v'\}]$, where $S$ consists of all children $s$ of $t$ in $T$ satisfying $v,v' \in \sigma(s)$.
\end{observation}
\begin{proof}
We first prove the following statement, and show it implies the observation.

\begin{claim}
 \label{claim:realize-fake-edge}
 For all $i \in [p]$ and $t \in T_i$, any two vertices $v,v' \in \sigma(t) \cap Z_i$ lie in the same connected component of $G[((\gamma(t) \backslash \sigma(t)) \cap Z_i) \cup \{v,v'\}]$.
\end{claim}
\begin{claimproof}
We shall apply induction on the depth of $t$ in $T$.
Suppose the statement holds for all children of $t$, and we show it also holds for $t$.
Let $Z(t,v,v') \coloneqq ((\gamma(t) \backslash \sigma(t)) \cap Z_i) \cup \{v,v'\}$.
Also, let $u \in \varPhi^{(t)}$ (resp., $u' \in \varPhi^{(t)}$) be the witness neighbor of $v$ (resp., $v'$).
Note that $u,u' \in \varPhi^{(t)} \subseteq R^{(t)}$.
Since $t \in T_i$, we have $u,u' \in Z_i$ and thus $u,u' \in Z(t,v,v')$.

We first show that $u$ (resp., $u'$) and $v$ (resp., $v'$) lie in the same connected component of $G[Z(t,v,v')]$.
Without loss of generality, it suffices to consider $u$ and $v$.
Note that $uv$ is an edge of $G_t$.
If $uv \in E(G)$, we are done.
Otherwise, $u,v \in \sigma(s)$ for some child $s$ of $t$.
We claim that $\sigma(s)$ is disjoint from all of $Z_1^{(t)},\dots,Z_p^{(t)}$.
Let $A$ be the apex set of $\tor(t)$.
We have $v \in \sigma(t) \subseteq A$ by condition \ref{item:rs-decomposition-1} of Lemma~\ref{lem-rsdecomp}.
If $u \in A$, then $uv$ is not a redundant edge of $\tor(t)$ as it preserves in $G_t$.
In this case, no vertex in $\beta(t) \backslash \sigma(t)$ is adjacent to both $u$ and $v$ in $\tor(t)$.
It follows that $\sigma(s) \subseteq \sigma(t)$, because $\tor(t)[\sigma(s)]$ is a clique.
As such, $\sigma(s)$ is disjoint from all of $Z_1^{(t)},\dots,Z_p^{(t)}$.
If $u \notin A$, then $u \in \varPhi^{(t)} \backslash (A \backslash \sigma(t))$.
Note that $A \backslash \sigma(t)$ is the apex set of $G_t - \sigma(t)$.
By condition \ref{item:decomp2-3} of Corollary~\ref{cor-decomp2}, for all $i \in [p]$, no vertex in $Z_i^{(t)}$ is neighboring to $\varPhi^{(t)} \backslash (A \backslash \sigma(t))$ in $G_t - \sigma(t)$.
In particular, $u$ is not neighboring to $Z_i^{(t)}$ in $G_t - \sigma(t)$ (and thus in $\tor(t)$).
This implies that $\sigma(s)$ is disjoint from all of $Z_1^{(t)},\dots,Z_p^{(t)}$, as $\tor(t)[\sigma(s)]$ is a clique.
We then have $\tor(t)[\sigma(s)]$ is edge-disjoint from all of $\tor(t)[Z_1^{(t)}],\dots,\tor(t)[Z_p^{(t)}]$.
Therefore, $\mathsf{col}(s) = \mathsf{col}(t) = i$, and $s \in T_i$.
As $v \in Z_i$ and $u \in \varPhi^{(t)} \subseteq R^{(t)} \subseteq Z_i$, we have $u,v \in \sigma(s) \cap Z_i$.
By our induction hypothesis, $u$ and $v$ are in the same connected component of $G[Z(s,u,v)]$.
Note that $Z(s,u,v) \subseteq Z(t,v,v')$, so $u$ and $v$ lie in the same connected component of $G[Z(t,v,v')]$.
For the same reason, $u'$ and $v'$ lie in the same connected component of $G[Z(t,v,v')]$.

In order to show $v$ and $v'$ lie in the same connected component of $G[Z(t,v,v')]$, it now suffices to show $u$ and $u'$ lie in the same connected component of $G[Z(t,v,v')]$.
We have $u,u' \in \varPhi^{(t)}$, and by condition \ref{item:decomp2-2} of Corollary~\ref{cor-decomp2}, $u$ and $u'$ are in the same connected component of $G'$, the graph obtained from $G_t - (\bigcup_{j=1}^p Z_j^{(t)} \cup \sigma(t))$ by deleting the ``bad'' edges, which are those connecting one vertex in $A \backslash \sigma(t)$ and one neighbor of $\bigcup_{j=1}^p Z_j^{(t)}$ that is not in $A \backslash \sigma(t)$.
So we only need to show that the two endpoints of every edge of $G'$ are in the same connected component of $G[Z(t,v,v')]$.
Consider an edge $ww' \in E(G')$.
If $ww' \in E(G)$, we are done.
Otherwise, $w,w' \in \sigma(s)$ for some child $s$ of $t$.
We claim that $\tor(t)[\sigma(s)]$ is edge-disjoint from all of $G_t[Z_1^{(t)}],\dots,G_t[Z_p^{(t)}]$.
Assume this is not the case, and assume $\tor(t)[\sigma(s)]$ contains an edge $e$ of $G_t[Z_1^{(t)}]$ without loss of generality.
Then $w,w'$ and the two endpoints of $e$ form a clique in $\tor(t) - \sigma(t)$.
Since the endpoints of $e$ are not in $A$, an edge connecting $w$ (resp., $w'$) and an endpoint of $e$ preserves in $G_t$.
Therefore, in $G_t - \sigma(t)$, both $w$ and $w'$ are neighboring to the two endpoints of $e$.
Let $\mathsf{Vort} \subseteq \beta(t)$ be the set of vortex vertices of $\tor(t)$, which is also the set of vortex vertices of $G_t - \sigma(t)$.
Observe that $w,w' \notin \mathsf{Vort}$, because $Z_1^{(t)}$ is not neighboring to $\mathsf{Vort}$ in $G_t - \sigma(t)$ by Corollary~\ref{cor-decomp2}.
If $w,w' \in A$, then $ww'$ is not redundant as $ww' \in E(G_t)$.
By definition, $w$ and $w'$ do not have common neighbors in $\beta(t) \backslash A$.
But this contradicts with the fact that the endpoints of $e$ are in $\beta(t) \backslash A$ and are neighboring to both $w$ and $w'$.
If $w,w' \notin A$, then $w,w' \notin \mathsf{Vort} \cup A$.
However, this contradicts with condition \ref{item:rs-decomposition-1} of Lemma~\ref{lem-rsdecomp}, because $\sigma(s)$ contains at least four vertices outside $\mathsf{Vort} \cup A$, i.e., $w,w'$, and the two endpoints of $e$ (which are vertices in $Z_1^{(t)}$ and are not in $\mathsf{Vort} \cup A$ by Corollary~\ref{cor-decomp2}).
Therefore, we must have one of $w,w'$ in $A$ and the other not in $A$.
Without loss of generality, assume $w \in A$ and $w' \notin A$.
It follows that $w'$ is not neighboring to $\bigcup_{j=1}^p Z_j^{(t)}$ in $G_t - \sigma(t)$, for otherwise $ww'$ is a ``bad'' edge and cannot preserve in $G'$.
But this contradicts with the fact that $w'$ is neighboring to the endpoints of $e$ in $G_t - \sigma(t)$, which are vertices in $Z_1^{(t)}$.
As a result, $\tor(t)[\sigma(s)]$ is edge-disjoint from all of $G_t[Z_1^{(t)}],\dots,G_t[Z_p^{(t)}]$, which implies $\mathsf{col}(s) = \mathsf{col}(t) = i$.
Note that $w,w' \in R^{(t)} \subseteq Z_i$ as $t \in T_i$, and thus $v,v' \in \sigma(s) \cap Z_i$.
By our induction hypothesis, $w$ and $w'$ are in the same connected component of $G[Z(s,w,w')]$.
Note that $Z(s,w,w') \subseteq Z(t,v,v')$, so $u$ and $v$ lie in the same connected component of $G[Z(t,v,v')]$.
\end{claimproof}

Finally, we show the statement in the observation.
Consider an edge $vv'$ of $\tor(t)[Z_i^{(t)}]$.
Let $S$ be the set of children $s$ of $t$ satisfying $v,v' \in \sigma(s)$.
If $vv' \in E(G)$, we are done.
Otherwise, $S \neq \emptyset$ and we arbitrarily pick $s \in S$.
The clique $\tor(t)[\sigma(s)]$ is not edge-disjoint from $\tor(t)[Z_i^{(t)}]$ as it contains the edge $vv'$.
Thus, $s \in T_i$.
Since $v,v' \in \sigma(s) \cap Z_i$, by Claim \ref{claim:realize-fake-edge}, $v$ and $v'$ lie in the same connected component of $G[Z(s,v,v')]$.
Recall that $Z(s,v,v') = ((\gamma(s) \backslash \sigma(s)) \cap Z_i) \cup \{v,v'\}$.
So $v$ and $v'$ are in the same connected component of $G[(\bigcup_{s \in S} (\gamma(s) \backslash \sigma(s)) \cap Z_i) \cup \{v,v'\}]$.
\end{proof}

Next, we show that $\tw(G/(Z_i \backslash Z')) = O_h(p+|Z'|)$ for all $i \in [p]$ and $Z' \subseteq Z_i$.
Without loss of generality, we only need to consider the case $i=1$.
Let $\pi\colon V(G) \rightarrow V(G/(Z_1 \backslash Z'))$ be the quotient map for the contraction of $G$ to $G/(Z_1 \backslash Z')$.
For each $t \in V(T)$, let $\beta^*(t) = \pi(\beta(t))$.
By Fact~\ref{fact-inducedtd}, $(T,\beta^*)$ is a tree decomposition of $G/(Z_1 \backslash Z')$.
We use the notations $\sigma^*(t)$, $\gamma^*(t)$, $\tor^*(t)$ to denote the adhesion, $\gamma$-set, and torso of a node $t \in V(T)$ in the tree decomposition $(T,\beta^*)$, in order to distinguish from those for $(T,\beta)$.
We observe the following simple fact.

\begin{observation} \label{obs-adhesion}
For all $t \in V(T)$, $\sigma^*(t) = \pi(\sigma(t))$.
\end{observation}
\begin{proof}
If $t$ is the root of $t$, then $\sigma^*(t) = \sigma(t) = \emptyset$.
Otherwise, let $t'$ be the parent of $t$.
By definition, $\sigma^*(t) = \beta^*(t) \cap \beta^*(t') = \pi(\beta(t)) \cap \pi(\beta(t'))$ and $\sigma(t) = \beta(t) \cap \beta(t')$.
So it is clear that $\pi(\sigma(t)) \subseteq \sigma^*(t)$.
To see $\sigma^*(t) \subseteq \pi(\sigma(t))$, consider a vertex $v \in \pi(\beta(t)) \cap \pi(\beta(t'))$.
Since $v \in \pi(\beta(t)) \cap \pi(\beta(t'))$, $\pi^{-1}(\{v\}) \cap \beta(t) \neq \emptyset$ and $\pi^{-1}(\{v\}) \cap \beta(t') \neq \emptyset$.
Furthermore, by Fact~\ref{fact-quotient}, $G[\pi^{-1}(\{v\})]$ is connected.
Therefore, $\pi^{-1}(\{v\}) \cap \sigma(t) \neq \emptyset$ and hence $v \in \pi(\sigma(t))$.
\end{proof}

By Lemma~\ref{lem-torsotw}, it suffices to show that $\tw(\tor^*(t)) = O_h(p+|Z'|)$ for all $t \in V(T)$.
Consider a node $t \in V(T)$.
In order to bound $\tw(\tor^*(t))$, we shall construct a graph of bounded treewidth that contains $\tor^*(t)$ as a minor.
Define $\mathsf{TOR}(t)$ as the graph obtained from $\tor(t)$ by making $\sigma(t)$ a clique.
Let $S_\mathsf{bad}$ be the set of children $s$ of $t$ satisfying $(\gamma(s) \backslash \sigma(s)) \cap Z' \neq \emptyset$.
Note that $|S_\mathsf{bad}| \leq |Z'|$ as the sets $\gamma(s) \backslash \sigma(s)$ are disjoint for the children $s$ of $t$.
Set $Z'' = Z' \cup (\bigcup_{s \in S_\mathsf{bad}} \sigma(s))$.
Since $|\sigma(s)| \leq h$ for all $s \in S_\mathsf{bad}$ and $|S_\mathsf{bad}| \leq |Z'|$, by \ref{item:rs-decomposition-1} of Lemma~\ref{lem-rsdecomp}, we have $|Z''| = O_h(|Z'|)$.
Consider the graph $G^* = \mathsf{TOR}(t)/(Z_1^{(t)} \backslash (Z'' \cap Z_1^{(t)}))$.
Because $Z_1^{(t)} \subseteq \beta(t) \backslash \sigma(t)$, the vertices in $\sigma(t)$ do not get contracted in $G^*$.
So the vertices in $\sigma(t)$ are also vertices in $G^*$.
We observe that $\tw(G^*) = O_h(p+|Z'|)$.
Indeed, $G^* - \sigma(t)$ is isomorphic to $(\tor(t)-\sigma(t))/(Z_1^{(t)} \backslash (Z'' \cap Z_1^{(t)}))$, and $(G_t - \sigma(t))/(Z_1^{(t)} \backslash (Z'' \cap Z_1^{(t)}))$ can be obtained from $(\tor(t)-\sigma(t))/(Z_1^{(t)} \backslash (Z'' \cap Z_1^{(t)}))$ by deleting the $O_h(1)$ redundant edges.
Therefore, $(G_t - \sigma(t))/(Z_1^{(t)} \backslash (Z'' \cap Z_1^{(t)}))$ can be viewed as a graph obtained from $G^*$ by deleting $O_h(1)$ vertices and edges, since $\sigma(t) \leq h$ by \ref{item:rs-decomposition-1} of Lemma~\ref{lem-rsdecomp}.
So the difference between $\tw((G_t - \sigma(t))/(Z_1^{(t)} \backslash (Z'' \cap Z_1^{(t)})))$ and $\tw(G^*)$ is $O_h(1)$.
Since $|Z''| = O_h(|Z'|)$, by condition \ref{item:decomp2-1} of Corollary~\ref{cor-decomp2}, we have
\begin{equation*}
    \tw((G_t - \sigma(t))/(Z_1^{(t)} \backslash (Z'' \cap Z_1^{(t)}))) = O_h(p+|Z'|),
\end{equation*}
which implies $\tw(G^*) = O_h(p+|Z'|)$.
Finally, we show that $G^*$ contains $\tor^*(t)$ as a minor.
To this end, we need the following observation.

\begin{observation} \label{obs-TORconnected}
For any $V \subseteq \beta^*(t)$, if $\tor^*(t)[V]$ is connected, then $\mathsf{TOR}(t)[\pi_{|\beta(t)}^{-1}(V)]$ is also connected, where $\pi_{|\beta(t)}: \beta(t) \rightarrow \beta^*(t)$ is the map $\pi$ restricted to $\beta(t)$.
\end{observation}
\begin{proof}
The statement in the observation is equivalent to saying that $\mathsf{TOR}(t)[\pi_{|\beta(t)}^{-1}(\{v\})]$ is connected for all $v \in \beta^*(t)$ and $\mathsf{TOR}(t)[\pi_{|\beta(t)}^{-1}(\{v,v'\})]$ is connected for all edges $vv' \in E(\tor^*(t))$.
Consider a vertex $v \in \beta^*(t)$.
We know that $G[\pi^{-1}(\{v\})]$ is connected and $\pi_{|\beta(t)}^{-1}(\{v\}) = \pi^{-1}(\{v\}) \cap \beta(t)$.
By Fact~\ref{fact-connintorso}, $\tor(t)[\pi_{|\beta(t)}^{-1}(\{v\})]$ is connected or every connected component of $\tor(t)[\pi_{|\beta(t)}^{-1}(\{v\})]$ intersects $\sigma(t)$.
In the former case, we are done, because $\mathsf{TOR}(t)[\pi_{|\beta(t)}^{-1}(\{v\})]$ can be obtained from $\tor(t)[\pi_{|\beta(t)}^{-1}(\{v\})]$ by adding edges.
In the latter case, we can also deduce that $\mathsf{TOR}(t)[\pi_{|\beta(t)}^{-1}(\{v\})]$ is connected, as every connected component of $\tor(t)[\pi_{|\beta(t)}^{-1}(\{v\})]$ intersects $\sigma(t)$ and $\mathsf{TOR}(t)[\sigma(t)]$ is a clique.
Next, consider an edge $vv' \in E(\tor^*(t))$.
If $vv'$ is an edge of $G/(Z_1 \backslash Z')$, then the graph $G[\pi^{-1}(\{v,v'\})]$ is connected by Fact~\ref{fact-quotient} and $\pi_{|\beta(t)}^{-1}(\{v,v'\}) = \pi^{-1}(\{v,v'\}) \cap \beta(t)$.
In this case, we can apply exactly the same argument as above to show that $\mathsf{TOR}(t)[\pi_{|\beta(t)}^{-1}(\{v,v'\})]$ is connected.
Otherwise, $v,v' \in \sigma^*(s)$ for some child $s$ of $t$.
By Observation~\ref{obs-adhesion}, $\sigma^*(s) = \pi(\sigma(s))$.
Thus, both $\pi_{|\beta(t)}^{-1}(\{v\})$ and $\pi_{|\beta(t)}^{-1}(\{v'\})$ intersect $\sigma(s)$.
We have already shown that $\mathsf{TOR}(t)[\pi_{|\beta(t)}^{-1}(\{v\})]$ and $\mathsf{TOR}(t)[\pi_{|\beta(t)}^{-1}(\{v'\})]$ are connected.
As both of the them intersect $\sigma(s)$ and $\mathsf{TOR}(t)[\sigma(s)]$ is a clique, we deduce that $\mathsf{TOR}(t)[\pi_{|\beta(t)}^{-1}(\{v,v'\})]$ is connected.
\end{proof}

\begin{observation}
$G^*$ contains $\tor^*(t)$ as a minor and thus $\tw(\tor^*(t)) = O_h(p+|Z'|)$.
\end{observation}
\begin{proof}
Let $\theta\colon \beta(t) \to V(G^*)$ be the quotient map for the contraction of $\mathsf{TOR}(t)$ to $G^*$.
We show that if two vertices $v,v' \in \beta(t)$ satisfying $\theta(v) = \theta(v')$, then $\pi(v) = \pi(v')$.
If $v = v'$, we are done.
Otherwise, since $\theta(v) = \theta(v')$, $v$ and $v'$ lie in the same connected component of $\mathsf{TOR}(t)[(Z_1^{(t)} \backslash (Z'' \cap Z_1^{(t)})]$, which is isomorphic to $\tor(t)[(Z_1^{(t)} \backslash (Z'' \cap Z_1^{(t)})]$ because $Z_1^{(t)} \cap \sigma(t) = \emptyset$.
So it suffices to show that the two endpoints of every edge in $\tor(t)[(Z_1^{(t)} \backslash (Z'' \cap Z_1^{(t)})]$ have the same image under $\pi$.
In other words, we can assume $vv'$ is an edge of $\tor(t)[(Z_1^{(t)} \backslash (Z'' \cap Z_1^{(t)})]$ and thus an edge of $\tor(t)[Z_1^{(t)}]$.
By Observation~\ref{obs-connectedbottom}, $v$ and $v'$ are in the same connected component of $G[(\bigcup_{s \in S} (\gamma(s) \backslash \sigma(s)) \cup \{v,v'\}) \cap Z_1]$, where $S$ consists of all children $s$ of $t$ satisfying $v,v' \in \sigma(s)$.
Since $v,v' \notin Z''$, for any child $s$ of $t$ satisfying $v,v' \in \sigma(s)$, we have $s \notin S_\mathsf{bad}$.
Thus, $S \cap S_\mathsf{bad} = \emptyset$ and $(\gamma(s) \backslash \sigma(s)) \cap Z' = \emptyset$ for all $s \in S$.
It follows that $(\bigcup_{s \in S} (\gamma(s) \backslash \sigma(s)) \cup \{v,v'\}) \cap Z_1 \subseteq Z_1 \backslash Z'$, which implies $v$ and $v'$ are in the same connected component of $G[Z_1 \backslash Z']$, and hence $\pi(v) = \pi(v')$.

We have seen that for all $v,v' \in \beta(t)$, $\theta(v) = \theta(v')$ implies $\pi(v) = \pi(v')$.
Therefore, there exists a unique map $\rho\colon V(G^*) \to \beta^*(t)$ satisfying $\pi_{|\beta(t)} = \rho \circ \theta$, which is surjective.
Note that $\beta^*(t) = V(\tor^*(t))$.
By Fact~\ref{fact-minormap}, to see $\tor^*(t)$ is a minor of $G^*$, it suffices to show that $G^*[\rho^{-1}(V)]$ is connected for all $V \subseteq \beta^*(t)$ such that $\tor^*(t)[V]$ is connected.
Consider a subset $V \subseteq \beta^*(t)$ such that $\tor^*(t)[V]$ is connected.
We have $\rho^{-1}(V) = \theta(\pi_{|\beta(t)}^{-1}(V))$.
By Observation~\ref{obs-TORconnected}, $\mathsf{TOR}(t)[\pi_{|\beta(t)}^{-1}(V)]$ is connected.
Furthermore, by Fact~\ref{fact-quotient}, $G^*[\rho^{-1}(V)]$ is connected iff $\mathsf{TOR}(t)[\pi_{|\beta(t)}^{-1}(V)]$ is connected.
Note that $\pi_{|\beta(t)}^{-1}(V) = \theta^{-1}(\rho^{-1}(V))$.
Therefore, $G^*[\rho^{-1}(V)]$ is connected.
Finally, since $\tor^*(t)$ is a minor of $G^*$, we have $\tw(\tor^*(t)) \leq \tw(G^*) = O_h(p+|Z'|)$.
\end{proof}

The above observation implies $\tw(G/(Z_1 \backslash Z')) = O_h(p+|Z'|)$ by Lemma~\ref{lem-torsotw}, because $(T,\beta^*)$ is a tree decomposition of $G/(Z_1 \backslash Z')$.
This completes the proof of Theorem~\ref{thm-contraction}.

\section{Applications}
\label{sec-app}

In this section, we briefy discuss some applications of Theorem~\ref{thm-contraction} in parameterized complexity building on \cite{BandyapadhyayLLSJ22,MarxMNT22}.
Indeed, Theorem~\ref{thm-contraction} directly results in parameterized algorithms of running time $n^{O(\sqrt{k})}$ or $2^{\widetilde{O}(\sqrt{k})} \cdot n^{O(1)}$ for a broad class of vertex/edge deletion problems on $H$-minor-free graphs.
For example, this includes all problems that can be formulated as {\sc Permutation CSP Deletion} or {\sc 2-Conn Permutation CSP Deletion} \cite{MarxMNT22}.

An instance of the (binary) CSP problem is specified by a triple $\varGamma = (X,D,\mathcal{C})$ where $X$ is a finite set of variables, $D$ is a finite domain, and $\mathcal{C}$ is a set of constraints each of which is of the form $c = (x,y,R)$ where $x,y \in X$ and $R \subseteq D^2$.
We say $\varGamma$ is a \emph{permutation CSP instance} if for every constraint $c = (x,y,R) \in \mathcal{C}$ it holds that $|\{b \in D \mid (a,b) \in R\}| \leq 1$ and $|\{b \in D \mid (b,a) \in R\}| \leq 1$ for all $a \in D$.
We write $|\varGamma| = |X| + |D|$ as the \emph{size} of $\varGamma$.
An \emph{assignment} of $\varGamma$ is a function $\alpha: Y \rightarrow D$ on a subset $Y \subseteq X$, and we say $\alpha$ is \emph{satisfying} if for all $c = (x,y,R) \in \mathcal{C}$ such that $x,y \in Y$, we have $(\alpha(x),\alpha(y)) \in R$.
A permutation CSP instance $\varGamma = (X,D,\mathcal{C})$ naturally induces a graph $G_\varGamma$ with vertex set $X$ in which two variables $x,y \in X$ are connected by an edge if there exists $c \in \mathcal{C}$ such that $c = (x,y,R)$ for some $R \subseteq D^2$.
We say the permutation CSP instance $\varGamma$ is $H$-\emph{minor-free} if its underlying graph $G_\varGamma$ is $H$-minor-free.
For a subset $\mathcal{C}' \subseteq \mathcal{C}$ of constraints, we write $\varGamma - \mathcal{C}' = (X,D,\mathcal{C} \backslash \mathcal{C}')$.
For a subset $X' \subseteq X$ of variables, we write $\varGamma - X' = (X \backslash X',D,\mathcal{C} \backslash \mathcal{C}_{X'})$, where $\mathcal{C}_{X'} \subseteq \mathcal{C}$ consists of all constraints that involve at least one variable from $X'$.

Let $\varGamma = (X,D,\mathcal{C})$ be a permutation CSP instance.
A \emph{size constraint} for $\varGamma$ is a pair $(w,\delta)$ where $w\colon X \times D \to \mathbb{N}$ is a weight function and $\delta \in \mathbb{N}$ is a threshold.
We say an assignment $\alpha\colon Y \to D$ of $\varGamma$ \emph{respects} the size constraint $(w,\delta)$ if $\sum_{y \in Y} w(y,\alpha(y)) \leq \delta$.
We say $\varGamma$ is \emph{satisfiable} subject to $(w,\delta)$ on a subset $Y \subseteq X$ if there is a satisfying assignment $\alpha\colon Y \to D$ of $\varGamma$ that respects $(w,\delta)$.

\subsection{Permutation CSP Deletion Problems}

In \cite{MarxMNT22} a subset of the authors define the problems {\sc Permutation CSP Vertex Deletion} and {\sc Permutation CSP Edge Deletion}.
These problems essentially ask whether one can remove from a given permutation CSP instance a small number of variables (resp., constraints), or equivalently vertices (resp., edges) in the underlying graph, such that the resulting instance is satisfiable on every connected component of the underlying graph.
The formal definitions are the following\footnote{We remark that our definition of (2-connected) permutation CSP deletion is a simplified version of the original definition in \cite{MarxMNT22}, which is already sufficient for our applications.
The original definition is more complicated and allows multiple size constraints of a more general form.}.

The {\sc Permutation CSP Vertex Deletion} problem is a parameterized problem that takes as input a tuple $(\varGamma,w,\delta,k)$, where $\varGamma = (X,D,\mathcal{C})$ is a permutation CSP instance, $(w,\delta)$ is a size constraint for $\varGamma$, and $k$ is the solution size parameter.
The goal is to find a subset $X' \subseteq X$ with $|X'| \leq k$ such that $\varGamma - X'$ is satisfiable subject to $(w,\delta)$\footnote{Strictly speaking, $(w,\delta)$ is a size constraint for $\varGamma$ instead of $\varGamma - X'$. But it naturally induces a size constraint for $\varGamma - X'$ by restricting $w$ to $(X \backslash X') \times D$. For convenience, we still denote it by $(w,\delta)$ here.} on every subset $Y \subseteq X \backslash X'$ satisfying that $G_{\varGamma - X'}[Y]$ is a connected component of $G_{\varGamma - X'}$.

Similarly, the {\sc Permutation CSP Edge Deletion} problem takes the same input, but the goal is to find a subset $\mathcal{C}' \subseteq \mathcal{C}$ with $|\mathcal{C}'| \leq k$ such that $\varGamma - \mathcal{C}'$ is satisfiable subject to $(w,\delta)$ on every subset $Y \subseteq X$ satisfying that $G_{\varGamma - \mathcal{C}'}[Y]$ is a connected component of $G_{\varGamma - \mathcal{C}'}$.

\begin{theorem}\label{thm-permCSP}
 Let $H$ be a fixed graph.
 Then the {\sc Permutation CSP Vertex Deletion} (or the {\sc Permutation CSP Edge Deletion}) problem on $H$-minor-free instances $(\varGamma,w,\delta,k)$ can be solved in $(|\varGamma|+\delta)^{O(\sqrt{k})}$ time.
\end{theorem}

\begin{proof}
 It is shown in \cite{MarxMNT22} that if Theorem~\ref{thm-contraction} (which is stated as a conjecture in \cite{MarxMNT22}) is true, then {\sc Permutation CSP Vertex/Edge Deletion} on $H$-minor-free instances can be solved in $(|\varGamma|+\delta)^{O(\sqrt{k})}$ time, which completes the proof.
 In the following, we briefly sketch the details  (the same ideas are also used in \cite{BandyapadhyayLLSJ22}).

 We focus on the vertex-deletion version.
 Let $(\varGamma,w,\delta,k)$ be the input instance and let $G = G_\varGamma$ be the underlying graph of $\varGamma$ which is $H$-minor-free.
 Using Theorem~\ref{thm-contraction}, we can compute the partition $Z_1,\dots,Z_p \subseteq V(G)$ for $p = \lceil\sqrt{k}\rceil$.
 Suppose $S \subseteq V(G)$ is an unknown solution for $(\varGamma,w,\delta,k)$.
 Since $|S| \leq k$ there exists some $i \in [p]$ such that $|Z_i \cap S| \leq \sqrt{k}$.
 We guess the index $i$ and the vertices in $Z_i \cap S$.
 Note that the number of guesses is bounded by $k \cdot |\varGamma|^{\sqrt{k}} = |\varGamma|^{O(\sqrt{k})}$.
 Now suppose we already know $i$ and $Z_i \cap S$.
 Thus, we know that there exists a feasible solution that is disjoint from $Z_i \backslash (Z_i \cap S)$.
 We can find such a solution by applying the standard dynamic programming approach on a tree decomposition of $G/(Z_i \backslash (Z_i \cap S))$ with running time $(|D|+\delta)^{O(w)} \cdot |\varGamma|^{O(1)}$ where $w$ is the width of the tree decomposition and $D$ is the domain of $\varGamma$.
 By Theorem~\ref{thm-contraction}, $\tw(G/(Z_i \backslash (Z_i \cap S))) = O(\sqrt{k})$ and the desired $(|\varGamma|+\delta)^{O(\sqrt{k})}$-time follows.

 The key insight for the DP on the tree decomposition is the following.
 Since $Z_i \backslash (Z_i \cap S)$ is disjoint from the solution, these vertices are ``undeletable'' variables of $\varGamma$.
 Since we restrict our attention to \emph{permutation} CSP, each connected component of $G[Z_i \backslash (Z_i \cap S)]$ can have only $|D|$ satisfying assignments (since the value of one variable already determines the values of the others).
 This essentially allows us to treat each connected component of $G[Z_i \backslash (Z_i \cap S)]$ as a single vertex (or variable).
 Equivalently, we can contract the part $Z_i \backslash (Z_i \cap S)$ in $G$ and do DP on a tree decomposition of $G/(Z_i \backslash (Z_i \cap S))$.
\end{proof}

Many natural problems can be formulated as {\sc Permutation CSP Vertex/Edge Deletion} problems including the following examples (see \cite{MarxMNT22} for more details):
\begin{itemize}
 \item {\sc Odd Cycle Transversal}: The input consists of a graph $G$ and a parameter $k$. The goal is to find a subset $S \subseteq V(G)$ of at most $k$ vertices such that $G - S$ contains no odd cycle.
 \item {\sc Vertex Multiway Cut}: The input consists of a graph $G$, a terminal set $T \subseteq V(G)$, and a parameter $k$. The goal is to find a subset $S \subseteq V(G)$ of at most $k$ vertices such that no two vertices in $T$ lie in the same connected component of $G - S$.
 \item {\sc Group Feedback Vertex Set}: The input consists of a graph $G$ with a map $\lambda:V(G) \times V(G) \rightarrow \varGamma$ where $\varGamma$ is a group and $\lambda(u,v) = (\lambda(v,u))^{-1}$ and a parameter $k$. The goal is to find a subset $S \subseteq V(G)$ of at most $k$ vertices such that $G - S$ contains no \emph{non-null} cycle where a cycle $(v_0,v_1,\dots,v_r=v_0)$ is non-null if $\prod_{i=1}^r \lambda(v_{i-1},v_i) = \mathbf{1}$.
 \item {\sc Component Order Connectivity}: The input consists of a graph $G$, a threshold $\delta$, and a parameter $k$. The goal is to find a subset $S \subseteq V(G)$ of at most $k$ vertices such that every component of $G-S$ is of size at most $\delta$.
 \item The edge-deletion version of the above problems where the goal is to find a subset $S \subseteq E(G)$ of at most $k$ \emph{edges} such that $G-S$ satisfies the corresponding properties.
\end{itemize}

\begin{corollary}
 There exist algorithms with running time $n^{O(\sqrt{k})}$ for {\sc Odd Cycle Transversal}, {\sc Vertex Multiway Cut}, {\sc Vertex Multicut}, {\sc Group Feedback Vertex Set}, {\sc Component Order Connectivity}, and the edge-deletion version of these problems on $H$-minor-free graphs, where $n$ is the number of vertices of the input graph and $k$ is the solution-size parameter.
\end{corollary}

\begin{proof}
 All problems except {\sc Vertex Multicut} (and {\sc Edge Multicut}) can be formulated as {\sc Permutation CSP Vertex/Edge Deletion} problems, so the $n^{O(\sqrt{k})}$-time algorithms directly follow from Theorem~\ref{thm-permCSP}.
 The problem {\sc Vertex Multicut} (and {\sc Edge Multicut}) can be solved in exactly the same way as described in the proof of Theorem~\ref{thm-permCSP} (see also \cite[Section~5.3]{BandyapadhyayLLSJ22}).
\end{proof}

Using the minor-preserving (quasi-)polynomial kernels from \cite{MarxMNT22} or the (quasi-)polynomial-size candidate sets given in \cite{BandyapadhyayLLSJ22}, we can also obtain (randomized) subexponential-time FPT algorithms for these problems.
Here the notation $\widetilde{O}(\cdot)$ hides $\log k$ factors.

\begin{corollary}
 There exist randomized algorithms with running time $2^{\widetilde{O}(\sqrt{k})} \cdot n^{O(1)}$ for {\sc Odd Cycle Transversal}, {\sc Vertex Multiway Cut}, {\sc Group Feedback Vertex Set} (for a fixed group), and the edge-deletion version of these problems on $H$-minor-free graphs, where $n$ is the number of vertices of the input graph and $k$ is the solution-size parameter.
\end{corollary}

\subsection{Two-Connected Permutation CSP Deletion Problems}

To extend the scope of problems covered by this approach, \cite{MarxMNT22} also considers the {\sc 2-Conn Permutation CSP Deletion} problem.
The problem is similar to {\sc Permutation CSP Deletion} defined above, but the satisfiability we care about is on the 2-connected components of the underlying graph (after deletion), instead of connected components.

The {\sc 2-Conn Permutation CSP Vertex Deletion} problem is a parameterized problem that takes as input a tuple $(\varGamma,w,\delta,k)$, where $\varGamma = (X,D,\mathcal{C})$ is a permutation CSP instance, $(w,\delta)$ is a size constraint for $\varGamma$, and $k$ is the solution size parameter.
The goal is to find a subset $X' \subseteq X$ with $|X'| \leq k$ such that $\varGamma - X'$ is satisfiable subject to $(w,\delta)$\footnote{Strictly speaking, $(w,\delta)$ is a size constraint for $\varGamma$ instead of $\varGamma - X'$. But it naturally induces a size constraint for $\varGamma - X'$ by restricting $w$ to $(X \backslash X') \times D$. For convenience, we still denote it by $(w,\delta)$ here.} on every subset $Y \subseteq X \backslash X'$ satisfying that $G_{\varGamma - X'}[Y]$ is a 2-connected component of $G_{\varGamma - X'}$.

Similarly, the {\sc 2-Conn Permutation CSP Edge Deletion} problem takes the same input, but the goal is to find a subset $\mathcal{C}' \subseteq \mathcal{C}$ with $|\mathcal{C}'| \leq k$ such that $\varGamma - \mathcal{C}'$ is satisfiable subject to $(w,\delta)$ on every subset $Y \subseteq X$ satisfying that $G_{\varGamma - \mathcal{C}'}[Y]$ is a 2-connected component of $G_{\varGamma - \mathcal{C}'}$.

\begin{theorem}\label{thm-2conn}
 Let $H$ be a fixed graph.
 Then the {\sc 2-Conn Permutation CSP Vertex Deletion} (or the {\sc 2-Conn Permutation CSP Edge Deletion}) problem on $H$-minor-free instances $(\varGamma,w,\delta,k)$ can be solved in $(|\varGamma|+\delta)^{O(\sqrt{k})}$ time.
\end{theorem}

\begin{proof}
 Again, it is shown in \cite{MarxMNT22} that if Theorem~\ref{thm-contraction} (which is stated as a conjecture in \cite{MarxMNT22}) is true, then {\sc 2-Conn Permutation CSP Vertex/Edge Deletion} on $H$-minor-free instances can be solved in $(|\varGamma|+\delta)^{O(\sqrt{k})}$ time, which completes the proof.
 Unlike Theorem~\ref{thm-permCSP}, the algorithm for {\sc 2-Conn Permutation CSP Vertex/Edge Deletion} is rather complicated.
 Roughly speaking, it is designed by first applying Theorem~\ref{thm-contraction} to obtain a so-called ``body guessing'' lemma and then using the body guessing lemma to (essentially) guess the 2-connected components in the solution (together with a DP on the tree decomposition).
 Discussing the technical details of this algorithm is out of the scope of this paper, and we refer the interested reader to \cite{MarxMNT22}.
\end{proof}

The following problems can be formulated as 2-connected permutation CSP deletion problems (see \cite{MarxMNT22} for more details):
\begin{itemize}
 \item {\sc Subset Feedback Vertex Set}: The input consists of a graph $G$, a terminal set $T \subseteq V(G)$, and a parameter $k$.
 The goal is to find a subset $S \subseteq V(G)$ of at most $k$ vertices such that every cycle in $G - S$ is disjoint from $T$.
 \item {\sc Subset Odd Cycle Transversal}: The input consists of a graph $G$, a terminal set $T \subseteq V(G)$, and a parameter $k$. The goal is to find a subset $S \subseteq V(G)$ of at most $k$ vertices such that every odd cycle in $G - S$ is disjoint from $T$.
 \item {\sc Subset Group Feedback Vertex Set}: The input consists of a graph $G$ with a map $\lambda\colon V(G) \times V(G) \rightarrow \varGamma$ where $\varGamma$ is a group and $\lambda(u,v) = (\lambda(v,u))^{-1}$, a terminal set $T \subseteq V(G)$, and a parameter $k$. The goal is to find a subset $S \subseteq V(G)$ of at most $k$ vertices such that every non-null cycle in $G - S$ is disjoint from $T$.
 \item {\sc 2-Conn Component Order Connectivity}:  The input consists of a graph $G$, a threshold $\delta$, and a parameter $k$. The goal is to find a subset $S \subseteq V(G)$ of at most $k$ vertices such that every every 2-connected component of $G-S$ is of size at most $\delta$.
 \item The edge-deletion version of the above problems where the goal is to find a subset $S \subseteq E(G)$ of at most $k$ \emph{edges} such that $G-S$ satisfies the corresponding properties.
\end{itemize}

\begin{corollary}
 There exist algorithms with running time $n^{O(\sqrt{k})}$ for {\sc Subset Feedback Vertex Set}, {\sc Subset Odd Cycle Transversal}, {\sc Subset Group Feedback Vertex Set}, {\sc 2-Conn Component Order Connectivity}, and the edge-deletion version of these problems on $H$-minor-free graphs, where $n$ is the number of vertices of the input graph and $k$ is the solution-size parameter.
\end{corollary}

Using the minor-preserving (quasi-)polynomial kernels for {\sc Subset Feedback Vertex Set} and a reduction from {\sc Subset Feedback Edge Set} to  {\sc Subset Feedback Vertex Set} given in \cite{MarxMNT22}, we can also obtain (randomized) subexponential-time FPT algorithms for these two problems.

\begin{corollary}
 There exist randomized algorithms with running time $2^{\widetilde{O}(\sqrt{k})} \cdot n^{O(1)}$ for {\sc Subset Feedback Vertex Set} and {\sc Subset Feedback Edge Set} on $H$-minor-free graphs, where $n$ is the number of vertices of the input graph and $k$ is the solution-size parameter.
\end{corollary}

\bibliographystyle{plainurl}
\bibliography{refs}

\end{document}